\documentclass{article}

\usepackage{cite}

\usepackage{amsmath, amsthm, amssymb}
\usepackage{algorithmic}

\theoremstyle{definition}
\newtheorem{example}{Example}[]
\newtheorem{definition}{Definition}[]
\newtheorem{theorem}{Theorem}[]
\newtheorem{lemma}{Lemma}[]

\usepackage{graphicx}
\usepackage{quantikz}
\usepackage{textcomp}
\usepackage[table]{xcolor}
\definecolor{darkred}{HTML}{a61c00}
\definecolor{mediumblue}{HTML}{1155cc}
\definecolor{cornflowerblue}{HTML}{4a86e8}

\usepackage[justification=centering]{caption}
\usepackage{subcaption}
\usepackage{subfig}
\usepackage{float}

\usepackage{authblk}

\usepackage[a4paper,
top=1.5cm,
bottom=1.5cm,
left=2.5cm,
right=2.5cm,
marginparwidth=1.75cm
]{geometry}

\begin{document}

% Title
% \title{Cost-Effective Multi-Output Quantum Reversible Circuits Based on Multi-Valued Input Reed-Muller Forms}

% \title{Minimization of a New Structure of Binary Quantum Circuits Using Multi-Valued Logic}

% \title{A Minimization Methodology Using Multi-Valued Logic for a New Structure for Quantum Circuits}

% \title{New Generalized ESOP with Decoders for Quantum Circuits}

\title{Minimization of AND-XOR Expressions with Decoders for Quantum Circuits}

% Authors
\author[1]{Sonia Yang\footnote{sonia.liu.yang@gmail.com}}
\author[1]{Ali Al-Bayaty\footnote{albayaty@pdx.edu}}
\author[1]{Marek Perkowski\footnote{h8mp@pdx.edu}}
\affil[1]{\textit{Department of Electrical and Computer Engineering, Portland State University, USA}}

\date{}

\maketitle

\begin{abstract}
    This paper introduces a new logic structure for reversible quantum circuit synthesis. Our synthesis method aims to minimize the quantum cost of reversible quantum circuits with decoders. In this method, multi-valued input, binary output (MVI) functions are utilized as a mathematical concept only, but the circuits are binary. We introduce the new concept of ``Multi-Valued Input Fixed Polarity Reed-Muller (MVI-RM)" forms. Our decoder-based circuit uses three logical levels in contrast to commonly-used methods based on Exclusive-or Sum of Products (ESOP) with two levels (AND-XOR expressions), realized by Toffoli gates. In general, the high number of input qubits in the resulting Toffoli gates is a problem that greatly impacts the quantum cost. Using decoders decreases the number of input qubits in these Toffoli gates. We present two practical algorithms for three-level circuit synthesis by finding the MVI-FPRM: products-matching and the newly developed butterfly diagrams. The best MVI-FPRM forms are factorized and reduced to approximate Multi-Valued Input Generalized Reed-Muller (MVI-GRM) forms.
\end{abstract}

\section{Introduction}
\label{introduction}
Researchers in the area of quantum circuits and quantum compilers \cite{qiskit, schmitt19} seek efficient methods to synthesize multi-output quantum reversible circuits \cite{mishchenko01, fazel07, schmitt19, meuli18, meuli19, rice09, rice11, nayeem11, maslov04, jiang22, saeedi13}. Several classical methods can be adapted for quantum circuits. Some methods \cite{drechsler95, sarabi92} are based on binary Fixed-Polarity Reed-Muller (FPRM) forms. For a function of $n$ variables, there are $2^n$ such canonical forms, and the synthesized quantum circuits corresponding to these forms can differ significantly in their quantum costs \cite{maslov03, al-bayaty23gala, al-bayaty24cala}. Such synthesis methods have also been extended to binary Generalized Reed-Muller (GRM) forms \cite{debnath96, dill97, dill01, helliwell88} to minimize the final quantum costs by minimizing the number of inputs to the Toffoli gates. Most of these methods attempt to minimize the number of many-input Toffoli gates. The idea of using decoders for classical Sum of Products (SOP) logic was invented by Tsutomu Sasao \cite{sasao84}, and it is applied for the first time here to binary quantum circuits based on Exclusive-or Sum of Products (ESOP), also known as AND-XOR expressions. The FPRM and GRM forms are special cases of ESOP expressions, but GRM forms tend to be less cost-expensive compared with non-GRM ESOP expressions in terms of the number of many-input Toffoli gates. Using decoders minimizes the number of inputs to the Toffoli gates. This paper introduces the concept of a logic with multi-valued inputs, binary outputs for the synthesis of entirely binary quantum reversible circuits \cite{perkowski89, schafer91}. 

This paper is organized as follows. Section 2 presents background information on Reed-Muller (RM) forms, expansions, reversible quantum gates, two definitions of quantum cost, and FPRM butterfly diagrams. Section 3 introduces multi-valued input, binary output (MVI) logic, and extends concepts from binary logic to MVI logic. Section 4 demonstrates how to realize a quantum circuit with decoders based on MVI-FPRM forms. Section 5 introduces a products-matching algorithm to find the MVI-FPRM of a function. Section 6 applies the products-matching algorithm and constructs two possible reversible quantum circuits with decoders for a two-bit adder. Section 7 introduces a butterfly diagram algorithm to find the MVI-FPRM of a function. Section 8 presents circuit synthesis based on MVI-GRM forms by factoring from MVI-FPRM forms. Lastly, Section 9 concludes the paper.

\section{Background}
\label{background}

This section provides background on binary Reed-Muller (RM) forms, Boolean expansions, quantum gates, quantum cost, and butterfly diagrams.

\subsection{Binary Reed-Muller Forms}
\label{rm forms}

There are many binary RM forms \cite{green91, drechsler95, sarabi92, debnath96, dill97, dill01, helliwell88}. In this paper, we focus on \textit{Fixed-Polarity Reed-Muller (FPRM) forms} \cite{drechsler95, sarabi92} and \textit{Generalized Reed-Muller (GRM) forms} \cite{debnath96, dill97, dill01, helliwell88}. This section gives background on the RM forms.

A literal in binary logic is either of positive ($x$) or negative ($\bar{x}$) polarity. The \textit{Positive Polarity Reed-Muller (PPRM) form} is a special case of the FPRM form, where all literals are positive polarity, e.g., $1 \oplus x_1 \oplus x_2 \oplus x_1x_2$ is a PPRM form, but $1 \oplus \bar{x}_1 \oplus x_2 \oplus x_1x_2$ is not because $\bar{x}_1$ is negative polarity. It is formally defined in Definition~\ref{def:pprm}.

\begin{definition}
\label{def:pprm}
    The binary Positive Polarity Reed-Muller (PPRM) form of a single-output function $f(x_1, x_2, \dots, x_n)$ is the function in the form of Eq.~\eqref{eq:pprm}, where all literals are of positive polarity and $a_i \in \{0,1\}$.
    \begin{equation}
    \label{eq:pprm}
        f(x_1,x_2,\dots,x_n) 
        = a_0 \oplus a_1(x_1) \oplus a_2(x_2) \oplus \dots \oplus a_{2^n-1}(x_1x_2\dots x_n)
    \end{equation}
    $\square$
\end{definition}

The PPRM form can be generalized to the FPRM form, where each variable is assigned a polarity so that all literals of a certain variable are of the assigned polarity. For example, a function with polarity 101 ($x_1$ is positive polarity, $x_2$ is negative polarity, and $x_3$ is positive polarity) is $\bar{x}_2 \oplus x_1\bar{x}_2 \oplus x_1x_3$. The FPRM form is defined in Definition~\ref{def:fprm}.

\begin{definition}
\label{def:fprm}
    The binary Fixed Polarity Reed-Muller (FPRM) form of a single-output function $f(x_1,x_2,\dots,x_n)$ with polarity $p_1$,$p_2$,\dots,$p_n$, where $p_i \in \{0,1\}$, is the function in the form of Eq.~\eqref{eq:fprm}, where $a_i \in \{0,1\}$.
    \begin{align}
    \label{eq:fprm}
        \hat{x}_i &= 
        \begin{cases}
            \bar{x}_i & \text{if } p_i=0 \\
            x_i & \text{if } p_i=1
        \end{cases} \nonumber \\
        f(x_1,x_2,\dots,x_n) 
        &= a_0 \oplus a_1(\hat{x}_1) \oplus a_2(\hat{x}_2) \oplus \dots \oplus a_{2^n-1}(\hat{x}_1\hat{x}_2\dots \hat{x}_n)
    \end{align}
    $\square$
\end{definition}

GRM forms are not based on polarity and are instead based on the set of variables in each term. A GRM form is one where no pair of terms has exactly the same subset of variables. An example of a GRM form would be $x_1x_2x_3 \oplus \bar{x}_1x_2 \oplus \bar{x}_2x_3$, but the expression $x_2x_3 \oplus \bar{x}_1x_2 \oplus \bar{x}_2x_3$ would not be a GRM because the terms $x_2x_3$ and $\bar{x}_2x_3$ both have the subset of variables $\{x_2,x_3\}$. A GRM form is formally defined in Definition~\ref{def:grm}.

\begin{definition}
\label{def:grm}
    A binary Generalized Reed-Muller (GRM) form of a single-output function $f(x_1,x_2,\dots,x_n)$ is the function in the form of Eq.~\eqref{eq:grm}, where $\hat{x}_i$ can be either positive polarity ($x_i$) or negative polarity ($\bar{x}_i$), and $a_i \in \{0,1\}$.
    \begin{equation}
    \label{eq:grm}
        f(x_1,x_2,\dots,x_n) 
        = a_0 \oplus a_1(\hat{x}_1) \oplus a_2(\hat{x}_2) \oplus \dots \oplus a_{2^n-1}(\hat{x}_1\hat{x}_2\dots \hat{x}_n)
    \end{equation}
    $\square$
\end{definition}

A function of $n$ variables has one PPRM form, $2^n$ possible FPRM forms, and $2^{n2^{n-1}}$ possible GRM forms. In this paper, we generalize these forms to multi-valued input logic in Section~\ref{mvi-rm} and utilize them for quantum circuit synthesis.

\subsection{Boolean Expansions}
\label{expansions}

This section presents the \textit{Shannon expansion} \cite{shannon}, \textit{positive Davio expansion} \cite{davio}, and \textit{negative Davio expansion} \cite{davio}.

The XOR version of the Shannon expansion is expressed in Eq.~\eqref{eq:shannon}, where the cofactor $f_{\bar{a}}$ is the function $f$ with $a=0$, and the cofactor $f_a$ is the function $f$ with $a=1$.
\begin{equation}
\label{eq:shannon}
    f = \bar{a}f_{\bar{a}} \oplus af_a
\end{equation}

The positive and negative Davio expansions can be derived from the Shannon expansion. The positive Davio expansion, stated in Eq.~\eqref{eq:positive-davio}, can be found by substituting $\bar{a}=1 \oplus a$.
\begin{align}
    \label{eq:positive-davio}
    f 
    &= \bar{a}f_{\bar{a}} \oplus af_a \nonumber \\
    &= (1 \oplus a)f_{\bar{a}} \oplus af_a \nonumber \\
    &= f_{\bar{a}} \oplus af_{\bar{a}} \oplus af_a \nonumber \\
    &= f_{\bar{a}} \oplus a(f_{\bar{a}} \oplus f_a) 
\end{align}
Similarly, the negative Davio expansion, stated in Eq.~\eqref{eq:negative-davio}, can be found by substituting $a = 1 \oplus \bar{a}$.
\begin{align}
\label{eq:negative-davio}
    f 
    &= \bar{a}f_{\bar{a}} \oplus af_a \nonumber \\
    &= \bar{a}f_{\bar{a}} \oplus (1 \oplus \bar{a})f_a \nonumber \\
    &= \bar{a}f_{\bar{a}} \oplus f_a \oplus \bar{a}f_a \nonumber \\
    &= f_a \oplus \bar{a}(f_{\bar{a}} \oplus f_a)
\end{align}

These expansions are generalized to ternary input logic in Section~\ref{ternary-expansions}.

\subsection{Reversible Quantum Gates}
\label{q-gates}

This section presents the gates that we use to synthesize circuits from Exclusive-or Sum of Products (ESOP) expressions or RM forms: the quantum \textit{NOT gate}, \textit{Controlled-NOT (CNOT) gate}, and \textit{$n$-bit Toffoli gate}, which are shown in Fig.~\ref{fig:q-not-gate}, Fig.~\ref{fig:q-cnot-gate}, and Fig.~\ref{fig:q-gen-toffoli-gate}, respectively. The NOT gate acts as the NOT operation, the CNOT gate acts as the XOR operation, and the $n$-bit Toffoli gate acts as both an AND operation and an XOR operation. These gates will appear in circuits throughout this paper. Minimizing the number of inputs to the Toffoli gates in a circuit is a current problem that many researchers face, which heavily impacts the quantum cost, as discussed in the next subsection.

\begin{figure}[!htb]
    \centering
    \begin{subfigure}{0.17\linewidth}
        \centering
        \begin{quantikz}
            \lstick{$x$} &\targ{} & \rstick{$\bar{x}$}
        \end{quantikz}
        \caption{NOT gate.}
        \label{fig:q-not-gate}
    \end{subfigure}
    \begin{subfigure}{0.21\linewidth}
        \centering
        \begin{quantikz}
            \lstick{$x_1$} & \ctrl{1} & \rstick{$x_1$} \\
            \lstick{$x_2$} & \targ{} & \rstick{$x_1 \oplus x_2$}
        \end{quantikz}
        \caption{CNOT gate.}
        \label{fig:q-cnot-gate}
    \end{subfigure}
    \begin{subfigure}{0.26\linewidth}
        \centering
        \begin{quantikz}
            \lstick{$x_1$} & \ctrl{1} & \rstick{$x_1$} \\
            \lstick{$x_2$} & \ctrl{1} & \rstick{$x_2$} \\
            \lstick{$x_3$} & \targ{} & \rstick{$x_1x_2 \oplus x_3$}
        \end{quantikz}
        \caption{3-bit Toffoli gate.}
        \label{fig:q-toffoli-gate}
    \end{subfigure}
    \begin{subfigure}{0.34\linewidth}
        \centering
        \begin{quantikz}[wire types={q,q,n,q,q}]
            \lstick{$x_1$} & \ctrl{1} & \rstick{$x_1$} \\
            \lstick{$x_2$} & \ctrl{1} & \rstick{$x_2$} \\
            &\vdots& \\
            \lstick{$x_{n-1}$} & \ctrl{1} & \rstick{$x_{n-1}$} \\
            \lstick{$x_n$} & \targ{} & \rstick{$x_1x_2\dots x_{n-1} \oplus x_n$}
        \end{quantikz}
        \caption{$n$-bit Toffoli gate.}
        \label{fig:q-gen-toffoli-gate}
    \end{subfigure}
    \caption{Reversible quantum gates.}
    \label{fig:q-gates}
\end{figure}
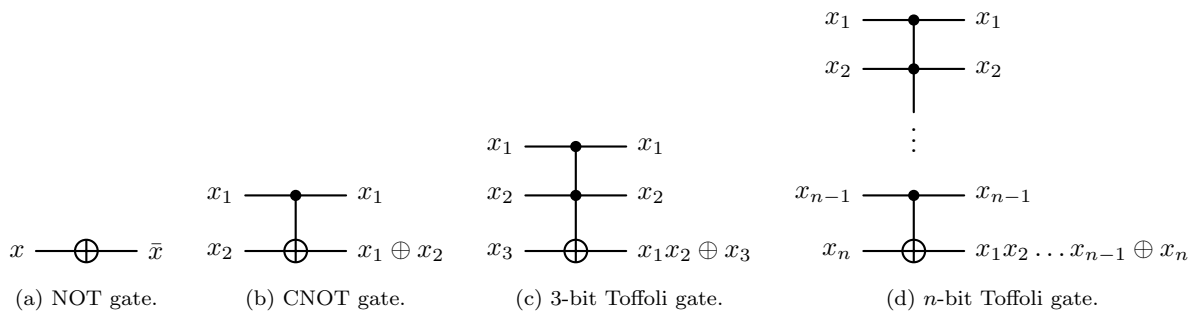

\subsection{Quantum Costs}
\label{cost}
The cost of a circuit can be calculated in many ways, two of which are the Maslov cost \cite{maslov03} and the Transpilation Quantum Cost (TQC) \cite{al-bayaty23gala, al-bayaty24cala}.

The Maslov cost is a metric that sums up the individual cost of all the different Toffoli gates depending on the number of controls, with zero controls being equivalent to a NOT gate, one control being equivalent to a CNOT gate, and two controls being equivalent to a standard 3-bit Toffoli gate, with the Maslov cost for these gates listed in Table~\ref{tab:costs}.

The Transpilation Quantum Cost (TQC) is a metric that calculates the quantum cost of a circuit after transpilation. This cost was originally designed for IBM QPUs, but it can be used for any hardware. The TQC of a circuit is determined by Eq.~\eqref{eq:al-bayaty-cost}, where $N_1$ is the number of native single-qubit gates, $N_2$ is the number of native double-qubit gates, $XC$ is the number of SWAP gates, and $D$ is the depth.
\begin{equation}
    \label{eq:al-bayaty-cost}
    \text{TQC} = N_1 + N_2 + XC + D
\end{equation}

The general quantum costs for different gates are stated in Table~\ref{tab:costs}. The Maslov cost directly calculates the final cost of a function's cost-expensive circuit realization, so it is a technology-independent approach. However, the TQC calculates the cost of the final transpiled circuit and is a technology-dependent approach.  The TQC in this paper is calculated based on the IBM Torino quantum computer \cite{ibm}.

Note that the cost of Toffoli gates with many inputs increases almost exponentially with each input qubit added.

\begin{table}[!htb]
    \centering
    \begin{tabular}{l|c|c}
         Gate & Maslov Cost & TQC \\
         \hline
         NOT & 1 & 1 \\
         CNOT & 1 & 14 \\
         3-Bit Toffoli & 5 & 54 \\
         4-Bit Toffoli & 13 & 109 \\
         5-Bit Toffoli & 29 & 219
    \end{tabular}
    \caption{Quantum cost of different reversible gates.}
    \label{tab:costs}
\end{table}

\subsection{ESOP versus GRM and Factorization}

This subsection presents an example that illustrates the difference between realizing a circuit from a non-GRM ESOP expression, a GRM, and a factorized GRM.

\begin{example}
\label{ex:esop-vs-grm}
   Let $f=x_1x_2x_3 \oplus \bar{x}_1\bar{x}_2\bar{x}_3$, where $x_1$, $x_2$, and $x_3$ are binary variables.

   If the reversible quantum circuit was realized directly from the ESOP expression, then that would lead to the circuit in Fig.~\ref{fig:ex-evgf-esop}, which contains 3 NOT gates and 2 4-bit Toffoli gates. Thus, the circuit has a Maslov cost of 29 and a TQC of 221.

   However, $f$ is equal to the GRM form $x_1x_2 \oplus \bar{x}_1\bar{x}_3 \oplus x_2\bar{x}_3$, as shown on the Karnaugh map in Fig.~\ref{fig:ex-evgf-grm-k-map}. The circuit realization of this GRM is illustrated in Fig.~\ref{fig:ex-evgf-grm}. This circuit consists of 2 NOT gates and 3 3-bit Toffoli gates, and thus has a Maslov cost of 17 and a TQC of 164, which is less costly compared to the non-GRM ESOP. Even though the GRM form has more terms, the terms each have fewer literals, decreasing the number of inputs to the Toffoli gates.

   The cost of the circuit can be further minimized by factoring the GRM form for $f$ to $x_1x_2 \oplus (\bar{x}_1 \oplus x_2)\bar{x}_3$. The circuit realization, shown in Fig.~\ref{fig:ex-evgf-factorized-grm}, consists of 2 NOT gates, 1 CNOT gate, and 2 3-bit Toffoli gates, and thus has a Maslov cost of 13 and a TQC of 124, making it less costly than both the non-GRM ESOP and the GRM.

   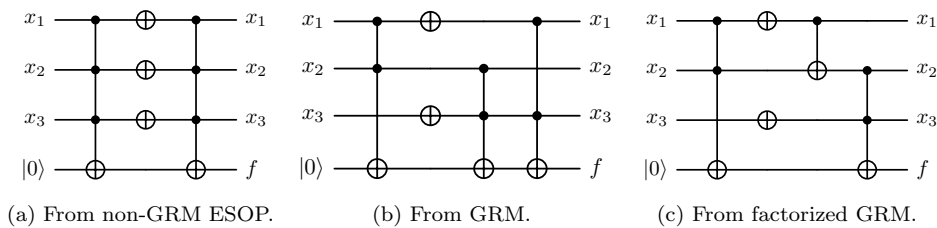
\begin{figure}[!htb]
        \centering
        \begin{subfigure}{0.22\linewidth}
            \centering
            \resizebox{\linewidth}{!}{
            \begin{quantikz}
                \lstick{$x_1$} & \ctrl{1} & \targ{} & \ctrl{1} & \rstick{$x_1$} \\
                \lstick{$x_2$} & \ctrl{1} & \targ{} & \ctrl{1} & \rstick{$x_2$} \\
                \lstick{$x_3$} & \ctrl{1} & \targ{} & \ctrl{1} & \rstick{$x_3$} \\
                \lstick{\ket{0}} & \targ{} & & \targ{} & \rstick{$f$}
            \end{quantikz}}
            \caption{From non-GRM ESOP.}
            \label{fig:ex-evgf-esop}
        \end{subfigure}
        \begin{subfigure}{0.277\linewidth}
            \centering
            \resizebox{\linewidth}{!}{
            \begin{quantikz}
                \lstick{$x_1$} & \ctrl{1} & \targ{} & &  \ctrl{2} & \rstick{$x_1$} \\
                \lstick{$x_2$} & \ctrl{2} & & \ctrl{1} & & \rstick{$x_2$} \\
                \lstick{$x_3$} & & \targ{} & \ctrl{1} & \ctrl{1} & \rstick{$x_3$} \\
                \lstick{\ket{0}} & \targ{} & & \targ{} & \targ{} & \rstick{$f$}
            \end{quantikz}}
            \caption{From GRM.}
            \label{fig:ex-evgf-grm}
        \end{subfigure}
        \begin{subfigure}{0.26\linewidth}
            \centering
            \resizebox{\linewidth}{!}{
            \begin{quantikz}
                \lstick{$x_1$} & \ctrl{1} & \targ{} & \ctrl{1} & & \rstick{$x_1$} \\
                \lstick{$x_2$} & \ctrl{2} & & \targ{} & \ctrl{1} & \rstick{$x_2$} \\
                \lstick{$x_3$} & & \targ{} & & \ctrl{1} & \rstick{$x_3$} \\
                \lstick{\ket{0}} & \targ{} & & & \targ{} & \rstick{$f$}
            \end{quantikz}}
            \caption{From factorized GRM.}
            \label{fig:ex-evgf-factorized-grm}
        \end{subfigure}
        \caption{Reversible quantum circuit realizations of $f$ from Example~\ref{ex:esop-vs-grm}.}
        \label{fig:ex-evgf}
    \end{figure}
    \begin{figure}[!htb]
        \centering
        \includegraphics[width=0.175\linewidth]{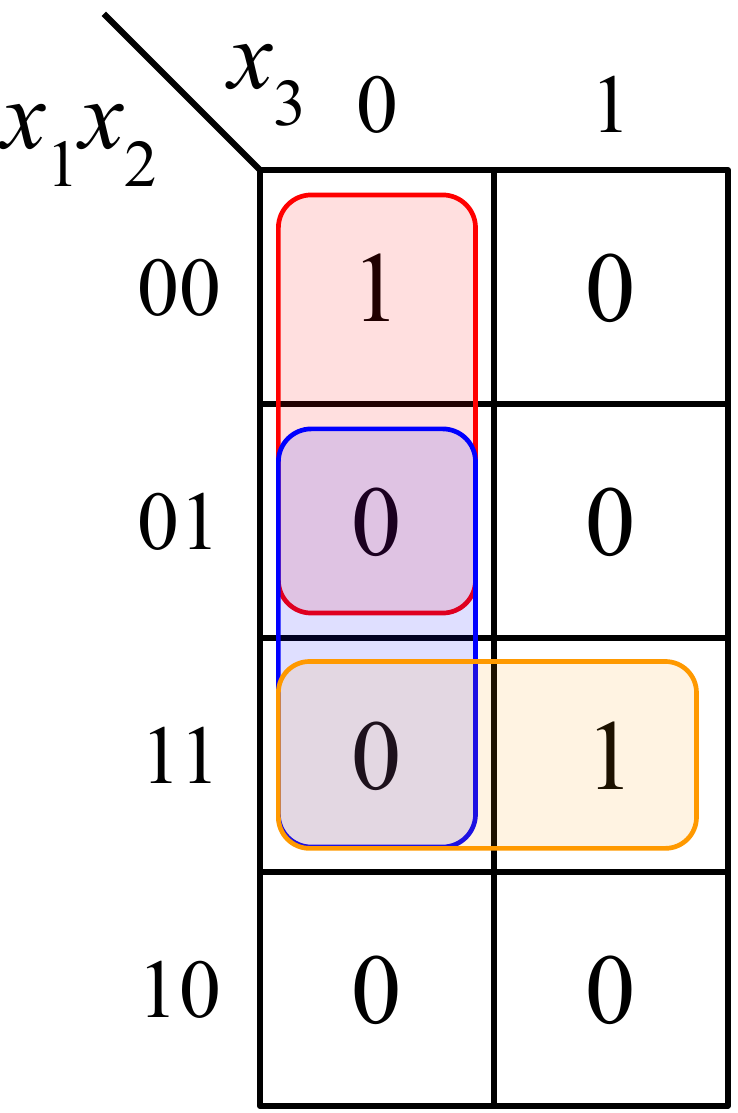}
        \caption{Karnaugh map of $f$ showing the groups for the GRM form: $\textcolor{orange}{x_1x_2} \oplus \textcolor{red}{\bar{x}_1\bar{x}_3} \oplus \textcolor{blue}{x_2\bar{x}_3}$.}
        \label{fig:ex-evgf-grm-k-map}
    \end{figure}

   Aiming for GRM forms and factorizing expressions to reduce the number of literals in a product and thus the number of inputs to the Toffoli gates is a core part of our methodology to minimize circuits.

\end{example}

\subsection{Butterfly Diagrams}
\label{butterfly}

There are many fundamentally different concepts of a ``butterfly diagram" in physics, astronomy \cite{rudiger95}, computer engineering, systems theory, and other fields. The concept of \textit{butterfly diagram} used in engineering is different than those in other areas, and it is useful for mathematical structure, visualization of algorithms, and efficient realization in parallel hardware. Well-known butterfly diagrams are created for the Fourier \cite{weinstein69, shanks69, vahid20}, Haar \cite{falkowski06}, Hadamard \cite{falkowski06}, and other spectral transforms. Less popular are diagrams created for Fixed-Polarity Reed-Muller transforms \cite{li06, jin20}. In this paper, we create new classes of butterfly diagrams that are useful to synthesize ternary input quantum circuits and binary quantum circuits with decoders. 

\begin{figure}[!htb]
    \centering
    \begin{subfigure}{0.25\linewidth}
        \centering
        \includegraphics[width=0.95\linewidth]{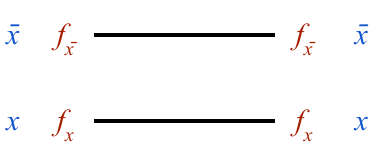}
        \caption{Shannon: $f = \textcolor{mediumblue}{\bar{a}}\textcolor{darkred}{f_{\bar{a}}} \oplus \textcolor{mediumblue}{a}\textcolor{darkred}{f_a}$.}
        \label{fig:butterfly-shannon}
    \end{subfigure}
    \hspace{0.01\linewidth}
    \begin{subfigure}{0.35\linewidth}
        \centering
        \includegraphics[width=0.8\linewidth]{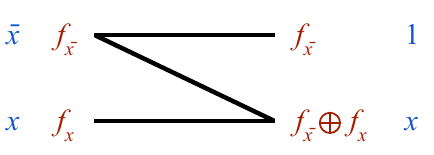}
        \caption{Positive Davio: $f = \textcolor{darkred}{f_{\bar{a}}} \oplus \textcolor{mediumblue}{a}(\textcolor{darkred}{f_{\bar{a}} \oplus f_a})$.}
        \label{fig:butterfly-positive-davio}
    \end{subfigure}
    \hspace{0.01\linewidth}
    \begin{subfigure}{0.35\linewidth}
        \centering
        \includegraphics[width=0.8\linewidth]{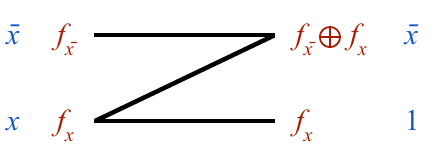}
        \caption{Negative Davio: $f = \textcolor{darkred}{f_a} \oplus \textcolor{mediumblue}{\bar{a}}(\textcolor{darkred}{f_{\bar{a}} \oplus f_a})$.}
        \label{fig:butterfly-negative-davio}
    \end{subfigure}
    \caption{The butterfly diagram kernels for the binary Shannon, positive Davio, and negative Davio expansions.}
    \label{fig:butterfly-expansions}
\end{figure}

Examples of butterfly diagram kernels are shown in Fig.~\ref{fig:butterfly-shannon}, Fig.~\ref{fig:butterfly-positive-davio}, and Fig.~\ref{fig:butterfly-negative-davio}, which correspond to the Shannon, positive Davio, and negative Davio expansions, respectively. In Fig.~\ref{fig:butterfly-expansions}, the \textcolor{darkred}{red} variables represent the inputs, while the \textcolor{mediumblue}{blue} variables represent the values the inputs are multiplied by. The middle part is the butterfly diagram kernel, where the points where the lines meet represent an XOR operation, and both sides correspond to the same function. 

For these butterfly diagrams, the left side represents the Shannon expansion, $f = \textcolor{mediumblue}{\bar{a}}\textcolor{darkred}{f_{\bar{a}}} \oplus \textcolor{mediumblue}{a}\textcolor{darkred}{f_a}$, where the inputs are the cofactors. The right side represents their respective expansions, which are stated in the respective captions.

The positive Davio expansion also represents converting the variable $a$ to positive polarity, and the negative Davio expansion represents converting the variable $a$ to negative polarity. Applying this method for each variable in a function can be used to transform a function from its \textit{minterms} to an FPRM, where the minterms are the terms that contain literals for all variables. 

An example of a butterfly diagram for a two-variable function which converts minterms to an FPRM of polarity 01 ($x_1$ is negative polarity and $x_2$ is positive polarity) is shown in Fig.~\ref{fig:fprm-butterfly-example}. The first layer (column) of the butterfly diagram converts all instances of the variable $x_2$ to positive polarity with the positive Davio butterfly diagram kernels, and the second layer converts all instances of the variable $x_1$ to negative polarity with the negative Davio butterfly diagram kernels. The inputs and outputs for the function $\textcolor{mediumblue}{\bar{x}_1\bar{x}_2 \oplus x_1\bar{x}_2 \oplus x_1x_2}$, which has the FPRM form $\textcolor{mediumblue}{\bar{x}_1x_2 \oplus 1}$ are listed in the figure. A value of \textcolor{darkred}{1} means that the corresponding term is in the function, while a value of \textcolor{darkred}{0} means that it is not. The concept of butterfly diagrams for FPRM forms is expanded to multi-valued input logic in Section~\ref{fprm-butterfly}.

\begin{figure}[!htb]
    \centering
    \includegraphics[width=0.4\linewidth]{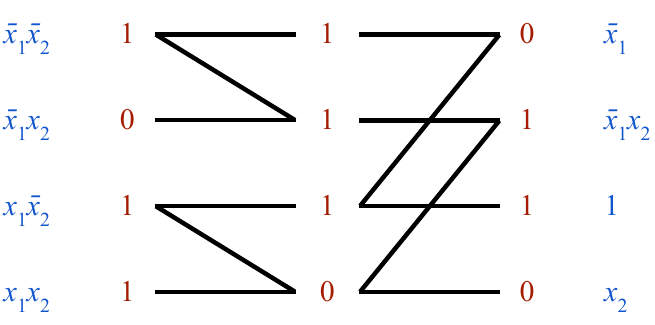}
    \caption{The butterfly diagram that transforms a function of $x_1$ and $x_2$ from the XOR of minterms: $\textcolor{mediumblue}{\bar{x}_1\bar{x}_2 \oplus x_1\bar{x}_2 \oplus x_1x_2}$, to the FPRM form with polarity 01: $\textcolor{mediumblue}{\bar{x}_1x_2 \oplus 1}$.}
    \label{fig:fprm-butterfly-example}
\end{figure}

\section{Multi-Valued Input, Binary Output Functions}
\label{mvi}
In literature~\cite{ilyas22}, the concept of multi-valued input, multi-valued output functions related to the classical Reed-Muller (RM) forms exists. However, this paper introduces a new structural design for binary quantum circuits based on canonical forms, namely MVI-FPRM and MVI-GRM, for \textit{multi-valued input}, \textit{binary output} (MVI) functions. In general \cite{mishchenko01, fazel07, schmitt19, meuli18, meuli19, rice09, rice11, nayeem11}, quantum circuits are synthesized from ESOP \cite{sasao93, csanky93, kazimirov21, papakonstantinou17, riener19, green91, song97} expressions (also known as AND-XOR expressions, Boolean polynomials, or Zhegalkin polynomials) that are formally built with two levels, or similar circuits. In contrast, the circuit structure from this paper has three levels: a decoder level, an AND level, and an XOR level. Because MVI functions use multi-valued variables as a representation of multiple binary variables, MVI functions are realized as binary circuits rather than multi-valued circuits. Our goal is to avoid large, expensive $n$-bit Toffoli gates ($n>3$). Circuits from these MVI functions can then be used to construct blocks such as logical comparators and arithmetic circuits, in the oracles of different quantum search algorithms, such as Grover's Algorithm~\cite{grover96}, BHT-QAOA~\cite{al-bayaty24bht-qaoa}, and the Quantum Walk Algorithm~\cite{alasow24}.

In addition, our defined forms and related optimization methods can be used in the future, when ternary input, binary output technology~\cite{ilyas22} becomes available for quantum reversible circuits.

In this paper, binary variables will be denoted by lowercase letters ($x$) and multi-valued variables will be denoted by uppercase letters ($X$).

\begin{definition}
\label{def:mvi function}
    A multi-valued input, completely specified binary output (MVI) function is a mapping $F(X_1, X_2, \dots, X_n): V_1 \times V_2 \times \dots \times V_n \rightarrow \{0,1\}$, where $X_i$ is a multi-valued variable that takes in a value from the set $V_i = \{0, 1, \dots. v_i-1\}$, where $v_i$ is the radix. $\square$
\end{definition}

It is important to note that MVI functions can take in variables of any \textit{radix}, meaning that the variables could be ternary (3-valued), quaternary (4-valued), quinary (5-valued), or another radix. The variables for an MVI function are also not required to be of the same radix. For example, a function could take in the quaternary variable $X_1$ and the ternary variable $X_2$.

In binary logic, a literal is equal to 1 if $x=1$, and it is positive ($x$) polarity, or if $x=0$, and it is negative ($\bar{x}$) polarity. Extending this to MVI logic leads to the definition of a \textit{multi-valued input (MVI) literal} in Definition~\ref{def:mvi-literal}. 

\begin{definition}
\label{def:mvi-literal}
    A multi-valued input, binary output (MVI) literal of the multi-valued variable $X$ for a given set of truth values $S \subseteq V = \{0, 1, \dots. v-1\}$, denoted by $X^{S}$, is defined as
    $$
    X^{S} = 
    \begin{cases}
    1 & \text{if } X \in S\\
    0  & \text{if } X \notin S.\\
    \end{cases}
    $$
    $\square$
\end{definition}

For instance, the literal $X^{1,2,3}$ is equal to 1 if the variable $X$ is equal to 1, 2, or 3. Otherwise, the literal is equal to 0. The variable $X$, which is inputted, is multi-valued. However, the literal $X^{1,2,3}$, which is outputted, is binary. A product of literals ($X_1^{S_1}X_2^{S_2} \dots X_n^{S_n}$) is a \textit{product term}, which is also called a \textit{product} or \textit{term} for short, and is also a binary value. An MVI literal with a single value (only one element in $S$, e.g., $X^0$ and $X^2$) will be called a \textit{single-valued literal}. One method of realizing an MVI function as a binary quantum circuit is demonstrated in Example~\ref{ex:mvi function}. 

\begin{example} 
    \label{ex:mvi function}
    The goal is to create a circuit for the function $f = X_1^{0,2,3}X_2^{0,1}$, where $X_1$ and $X_2$ are quaternary (4-valued) variables.
    
    \begin{table}[!htb]
        \centering
        \begin{subtable}[!h]{0.3\linewidth}
            \centering
            \resizebox{0.5\linewidth}{!}{
            \begin{tabular}{cc|c}
                $x_{1a}$ & $x_{1b}$ & $X_1$\\
                \hline
                0 & 0 & 0 \\
                0 & 1 & 1 \\
                1 & 0 & 2 \\
                1 & 1 & 3 \\
            \end{tabular}}
            \caption{$X_1$ in terms of the binary variables $x_{1a}$ and $x_{1b}$.}
        \end{subtable}
        \begin{subtable}[!h]{0.3\linewidth}
            \centering
            \resizebox{0.5\linewidth}{!}{
            \begin{tabular}{cc|c}
                $x_{2a}$ & $x_{2b}$ & $X_2$\\
                \hline
                0 & 0 & 0 \\
                0 & 1 & 1 \\
                1 & 0 & 2 \\
                1 & 1 & 3 \\
            \end{tabular}}
            \caption{$X_2$ in terms of the binary variables $x_{2a}$ and $x_{2b}$.}
        \end{subtable}
        \caption{$X_1$ and $X_2$ encoded by binary variables.}
        \label{tab:ex1-x-encoding}
    \end{table}

    To create a binary circuit for this function, we can represent the quaternary variables $X_1$ and $X_2$ with the binary variables $x_{1a}$ and $x_{1b}$ for $X_1$, and $x_{2a}$ and $x_{2b}$ for $X_2$. As shown in Table~\ref{tab:ex1-x-encoding}, the value for 0 can be encoded as 00, 1 as 01, 2 as 10, and 3 as 11. Converting the tables into Karnaugh maps for $X_1$ and $X_2$ in terms of $x_{1a}$, $x_{1b}$, $x_{2a}$, and $x_{2b}$ gives the maps shown in Fig.~\ref{fig:ex1-x-map}. Karnaugh maps for the literals in the function $f$, $X_1^{0,2,3}$ and $X_2^{0,1}$, are also shown in Fig.~\ref{fig:ex1-x-map}.
    
    \begin{figure}[!htb]
        \centering
        \begin{subfigure}[!h]{0.15\linewidth}
            \centering
            \includegraphics[width=\linewidth]{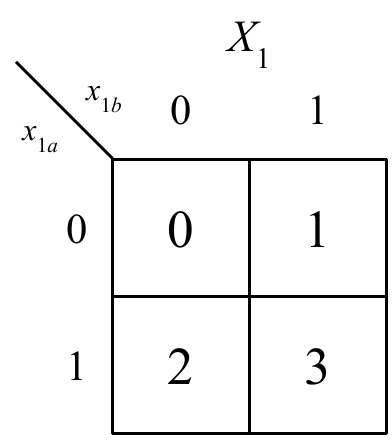}
            \caption{$X_1$.}
        \end{subfigure}
        \begin{subfigure}[!h]{0.15\linewidth}
            \centering
            \includegraphics[width=\linewidth]{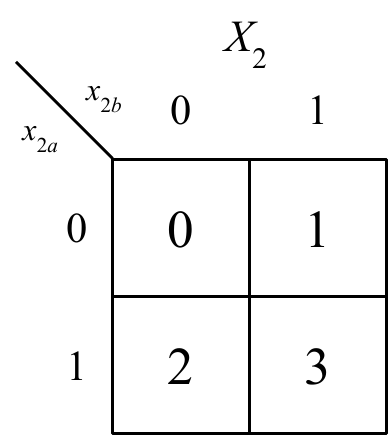}
            \caption{$X_2$.}
        \end{subfigure}
        \centering
        \begin{subfigure}[!h]{0.15\linewidth}
            \centering
            \includegraphics[width=\linewidth]{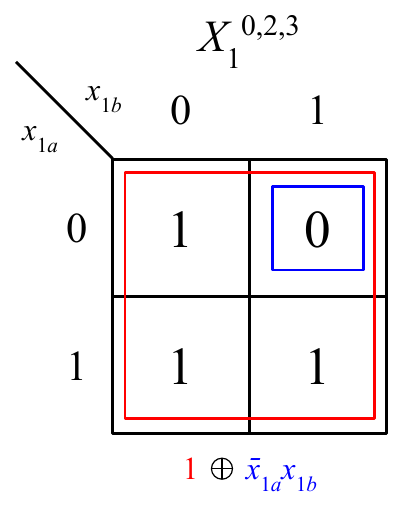}
            \caption{$X_1^{0,2,3}$.}
        \end{subfigure}
        \begin{subfigure}[!h]{0.15\linewidth}
            \centering
            \includegraphics[width=\linewidth]{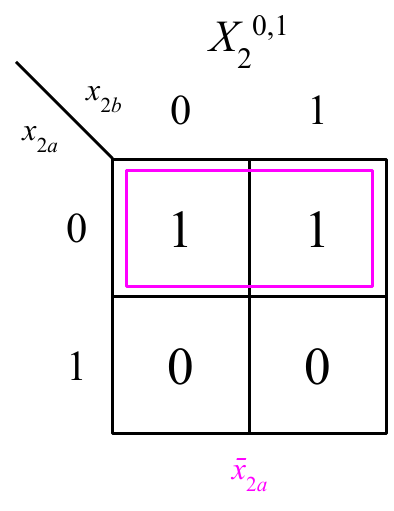}
            \caption{$X_2^{0,1}$.}
        \end{subfigure}
        \caption{Karnaugh maps for $X_1$, $X_2$ and the literals $X_1^{0,2,3}$, $X_2^{0,1}$ encoded by binary variables for Example~\ref{ex:mvi function}.}
        \label{fig:ex1-x-map}
    \end{figure}

    \begin{figure}[!htb]
    \centering
        \begin{subfigure}{0.5\linewidth}
            \centering
            \includegraphics[width=.8\linewidth]{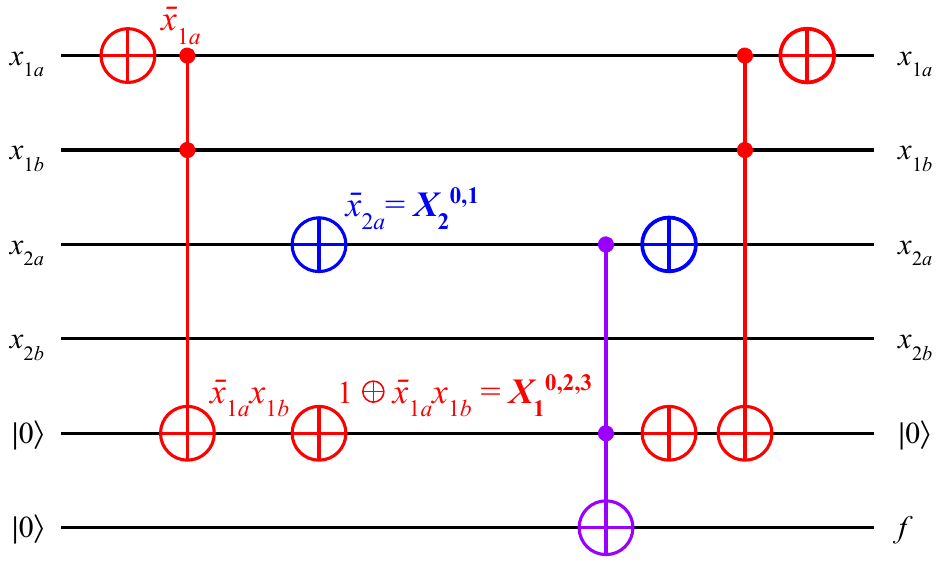}
            \caption{Binary quantum circuit realization of the function $f$ with mirror gates.}
            \label{fig:ex1-circuit}
        \end{subfigure}
        \begin{subfigure}{0.26\linewidth}
            \centering
            \includegraphics[width=\linewidth]{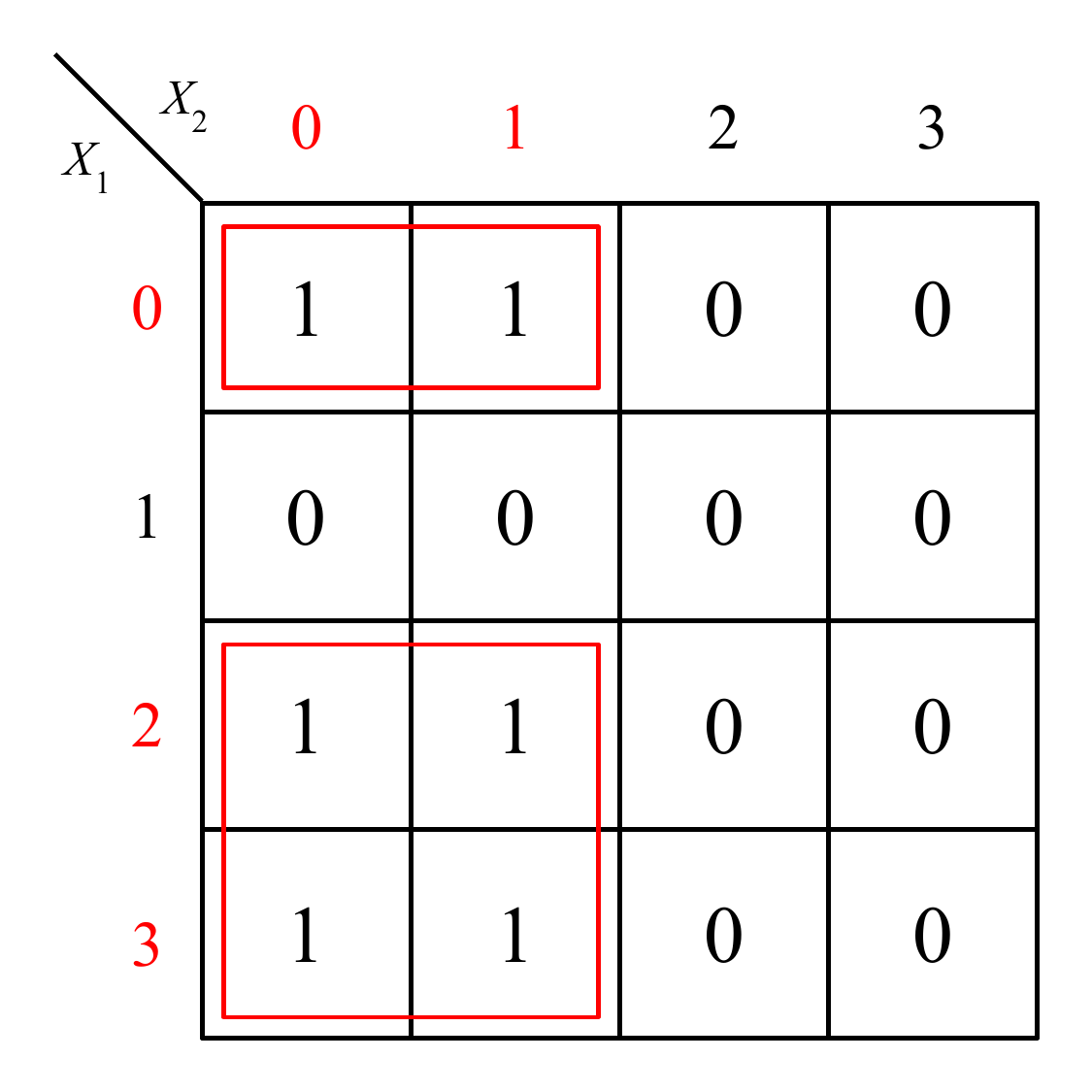}
            \caption{Marquand chart of $f$.}
            \label{fig:ex1-f-chart}
            \end{subfigure}
        \caption{The circuit and Marquand chart of $f=X_1^{0,2,3}X_2^{0,1}$ for Example~\ref{ex:mvi function}.}
    \end{figure}
    
    From the groupings shown in Fig.~\ref{fig:ex1-x-map} for $X_1^{0,2,3}$, it can be found that $X_1^{0,2,3}=1 \oplus \bar{x}_{1a}x_{1b}$ and $X_2^{0,1}=\bar{x}_{2a}$. $f$ can be realized as the circuit shown in Fig.~\ref{fig:ex1-circuit}. This function is also represented in the Marquand chart shown in Fig.~\ref{fig:ex1-f-chart}. The rows represent the possible values of $X_1$ and the columns represent the possible values of $X_2$. The numbers in the cells represent the value of $f$ for the given values of $X_1$ and $X_2$. The values of $X_1$ and $X_2$ that lead to a value of 1 for their respective literals $X_1^{0,2,3}$ and $X_2^{0,1}$ are in \textcolor{red}{red}. Note that the function $f$ is only equal to 1 if both literals are 1, and that the MVI expression naturally factorizes the function. Also note that the values listed for the rows and columns of the Marquand chart are in natural order, and that, in contrast to Karnaugh maps, the groupings on the Marquand chart are not required to be adjacent.

\end{example}

\subsection{Operating with Multi-Valued Input Literals}
\label{mvi operations}

Many laws of Boolean algebra \cite{wakerly99, song96} can be extended to MVI literals. This section covers the basic operations for algebraically manipulating equations dealing with these literals.

Notably, $X^S$ is always 1 when $S=V=\{0,1,\dots,v-1\}$, so $X^V=1$, and $X^S$ is always 0 when $S=\varnothing$, where $\varnothing$ is the empty set, and $X$ is a $v$-valued variable ($v\geq2$). The AND, OR, and XOR operations ($X_1^{S_1}X_2^{S_2}$,  $X_1^{S_1}+X_2^{S_2}$, and  $X_1^{S_1} \oplus X_2^{S_2}$) are identical to Boolean algebra, since the literals are binary. However, when applying the AND, OR, and XOR operations on literals of the same variable, e.g., $X^{1,2} \oplus X^{0,1}$, there are rules to simplify such expressions, as stated in Eq.~\eqref{eq:mvi-and-op}, Eq.~\eqref{eq:mvi-or-op}, and Eq.~\eqref{eq:mvi-xor-op} respectively, where $\cup$ is set union, $\cap$ is set intersection, and $\Delta$ is set symmetric difference. The symmetric difference $A \Delta B$ of two sets $A$ and $B$ is equal to the set of all elements in either $A$ or $B$ but not both, as the set counterpart to the XOR operation.
\begin{align}
    \label{eq:mvi-and-op}
    X^{S_1}X^{S_2} &= X^{S_1 \cap S_2} \\
    \label{eq:mvi-or-op}
    X^{S_1} + X^{S_2} &= X^{S_1 \cup S_2} \\
    \label{eq:mvi-xor-op}
    X^{S_1} \oplus X^{S_2} &= X^{S_1 \Delta S_2} \\
    \label{eq:mvi-not-op}
    \overline{X^{S}} = 1 \oplus X^S &= X^{V \Delta S}
\end{align}

For example, $X^{0,1,3,4}X^{1,2,4} = X^{1,4}$, $X^{0,1,3,4} + X^{1,2,4} = X^{0,1,2,3,4}$, and $X^{0,1,3,4} \oplus X^{1,2,4} = X^{0,2,3}$.

The NOT operation in Boolean ring algebra is $\bar{x} = 1 \oplus x$. The NOT operation for MVI logic is similarly defined in Eq.~\eqref{eq:mvi-not-op}. For example, for the quinary (5-valued) variable $X$, $\overline{X^{0,1,4}} = X^{2,3}$.

In this paper, the AND and XOR operations will be utilized, where the AND operation can be realized with a Toffoli gate, and the XOR operation can be realized with a CNOT or Toffoli gate.

\subsection{Decoders}
\label{decoders}
A decoder converts binary variables into MVI literals. All possible literals should be able to be represented by an XOR expression of the decoder's outputted literals, as shown in Example~\ref{ex:decoder}.

\begin{example}
    \label{ex:decoder}
    Let $X$ be a ternary (3-valued) variable, which is encoded by the binary variables $x_1$ and $x_2$. A possible decoder for $X$ can output the literals $X^{1,2}$, $X^1$, and $X^0$. Every literal can be created by an XOR expression with the outputted literals, as shown in Table~\ref{tab:ex2}. Any literals of $X$ that are needed can then be synthesized using CNOT gates.
    
    \begin{table}[!htb]
        \centering
        \renewcommand{\arraystretch}{1.2}
        \begin{tabular}{c|c|c}
            Literal & Binary Code & In Terms of the Outputted Literals \\
            \hline
            $X^2$ & 001 & $X^{1,2} \oplus X^1$ \\
            $X^1$ & 010 & $X^1$ \\
            $X^{1,2}$ & 011 & $X^{1,2}$ \\
            $X^0$ & 100 & $X^0$ \\
            $X^{0,2}$ & 101 & $X^{1,2} \oplus X^1 \oplus X^0$ \\
            $X^{0,1}$ & 110 & $X^1 \oplus X^0$ \\
            $X^{0,1,2}$ & 111 & $X^{1,2} \oplus X^0$ \\
        \end{tabular}
        \caption{All literals represented by a binary code and XOR expression with $X^{1,2}$, $X^1$, and $X^0$ for Example~\ref{ex:decoder}.}
        \label{tab:ex2}
    \end{table}
    
    This decoder can also be represented as the \textit{polarity matrix} $P$.
    $$
    P=
    \overset{0 \quad 1 \quad2}{
    \begin{bmatrix}
        0 & 1 & 1 \\
        0 & 1 & 0 \\
        1 & 0 & 0
    \end{bmatrix}
    }
    =
    \begin{bmatrix}
    T^1 \\
    T^2 \\
    T^3
    \end{bmatrix}
    $$
    
    The first row represents the literal $X^{1,2}$, since the values in column 1 and column 2 are 1, where the columns are labeled with the first column as 0. The second row represents the literal $X^1$, since the value of column 1 is 1. And the third row represents the literal $X^0$, since the value of column 0 is 1.
    
    The rows, which correspond to each outputted literal, can be represented as the row vectors $T^1$, $T^2$, and $T^3$. Note that these vectors are \textit{linearly independent}, meaning that no vector in the set {$T^1$, $T^2$, $T^3$} can be obtained from linearly operating on the other vectors \cite{roman05}. Here, XOR is the linear operation used. The polarity matrix $P$, or more generally $P_i$, will be elaborated on in the next subsection.
    
    \begin{figure}[htbp]
        \centering
        \begin{subfigure}[t]{0.1\linewidth}
            \centering
            \includegraphics[width=\linewidth]{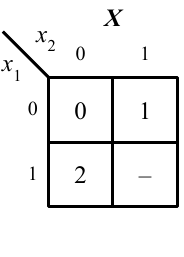}
            \caption{$X$.}
            \label{fig:k-map-x}
        \end{subfigure}
        \begin{subfigure}[t]{0.1\linewidth}
            \centering
            \includegraphics[width=\linewidth]{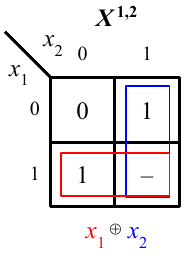}
            \caption{$X^{1,2}$.}
            \label{fig:k-map-x12}
        \end{subfigure}
        \begin{subfigure}[t]{0.1\linewidth}
            \centering
            \includegraphics[width=\linewidth]{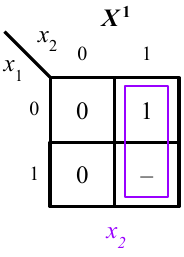}
            \caption{$X^1$.}
            \label{fig:k-map-x1}
        \end{subfigure}
        \begin{subfigure}[t]{0.1\linewidth}
            \centering
            \includegraphics[width=\linewidth]{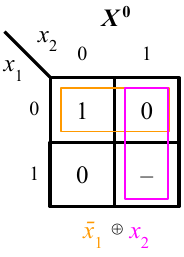}
            \caption{$X^0$.}
            \label{fig:k-map-x0}
        \end{subfigure}
        \begin{subfigure}[t]{0.4\linewidth}
            \centering
            \includegraphics[width=0.8\linewidth]{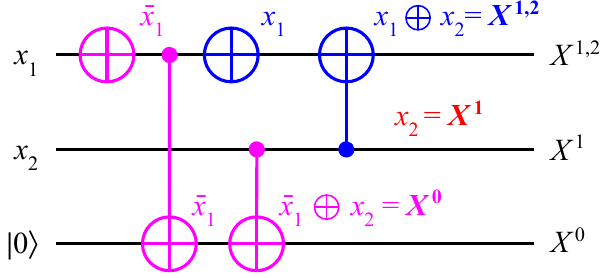}
            \caption{Decoder.}
            \label{fig:decoder-example}
        \end{subfigure}
        \caption{Karnaugh maps and the decoder for the literals $X^{1,2}$, $X^1$, and $X^0$ in Example~\ref{ex:decoder}.}
    \end{figure}
    
    The realization of a binary quantum circuit for this decoder is shown in Fig.~\ref{fig:decoder-example}.
\end{example}

\subsection{Polarity With Multi-Valued Input Literals}
\label{mvi polarity}
In binary logic, the polarity is either negative or positive. However, for $v$-valued logic, the polarity is defined as a matrix since all possible functions can only be defined with at least $v$ different literals. The literals represented by the polarity matrix $P$ from Example~\ref{ex:decoder} correspond to the decoder's outputted literals. The polarity matrix is formally defined in Definition~\ref{def:polarity}.

\begin{definition}
    \label{def:polarity}
    The polarity matrix $P$ (or polarity for short) of a multi-valued variable $X$ is a matrix where the $r$th row vector $T^r$ of the matrix is the binary representation of the polarity literal $P^r$. $\square$
\end{definition}

The concept of polarity for multi-valued variables is formally expanded on in Theorem~\ref{thm:polarity}.

\begin{theorem}
    \label{thm:polarity}
    A multi-valued literal $X^S$, where $S \subseteq V=\{0,1,\dots,v-1\}$, can be represented by $v$ polarity literals $P^1, P^2, \dots, P^v$. The values of the polarity literals form the row vectors $T^r$ of the linearly independent $v \times v$ matrix $P$.
\end{theorem}

\begin{proof}
    A property of a linearly independent $m \times m$ matrix $O$ ($m \in \mathbb{N}$) with elements $o_{ij}\in\{0,1\}$ is that any $m$-dimensional vector $U$ with elements $u_{ij}\in\{0,1\}$ can be represented as a bit-by-bit XOR operation on the row vectors of the matrix $O$. 
    
    Thus, any set of truth values $S \subseteq V=\{0,1,\dots,v-1\}$ of a literal $X^S$ can be represented as an XOR expression of the values $P^r$ with $r \subseteq V=\{0,1,\dots,v-1\}$, where $T^r$ is the $r$th row vector of the linearly independent $v \times v$ matrix $P$.
\end{proof}

Theorem 1 states that any MVI literal $X^S$ can be expressed in terms of the polarity literals of a polarity $P$. Note that a form, according to Theorem 1, is a canonical form, meaning that the form is unique for each function and polarity. This is because the polarity matrix $P$ is linearly independent.

The number of potential polarities $n_P$ (where order does not matter) for an $v$-valued variable is stated in Eq.~\ref{eq:num polarities} \cite{roman05}.

\begin{equation}
    \label{eq:num polarities}
    n_P = \frac{1}{v!}\prod_{k=0}^{v-1}2^v-2^k
\end{equation}

The number of polarities from binary (2-valued) to quinary (5-valued) logic is shown in Table~\ref{tab:num-polarities}.

\begin{table}[!htb]
    \centering
    \begin{tabular}{l|c}
         & Number of Polarities \\
        \hline
        Binary (2-Valued) & 3 \\
        Ternary (3-Valued) & 28 \\
        Quaternary (4-Valued) & 840 \\
        Quinary (5-Valued) & 83,328 \\
    \end{tabular}
    \caption{Number of polarities for different valued logic.}
    \label{tab:num-polarities}
\end{table}

All 28 of the possible polarities for a ternary (3-valued) variable are listed in Table~\ref{tab:ternary-polarities}

\begin{table}[!htb]
    \centering
    \begin{tabular}{ccccccc}
         $\begin{bmatrix}
             1 & 0 & 0 \\
             0 & 1 & 0 \\
             0 & 0 & 1
         \end{bmatrix}$
         & 
         $\begin{bmatrix}
             1 & 0 & 0 \\
             0 & 1 & 0 \\
             1 & 0 & 1
         \end{bmatrix}$
         &
         $\begin{bmatrix}
             1 & 0 & 0 \\
             0 & 1 & 0 \\
             0 & 1 & 1
         \end{bmatrix}$
         & 
         $\begin{bmatrix}
             1 & 0 & 0 \\
             0 & 1 & 0 \\
             1 & 1 & 1
         \end{bmatrix}$
         &
         $\begin{bmatrix}
             1 & 0 & 0 \\
             0 & 0 & 1 \\
             1 & 1 & 0
         \end{bmatrix}$
         & 
         $\begin{bmatrix}
             1 & 0 & 0 \\
             0 & 0 & 1 \\
             0 & 1 & 1
         \end{bmatrix}$
         &
         $\begin{bmatrix}
             1 & 0 & 0 \\
             0 & 0 & 1 \\
             1 & 1 & 1
         \end{bmatrix}$ \\[20px]
         $\begin{bmatrix}
             1 & 0 & 0 \\
             1 & 0 & 0 \\
             1 & 0 & 1
        \end{bmatrix}$
        &
        $\begin{bmatrix}
             1 & 0 & 0 \\
             1 & 1 & 0 \\
             0 & 1 & 1
        \end{bmatrix}$
        &
        $\begin{bmatrix}
             1 & 0 & 0 \\
             1 & 1 & 0 \\
             1 & 1 & 1
        \end{bmatrix}$
        &
        $\begin{bmatrix}
             1 & 0 & 0 \\
             1 & 0 & 1 \\
             0 & 1 & 1
        \end{bmatrix}$
        &
        $\begin{bmatrix}
             1 & 0 & 0 \\
             1 & 0 & 1 \\
             1 & 1 & 1
        \end{bmatrix}$
        &
        $\begin{bmatrix}
             0 & 1 & 0 \\
             0 & 0 & 1 \\
             1 & 1 & 0
        \end{bmatrix}$
        &
        $\begin{bmatrix}
             0 & 1 & 0 \\
             0 & 0 & 1 \\
             1 & 0 & 1
        \end{bmatrix}$ \\[20px]
        $\begin{bmatrix}
             0 & 1 & 0 \\
             0 & 0 & 1 \\
             1 & 1 & 1
        \end{bmatrix}$
        &
        $\begin{bmatrix}
             0 & 1 & 0 \\
             1 & 1 & 0 \\
             1 & 0 & 1
        \end{bmatrix}$
        &
        $\begin{bmatrix}
             0 & 1 & 0 \\
             1 & 1 & 0 \\
             0 & 1 & 1
        \end{bmatrix}$
        &
        $\begin{bmatrix}
             0 & 1 & 0 \\
             1 & 1 & 0 \\
             1 & 1 & 1
        \end{bmatrix}$
        &
        $\begin{bmatrix}
             0 & 1 & 0 \\
             1 & 0 & 1 \\
             0 & 1 & 1
        \end{bmatrix}$
        &
        $\begin{bmatrix}
             0 & 1 & 0 \\
             0 & 1 & 1 \\
             1 & 1 & 1
        \end{bmatrix}$
        &
        $\begin{bmatrix}
             0 & 0 & 1 \\
             1 & 1 & 0 \\
             1 & 0 & 1
        \end{bmatrix}$ \\[20px]
        $\begin{bmatrix}
             0 & 0 & 1 \\
             1 & 1 & 0 \\
             0 & 1 & 1
        \end{bmatrix}$
        &
        $\begin{bmatrix}
             0 & 0 & 1 \\
             1 & 0 & 1 \\
             0 & 1 & 1
        \end{bmatrix}$
        &
        $\begin{bmatrix}
             0 & 0 & 1 \\
             1 & 0 & 1 \\
             1 & 1 & 1
        \end{bmatrix}$
        &
        $\begin{bmatrix}
             0 & 0 & 1 \\
             0 & 1 & 1 \\
             1 & 1 & 1
        \end{bmatrix}$
        &
        $\begin{bmatrix}
             1 & 1 & 0 \\
             1 & 0 & 1 \\
             1 & 1 & 1
        \end{bmatrix}$
        &
        $\begin{bmatrix}
             1 & 1 & 0 \\
             0 & 1 & 1 \\
             1 & 1 & 1
        \end{bmatrix}$
        &
        $\begin{bmatrix}
             1 & 0 & 1 \\
             0 & 1 & 1 \\
             1 & 1 & 1
        \end{bmatrix}$
    \end{tabular}
    \caption{All ternary polarities.}
    \label{tab:ternary-polarities}
\end{table}

Vectors that are not linearly independent can also be used to represent the polarity literals. Thus, Theorem~\ref{thm:polarity} can be generalized to Lemma~\ref{lem:polarity-gen}.

\begin{lemma}
    \label{lem:polarity-gen}
    Any set of vectors $T^r$ can be used for the polarity literals $P^r$, if every possible set $S \subseteq V$ for the literal $X^S$ can be generated by an AND-XOR expression with a subset of the polarity literals 
    $\square$
\end{lemma}

In a polarity matrix $P$ that is not linearly independent, there is more than one way to generate $X^S$ from the given polarity literals. This is because, since $P$ is not linearly independent, some vectors $T^r$ (corresponding to $P^r$) can be represented by an XOR combination of other row vectors in $P$, leading to multiple representations of the same function. Thus, the XOR expressions created with such polarity literals are no longer canonical forms. These polarities can be used to make quantum circuits, but they are costlier than their linear independent alternatives.

\begin{table}[!htb]
    \centering
    \begin{tabular}{l|l}
        Name & Symbol \\
        \hline
        Binary Variable & $x$ \\
        MVI Variable & $X$ \\
        MVI Literal & $X^S$ \\
        Set of Truth Values & $S$ \\
        Set of All Possible Truth Values & $V=\{0,1,\dots,v-1\}$ \\
        Polarity Matrix & $P$ \\
        Polarity Literal & $P^r$ \\
        Row Vector from Literal & $T^r$ 
    \end{tabular}
    \caption{Notation.}
    \label{tab:notation}
\end{table}

Table~\ref{tab:notation} summarizes the notation used throughout this paper, although slight differences may appear (e.g., both $P$ and $Q$ will be used for polarity). When working with more than one variable, subscripts will be used, such as with $X_i^{S_i}$.

\subsection{Ternary Input Shannon and Davio-like Expansions}
\label{ternary-expansions}
The Shannon expansion \cite{shannon}, which is stated in Eq.~\eqref{eq:shannon}, and the Davio expansions \cite{davio}, which are stated in Eq.~\eqref{eq:positive-davio} and Eq.~\eqref{eq:negative-davio}, from binary logic can be generalized to MVI logic. This subsection introduces similar expressions for ternary input logic.

Recall that the original Shannon expansion, with XOR, from Eq.~\eqref{eq:shannon} is 
\begin{equation*}
    f = af_{a} \oplus \bar{a}f_{\bar{a}}.
\end{equation*}

The literal $a$ is 1 if $a=1$, and the literal $\bar{a}$ is 1 if $a=0$. Thus, let $a=a^1$ and $\bar{a}=a^0$, where $a^k$ is 1 if $a=k$. The cofactors $f_a$ and $f_{\bar{a}}$ can be rewritten as $f_{a=1}$ and $f_{a=0}$, with $f_{a=k}$ equal to the function $f$ when $a=k$.

The Shannon expansion can thus be rewritten as
\begin{align*}
    f 
    &= a^1f_{a=1} \oplus a^0f_{a=0} \\
    &= a^0f_{a=0} \oplus a^1f_{a=1}.
\end{align*}

This can be extended to ternary input functions. The binary variable $a$ can be replaced by the multi-valued variable $X$. The ternary input Shannon expansion is expressed in Eq.~\eqref{eq:ternary-input-shannon}.

\begin{equation}
    \label{eq:ternary-input-shannon}
    f = X^0f_{X=0} \oplus X^1f_{X=1} \oplus X^2f_{X=2}
\end{equation}

The polarity matrix for the ternary variable $X$ with $P^1=X^0$, $P^2=X^1$, and $P^3=X^2$ is
$$\begin{bmatrix}
    1 & 0 & 0 \\
    0 & 1 & 0 \\ 
    0 & 0 & 1 \\
\end{bmatrix}$$
which is equivalent to finding a function in terms of $X^0$, $X^1$, and $X^2$. This is the same as finding the Shannon expansion. 

Like their binary counterparts, the ternary input Davio-like expansions can be obtained from the ternary input Shannon expansion. One potential Davio-like expansion is shown in Eq.~\eqref{eq:ternary-davio-example1}, which corresponds to the polarity with literals $X^{0,1,2}=1$, $X^{0,1}$, and $X^{1,2}$.
\begin{align}
    \label{eq:ternary-davio-example1}
    f
    &= X^0f_{X=0} \oplus X^1f_{X=1} \oplus X^2f_{X=2} \nonumber \\
    &= (1 \oplus X^{1,2})f_{X=0} \oplus (1 \oplus X^{0,1} \oplus X^{1,2})f_{X=1} \oplus (1 \oplus X^{0,1})f_{X=2} \nonumber \\
    &= f_{X=0} \oplus X^{1,2}f_{X=0} \oplus f_{X=1} \oplus X^{0,1}f_{X=1} \oplus X^{1,2}f_{X=1} \oplus f_{X=2} \oplus X^{0,1}f_{X=2} \nonumber \\
    &= f_{X=0} \oplus f_{X=1} \oplus f_{X=2} \oplus X^{0,1}f_{X=1} \oplus X^{0,1}f_{X=2} \oplus X^{1,2}f_{X=0} \oplus X^{1,2}f_{X=1} \nonumber \\
    &= (f_{X=0} \oplus f_{X=1} \oplus f_{X=2}) \oplus X^{0,1}(f_{X=1} \oplus f_{X=2}) \oplus X^{1,2}(f_{X=0} \oplus f_{X=1})
\end{align}

However, this is not the only possible ternary input Davio-like expansion. Another variant is shown in Eq.~\eqref{eq:ternary-davio-example2}, which corresponds to the polarity $X^0$, $X^{0,1}$, $X^{0,2}$.
\begin{align}
    \label{eq:ternary-davio-example2}
    f
    &= X^0f_{X=0} \oplus X^1f_{X=1} \oplus X^2f_{X=2} \nonumber \\
    &= X^0f_{X=0} \oplus (X^0 \oplus X^{0,1})f_{X=1} \oplus (X^0 \oplus X^{0,2})f_{X=2} \nonumber \\
    &= X^0f_{X=0} \oplus X^0f_{X=1} \oplus X^{0,1}f_{X=1} \oplus X^0f_{X=2} \oplus X^{0,2}f_{X=2} \nonumber \\
    &= X^0f_{X=0} \oplus X^0f_{X=1} \oplus X^0f_{X=2} \oplus X^{0,1}f_{X=1} \oplus X^{0,2}f_{X=2} \nonumber \\
    &= X^0(f_{X=0} \oplus f_{X=1} \oplus f_{X=2}) \oplus X^{0,1}(f_{X=1}) \oplus X^{0,2}(f_{X=2})
\end{align}
The other polarity matrices correspond to all the possible Davio-like expansions, so there are 27 Davio-like expansions in total. These relations can also be generalized to any multi-valued input logic.

Although multi-valued input Shannon and Davio-like expansions won't be elaborated further in this paper, they are important, as they link the concept of polarity for multi-valued input logic with classical Boolean algebra.

\subsection{Extending Reed-Muller Forms to Multi-Valued Input Functions}
\label{mvi-rm}
The approach presented in this paper for generating multi-valued input FPRM forms, defined in Definition~\ref{def:mvi-fprm}, is similar to those from \cite{schafer91}. We propose calling multi-valued input FPRM forms ``MVI-FPRM" forms instead. Similarly, we call multi-valued input GRM forms, defined in Definition~\ref{def:mvi-grm}, ``MVI-GRM" forms.

Before MVI-FPRM forms are formally defined, let us consider a simple example of such forms in Example~\ref{ex:fprm}, which shows the algebraic derivation.

\begin{example}
    \label{ex:fprm}
    Take the MVI function $F_1=X_1^{0,2,3}X_2^{0,1}$ with quaternary (4-valued) variable $X_1$ and ternary (3-valued) variable $X_2$. Note that this differs from the function $f$ in Example~\ref{ex:mvi function} because $X_2$ is ternary for $F_1$ while it was quaternary for $f$.

    Suppose that the function has the polarity $P_1$,$P_2$.
    $$
    P_1=
    \begin{bmatrix}
        1 & 1 & 1 & 1 \\
        0 & 1 & 0 & 1 \\
        0 & 0 & 1 & 1 \\
        0 & 1 & 1 & 1
    \end{bmatrix}
    =
    \begin{bmatrix}
        T_1^1 \\
        T_1^2 \\
        T_1^3 \\
        T_1^4
    \end{bmatrix}
    , \
    P_2=
    \begin{bmatrix}
        1 & 1 & 1 \\
        1 & 0 & 0 \\
        0 & 0 & 1
    \end{bmatrix}
    =
    \begin{bmatrix}
        T_2^1 \\
        T_2^2 \\
        T_2^3
    \end{bmatrix}
    $$

    This means that the polarity literals for $X_1$ are $$P_1^1=X_1^{0,1,2,3}=1, \ P_1^2=X_1^{1,3}, \ P_1^3=X_1^{2,3}, \ P_1^4=X_1^{1,2,3}.$$ The polarity literals for $X_2$ are $$P_2^1=X_1^{0,1,2}=1, \ P_2^2=X_1^0, \ P_2^3=X_1^2.$$
    
    The literals $X_1^{0,2,3}$ and $X_2^{0,1}$ can be represented in terms of the polarity literals as follows:
    \begin{align*}
        X_1^{0,2,3} &= P_1^1 \oplus P_1^3 \oplus P_1^4 \\
        &= X_1^{0,1,2,3} \oplus X_1^{2,3} \oplus X_1^{1,2,3} \\
        X_2^{0,1} &= P_2^1 \oplus P_2^3 \\
        &= X_2^{0,1,2} \oplus X_2^2.
    \end{align*}

    From substituting these values into the function $F_1$, taking into account that $P_1^1=P_2^1=1$, and expanding, we obtain the MVI-FPRM form for $F_1$, as demonstrated below:  
    
    \begin{align*}
        F_1 &=X_1^{0,2,3}X_2^{0,1} \\ 
        &= (P_1^1 \oplus P_1^3 \oplus P_1^4)(P_2^1 \oplus P_2^3) \\
        &= (1 \oplus P_1^3 \oplus P_1^4)(1 \oplus P_2^3) \\
        &= 1 \oplus P_1^3 \oplus P_1^4 \oplus P_2^3 \oplus P_1^3P_2^3 \oplus P_1^4P_2^3 \\
        &= 1 \oplus X_1^{2,3} \oplus X_1^{1,2,3}\oplus X_2^2 \oplus X_1^{2,3}X_2^2 \oplus X_1^{1,2,3}X_2^2.
    \end{align*}
\end{example}

Definition~\ref{def:mvi-fprm} extends the concept of FPRM forms for Boolean functions to the concept of MVI-FPRM forms for MVI functions.

\begin{definition}
\label{def:mvi-fprm}
    The MVI-FPRM form of a single-output MVI function $F(X_1, X_2,\dots, X_n)$ with polarity $P_1, P_2,\dots, P_n$ is defined with the so-called spectral coefficients $M_{P_1^{r_1}, P_2^{r_2},\dots, P_n^{r_n}}$ in Eq.~\eqref{eq:mvi-fprm}. Similar to binary FPRM forms, MVI-FPRM forms assign each variable a fixed polarity.
    \begin{align}
            \label{eq:mvi-fprm}
            &F(X_1,X_2,\dots, X_n) = \bigoplus_{r_i \in V_i} M_{P_1^{r_1},P_2^{r_2},\dots,P_n^{r_n}}P_1^{r_1}P_2^{r_2} \cdots P_n^{r_n}
    \end{align}
    $\square$
\end{definition}

Please note that while in binary, for a function of $n$ variables, there are $2^n$ possible FPRM forms, but in the case of ternary input functions, there are $28^n$ MVI-FPRM forms, which makes synthesis much more difficult and points to the importance of selecting good decoders.

Recall that the binary GRM is defined as a form where each product term has a different set of variables than every other term. Therefore, we extend this to the MVI-GRM form. 

\begin{definition}
\label{def:mvi-grm}
    The MVI-GRM form of a single-output MVI function is defined as a form where every two product terms have a different set of multi-valued variables. $\square$
\end{definition}

For instance, $X_1^{0,1}X_2^{0,2,3} \oplus X_1^{0,2}X_2^{1}$ is not an MVI-GRM form because both terms have the set of variables $\{X_1, X_2\}$, but $X_1^{0,1}X_2^{0,2,3} \oplus X_1^{0,2}X_3^{1}$ is an MVI-GRM form.

Although the names "spectral" and "spectrum" will be used, no knowledge of spectral theory is necessary. The spectral coefficients are a format for representing the MVI-FPRM form of a function with 0's and 1's.

The polarity of an MVI-FPRM form is defined in Definition~\ref{def:mvi-fprm polarity}.

\begin{definition}
    \label{def:mvi-fprm polarity}
    The polarity of an MVI-FPRM form is the vector of polarity matrices describing the polarity for each multi-valued literal. $\square$
\end{definition}

\section{Circuit Synthesis Based on MVI-FPRM}
\label{mvi-fprm-circuit}

An approach for visualizing the MVI-FPRM form of a function using Marquand charts is shown in Example~\ref{ex:fprm spectrum}.

\begin{example}
    \label{ex:fprm spectrum}
    The MVI-FPRM form of the function $F_1=X_1^{0,2,3}X_2^{0,1}$ from Example~\ref{ex:fprm} with polarities $P_1$ and $P_2$ can be represented by its spectrum $M$ as shown in Fig.~\ref{fig:ex4-spectrum}, which illustrates the standard trivial functions, which are all the possible product terms with the given polarity literals.

    \begin{figure}[!htb]
        \centering
        \includegraphics[width=0.6\linewidth]{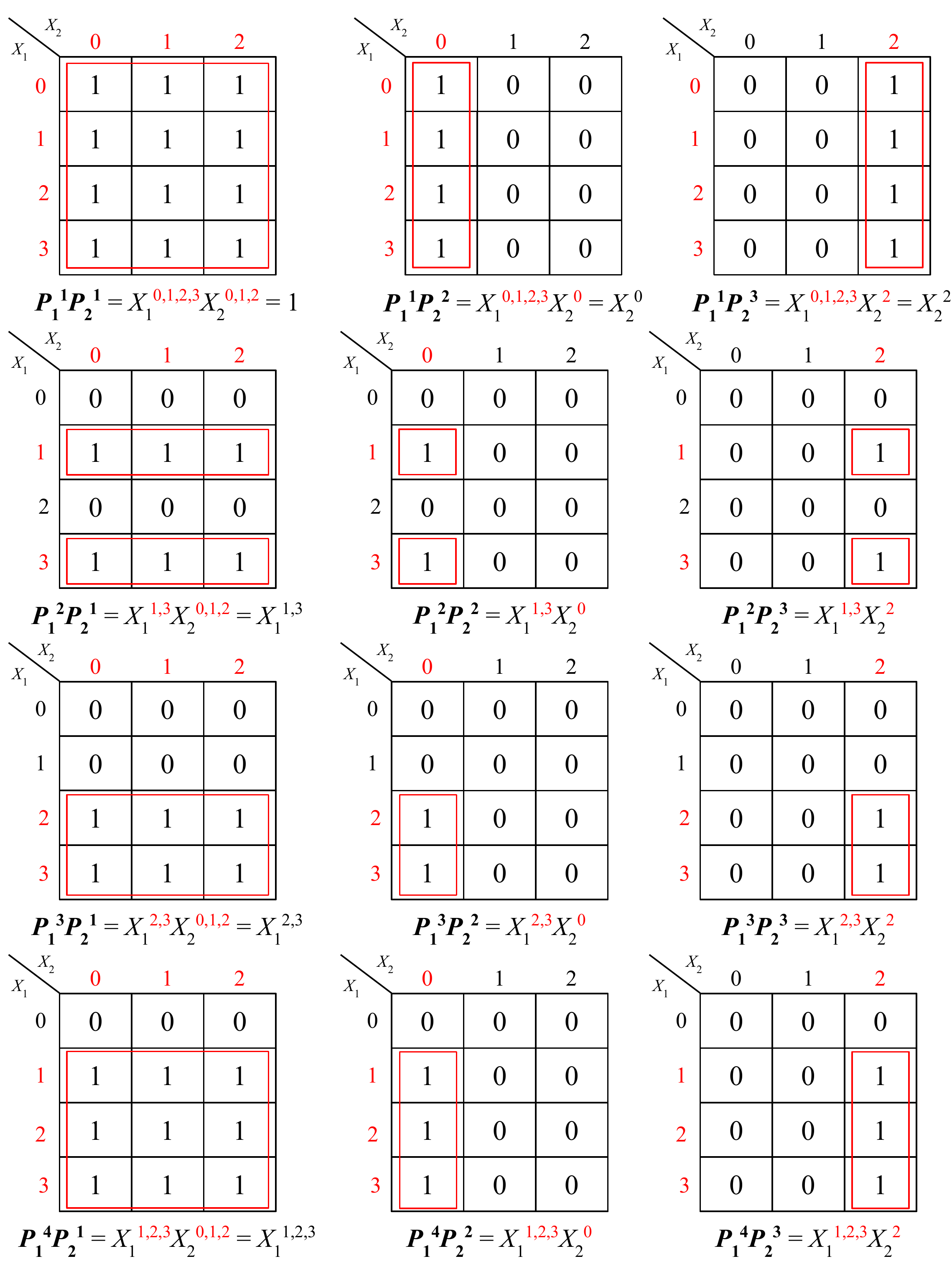}
        \caption{Standard trivial functions for variables $X_1$ and $X_2$ with polarities $P_1$ and $P_2$ respectively for Example~\ref{ex:fprm spectrum}.}
        \label{fig:ex4-spectrum}
    \end{figure}

    The function $F_1$ can be shown to be equal to 
    $$F_1=1 \oplus P_1^3 \oplus P_1^4 \oplus P_2^3 \oplus P_1^3P_2^3 \oplus P_1^4P_2^3$$ 
    (remember that $P_1^1=P_2^1=1$).

    The MVI-FPRM form for this function can be visually verified by taking the XOR of the Marquand charts for each term in $1 \oplus P_1^3 \oplus P_1^4 \oplus P_2^3 \oplus P_1^3P_2^3 \oplus P_1^4P_2^3$ and comparing it to the Marquand chart for $F_1$.

    \begin{figure}[!htb]
        \centering
        \includegraphics[width=0.7\linewidth]{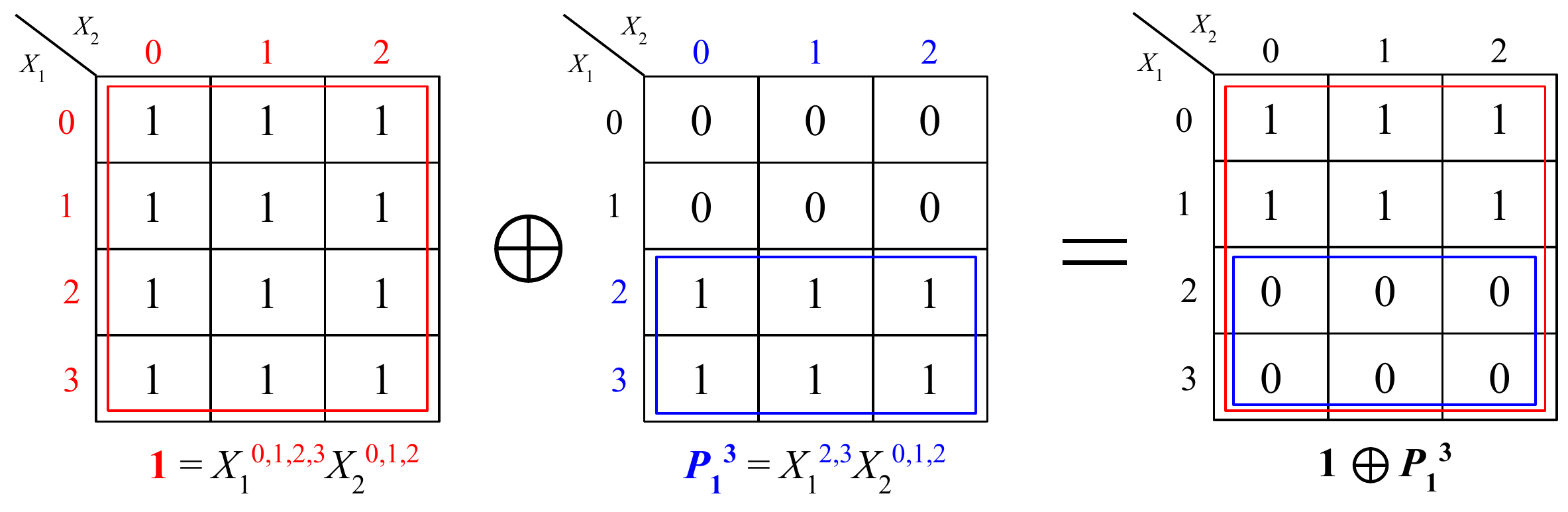}
        \caption{Step 1: XOR $1$ and $P_1^3$ for Example~\ref{ex:fprm spectrum}.}
        \label{fig:ex4-xor-step1}
    \end{figure}

    \begin{figure}[!htb]
        \centering
        \includegraphics[width=0.7\linewidth]{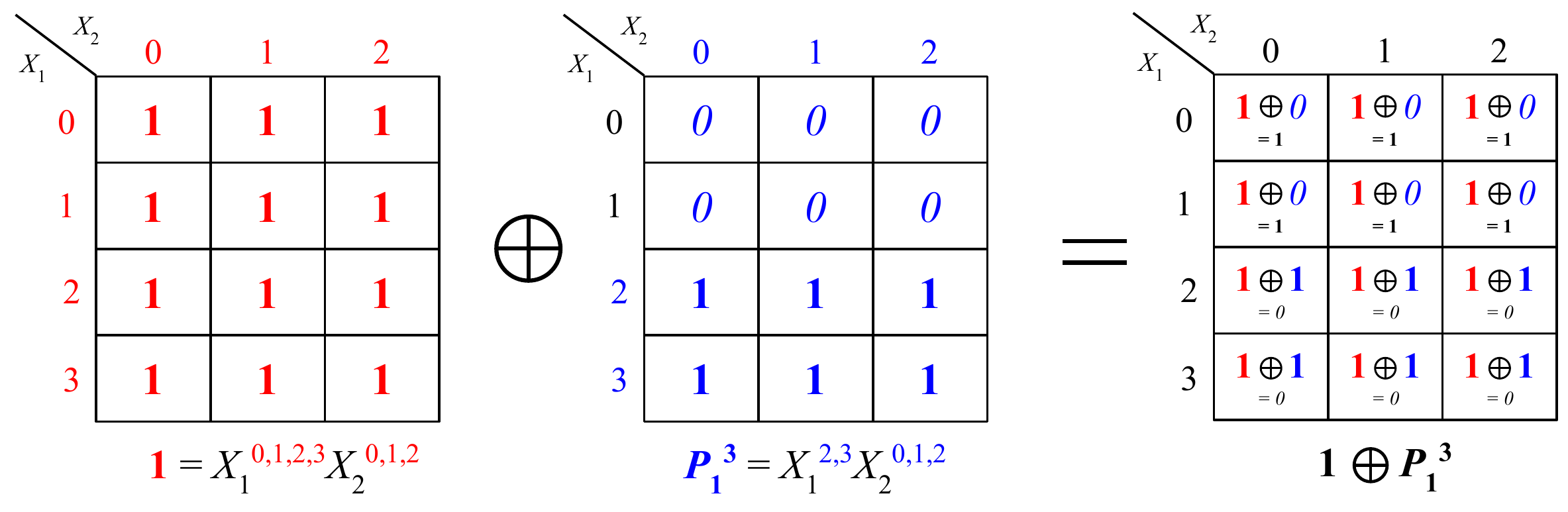}
        \caption{Process showing the XOR operation in each cell for Example~\ref{ex:fprm spectrum}.}
        \label{fig:ex4-xor-step1-xorvis}
    \end{figure}

    Let's start with the first two terms: $1$ and $P_1^3$. The first course of action is to draw both of the Marquand Charts for $1$ and $P_1^3$. Next, group together all the 1's in the Marquand charts, where the group for $1$ is in \textcolor{red}{red}, while the group for $P_1^3$ is in \textcolor{blue}{blue}. Then, overlap the two groups on a new Marquand chart. For each cell in the chart, if an odd number of groups cover it, then the value for that cell is 1; if an even number of groups cover it, then the value for that cell is 0. Fig.~\ref{fig:ex4-xor-step1} shows the first step of the verification process, which combines all of the previous steps into one. The resulting Marquand chart is for the function $1 \oplus P_1^3$. This method essentially finds the XOR of the values in each cell for the functions, which is depicted in Fig.~\ref{fig:ex4-xor-step1-xorvis}. 

    \begin{figure}[!htb]
        \centering
        \begin{subfigure}[!h]{0.7\linewidth}
            \centering
            \includegraphics[width=\linewidth]{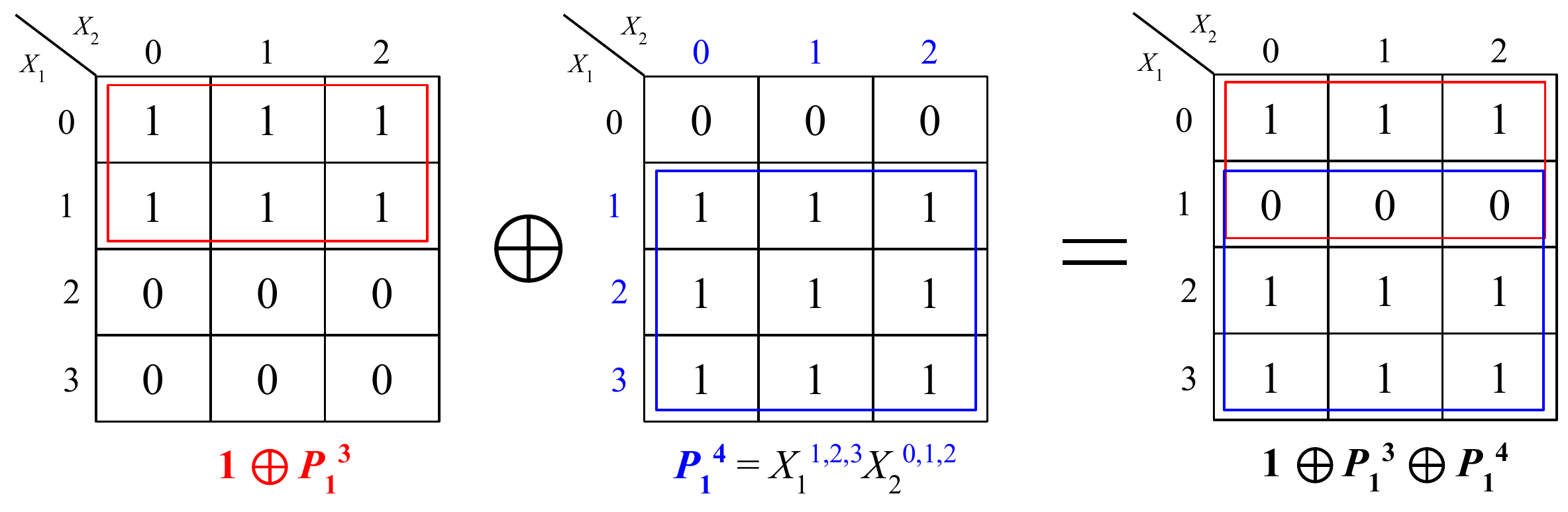}
            \caption{Step 2: XOR $1 \oplus P_1^3$ with $P_1^4$.}
        \end{subfigure}
        \begin{subfigure}[!h]{0.7\linewidth}
            \centering
            \includegraphics[width=\linewidth]{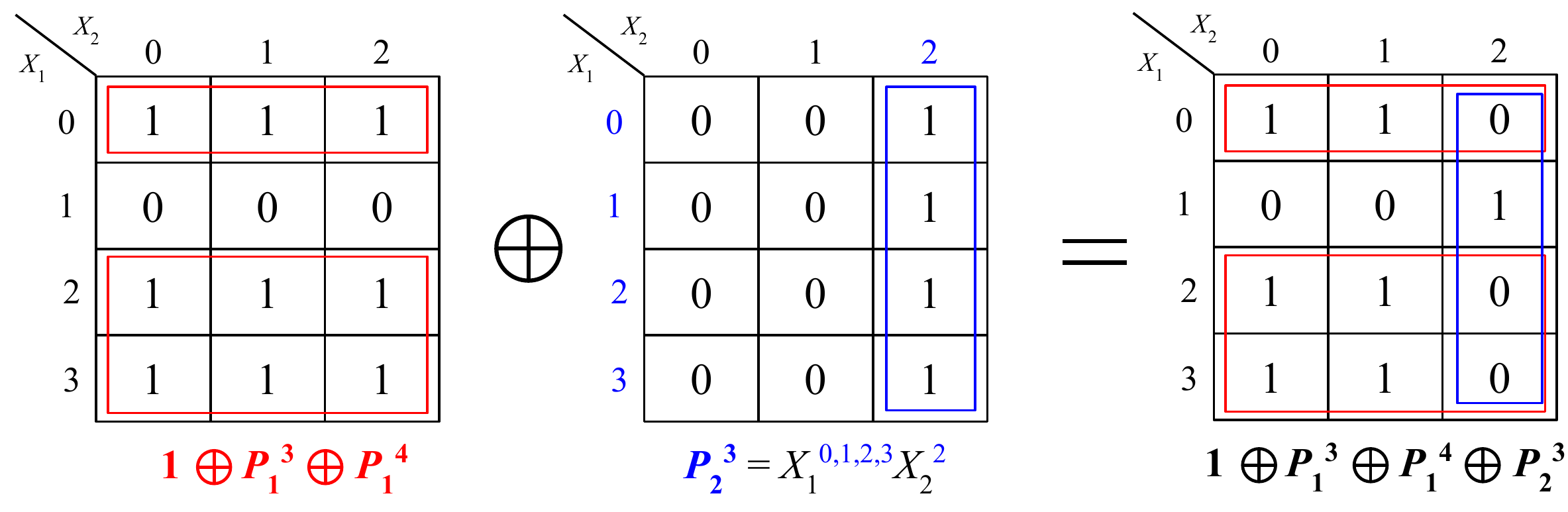}
            \caption{Step 3: XOR $1 \oplus P_1^3 \oplus P_1^4$ with $P_2^3$.}
        \end{subfigure}
        \begin{subfigure}[!h]{0.7\linewidth}
            \centering
            \includegraphics[width=\linewidth]{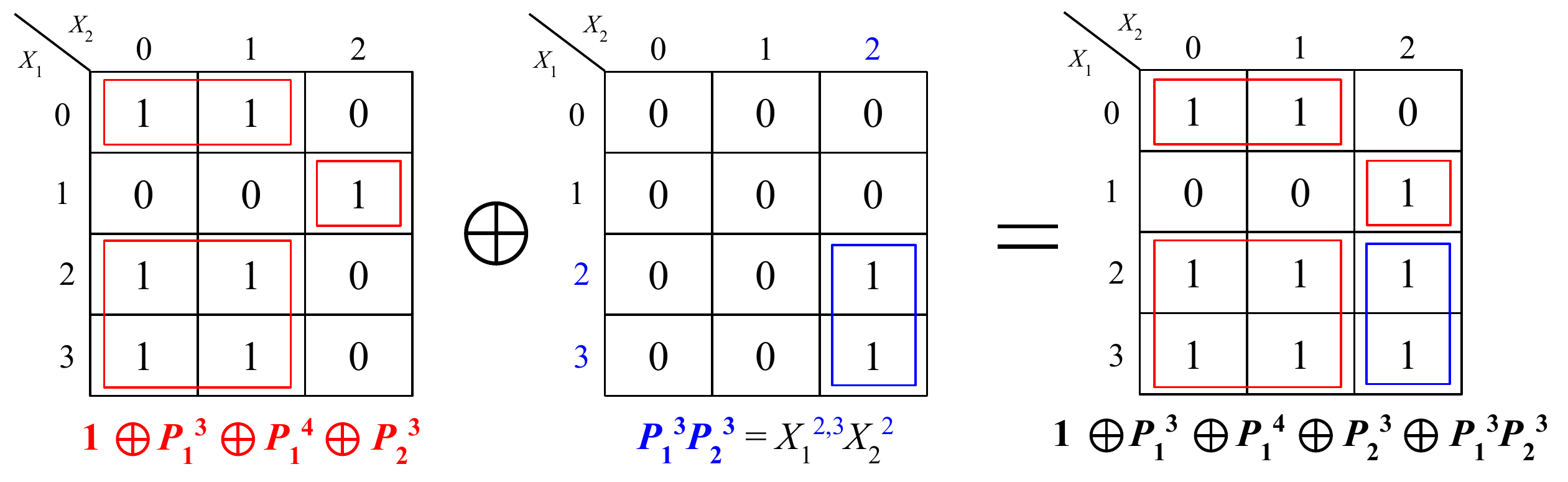}
            \caption{Step 4: XOR $1 \oplus P_1^3 \oplus P_1^4 \oplus P_2^3$ with $P_1^3P_2^3$.}
        \end{subfigure}
        \begin{subfigure}[!h]{0.7\linewidth}
            \centering
            \includegraphics[width=\linewidth]{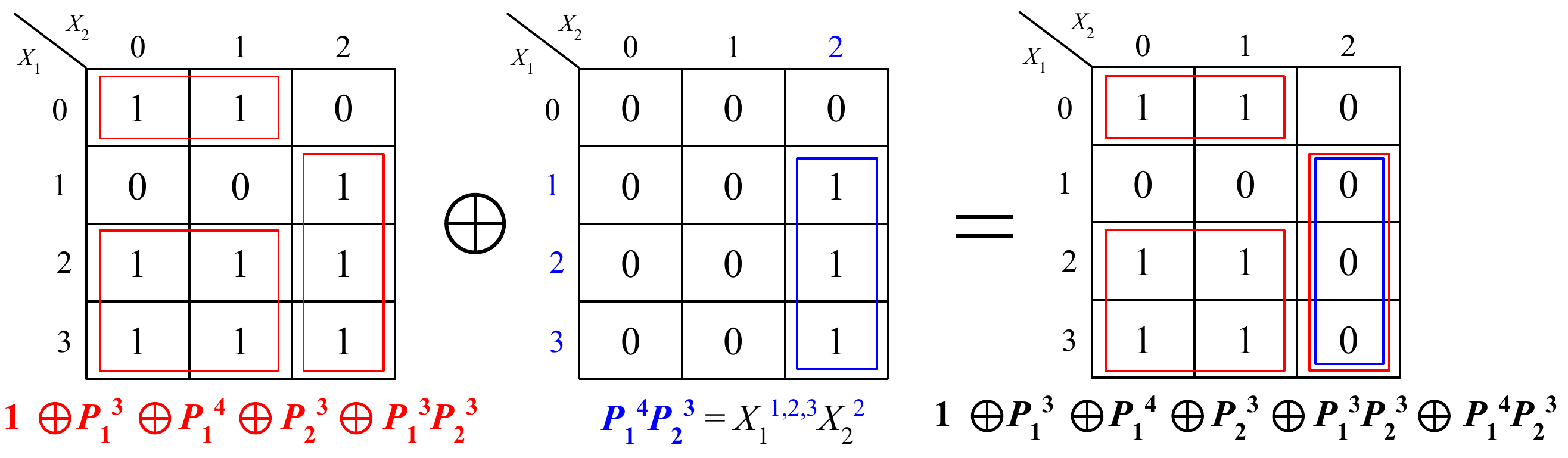}
            \caption{Step 5: XOR $1 \oplus P_1^3 \oplus P_1^4 \oplus P_2^3 \oplus P_1^3P_2^3$ with $P_1^4P_2^3$ to get $F_1$.}
        \end{subfigure}
        \caption{The remaining four steps for applying XOR on the Marquand charts for all terms in $1 \oplus P_1^3 \oplus P_1^4 \oplus P_2^3 \oplus P_1^3P_2^3 \oplus P_1^4P_2^3$ for Example~\ref{ex:fprm spectrum}.}
        \label{fig:ex4-xor-process}
    \end{figure}

    Applying this strategy with all the other terms gives the Marquand charts shown in Fig.~\ref{fig:ex4-xor-process}. The resulting chart is equal to the one in Fig.~\ref{fig:ex4-f-chart}, which is the Marquand chart that is obtained from the original function $F_1=X_1^{0,2,3}X_2^{0,1}$.

    The spectral coefficients for the function $F_1$ with the polarities $P_1$ and $P_2$ are listed in Table~\ref{tab:ex4-spectrum}. A value of 1 for the spectral coefficient $M_{P_1^{r_1}P_2^{r_2}}$ indicates that the corresponding term $P_1^{r_1}P_2^{r_2}$ is in the MVI-FPRM form of the function $F_1$, while a value of 0 indicates that it is not.

    \begin{table}[!htb]
        \centering
        \renewcommand{\arraystretch}{1.3}
        \resizebox{\textwidth}{!}{%
        \begin{tabular}{l|ccc ccc ccc ccc}
            Coefficient 
            & $M_{P_1^1P_2^1}$ & $M_{P_1^1P_2^2}$ & $M_{P_1^1P_2^3}$ 
            & $M_{P_1^2P_2^1}$ & $M_{P_1^2P_2^2}$ & $M_{P_1^2P_2^3}$ 
            & $M_{P_1^3P_2^1}$ & $M_{P_1^3P_2^2}$ & $M_{P_1^3P_2^3}$
            & $M_{P_1^4P_2^1}$ & $M_{P_1^4P_2^2}$ & $M_{P_1^4P_2^3}$
            \\[2.5px]
            \hline
            Value
            & 1 & 0 & 1
            & 0 & 0 & 0
            & 1 & 0 & 1
            & 1 & 0 & 1
            \\
        \end{tabular}}
        \caption{Spectrum of function $F_1$ for Example~\ref{ex:fprm spectrum}.}
        \label{tab:ex4-spectrum}
    \end{table}

    \begin{figure}[tbp]
        \centering
        \begin{subfigure}{0.35\linewidth}
            \centering
            \includegraphics[width=\linewidth]{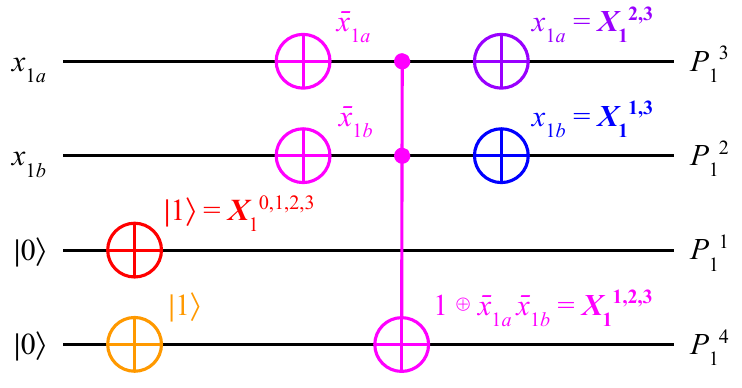}
            \caption{Decoder for $P_1$.}
            \label{fig:ex4-x1-decoder}
        \end{subfigure}
        \begin{subfigure}{0.275\linewidth}
            \centering
            \includegraphics[width=\linewidth]{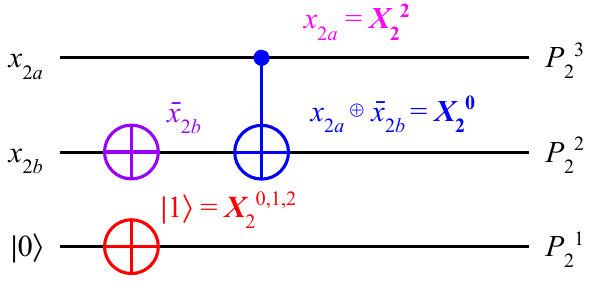}
            \caption{Decoder for $P_2$.}
            \label{fig:ex4-x2-decoder}
        \end{subfigure}
        \\
        \begin{subfigure}{0.22\linewidth}
            \centering
            \includegraphics[width=\linewidth]{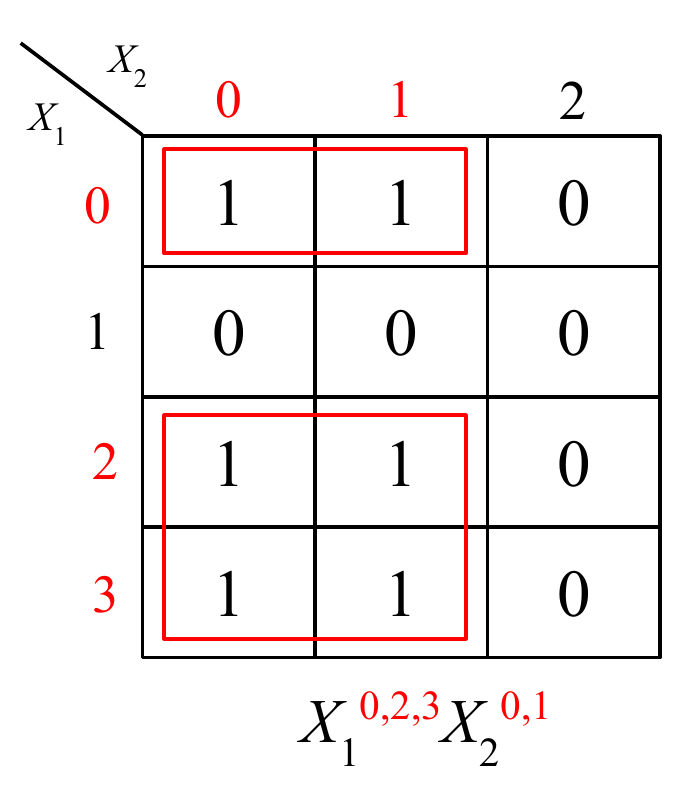}
            \caption{Marquand chart for function $F_1$.}
            \label{fig:ex4-f-chart}
        \end{subfigure}
        \begin{subfigure}{0.4\linewidth}
            \includegraphics[width=\linewidth]{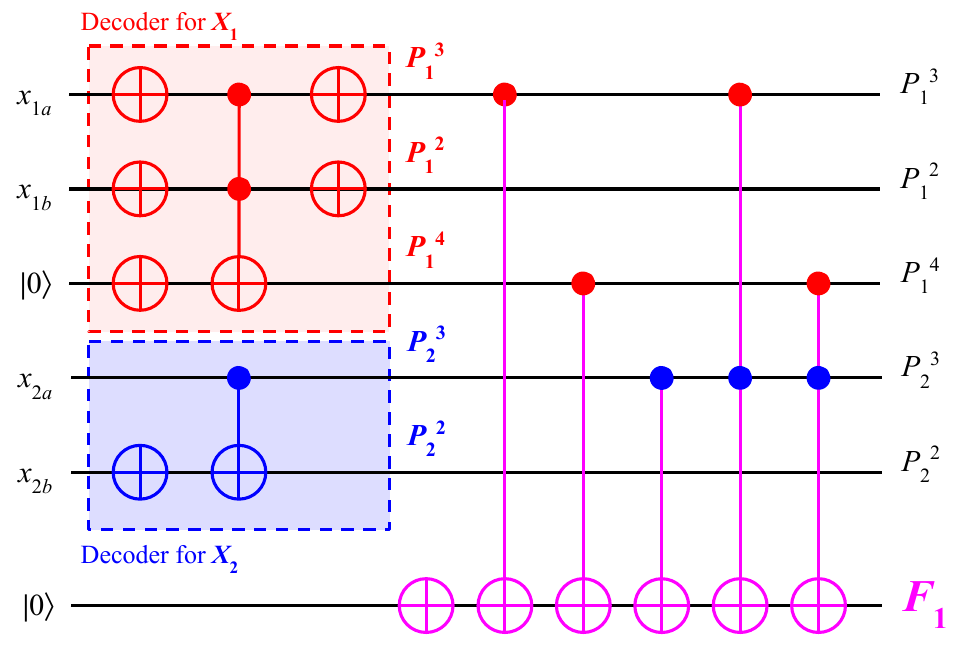}
            \caption{Circuit for $F_1$.}
            \label{fig:ex4-circuit}
        \end{subfigure}
        \caption{Marquand chart, decoders for $P_1$ and $P_2$, and circuit for $F_1$ from Example~\ref{ex:fprm spectrum}.}
    \end{figure}

    Fig.~\ref{fig:ex4-x1-decoder} shows the decoder for $P_1$ (with $X_1$ encoded by binary variables $x_{1a}$ and $x_{1b}$) and Fig.~\ref{fig:ex4-x2-decoder} shows the decoder for $P_2$ (with $X_2$ encoded by binary variables $x_{2a}$ and $x_{2b}$). Realizing $F_1$ as a binary quantum circuit leads to the circuit shown in Fig.~\ref{fig:ex4-circuit}, which uses the results of the decoders (with $P_1^1$ and $P_2^1$ omitted because they both equal 1) to calculate $F_1$. It can be noted that the bigger the circuit is, the more negligible the cost of the decoders becomes.
\end{example}

So far, the examples have covered functions with only one product term. However, these methods can also apply to functions with multiple terms, as demonstrated in Example~\ref{ex:fprm two terms}.

\begin{example}
    \label{ex:fprm two terms}
    The MVI-FPRM form of the function $F_2 = X_1^{0,2,3}X_2^{0,1} \oplus X_1^0X_2^2$ with the same polarities from Examples~\ref{ex:fprm} and \ref{ex:fprm spectrum} can be calculated by finding the XOR of the MVI-FPRM forms for the terms $X_1^{0,2,3}X_2^{0,1}$ and $X_1^0X_2^2$.

    The Marquand chart for $F_2$ is shown in Fig.~\ref{fig:ex5-f-chart}.

    \begin{figure}[!htb]
        \centering
        \includegraphics[width=0.25\linewidth]{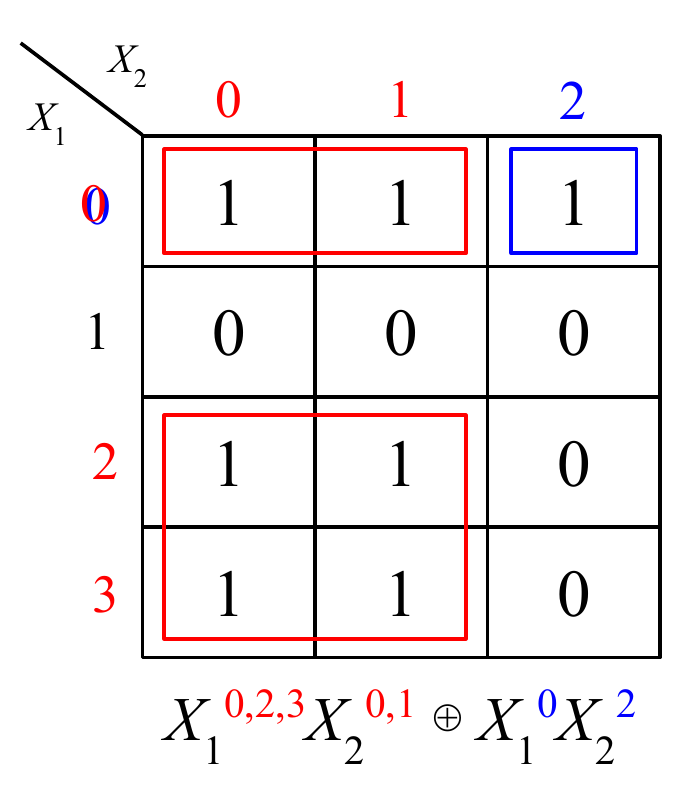}
        \caption{Marquand chart for function $F_2$ from Example~\ref{ex:fprm two terms}.}
        \label{fig:ex5-f-chart}
    \end{figure}

    In Examples~\ref{ex:fprm} and \ref{ex:fprm spectrum}, it was found and verified that the MVI-FPRM form of the first term ($F_1$) is 
    \begin{align*}
        X_1^{0,2,3}X_2^{0,1}
        &= P_1^1P_2^1 \oplus P_1^3P_2^1 \oplus P_1^4P_2^1 \oplus P_1^1P_2^3 \oplus P_1^3P_2^3 \oplus P_1^4P_2^3 \\
        &= 1 \oplus P_1^3 \oplus P_1^4 \oplus P_2^3 \oplus P_1^3P_2^3 \oplus P_1^4P_2^3.
    \end{align*}

    The MVI-FPRM form of the second term, $X_1^0X_2^2$, in the function can be found to be
    \begin{align*}
        X_1^0X_2^2
        &= P_1^1P_2^3 \oplus P_1^4P_2^3 \\
        &= P_2^3 \oplus P_1^4P_2^3.
    \end{align*}
    This is verified algebraically as follows:
    \begin{align*}
        P_1^1P_2^3 \oplus P_1^4P_2^3 
        &= (P_1^1 \oplus P_1^4)P_2^3 \\
        &= (X_1^{0,1,2,3} \oplus X^{1,2,3})X_2^2 \\
        &= X_1^0X_2^2.
    \end{align*}

    The MVI-FPRM form for the whole function $F_2$ can be found by using spectral coefficients or algebraically. Both methods will be done in this example, starting with spectral coefficients. 
    
    The spectral coefficients for the terms $X_1^{0,2,3}X_2^{0,1}$ and $X_1^0X_2^2$ are listed in Table~\ref{tab:ex5-spectral-coeff}.

    \begin{table}[!htb]
        \centering
        \renewcommand{\arraystretch}{1.3}
        \resizebox{\textwidth}{!}{%
        \begin{tabular}{l|ccc ccc ccc ccc}
            Coefficient 
            & $M_{P_1^1P_2^1}$ & $M_{P_1^1P_2^2}$ & $M_{P_1^1P_2^3}$ 
            & $M_{P_1^2P_2^1}$ & $M_{P_1^2P_2^2}$ & $M_{P_1^2P_2^3}$ 
            & $M_{P_1^3P_2^1}$ & $M_{P_1^3P_2^2}$ & $M_{P_1^3P_2^3}$
            & $M_{P_1^4P_2^1}$ & $M_{P_1^4P_2^2}$ & $M_{P_1^4P_2^3}$
            \\[2.5px]
            \hline
            Value for $X_1^{0,2,3}X_2^{0,1}$
            & 1 & 0 & 1
            & 0 & 0 & 0
            & 1 & 0 & 1
            & 1 & 0 & 1
            \\
            Value for $X_1^0X_2^2$
            & 0 & 0 & 1
            & 0 & 0 & 0
            & 0 & 0 & 0
            & 0 & 0 & 1
            \\
        \end{tabular}}
        \caption{Spectrum of the functions $X_1^{0,2,3}X_2^{0,1}$ and $X_1^0X_2^2$ for Example~\ref{ex:fprm two terms}.}
        \label{tab:ex5-spectral-coeff}
    \end{table}
    
    The spectral coefficients for $F_2$ can be calculated using the spectral coefficients from its terms $X_1^{0,2,3}X_2^{0,1}$ and $X_1^0X_2^2$.
    
    This process is expressed in Table~\ref{tab:ex5-spectral-coeff-xor}, where the XOR of the spectral coefficients for $X_1^{0,2,3}X_2^{0,1}$ and $X_1^0X_2^2$ in each column is calculated, and the result is the spectral coefficient for $F$.

    \begin{table}[!htb]
        \centering
        \renewcommand{\arraystretch}{1.3}
        \resizebox{\textwidth}{!}{%
        \begin{tabular}{l|ccc ccc ccc ccc}
            Coefficient 
            & $M_{P_1^1P_2^1}$ & $M_{P_1^1P_2^2}$ & $M_{P_1^1P_2^3}$ 
            & $M_{P_1^2P_2^1}$ & $M_{P_1^2P_2^2}$ & $M_{P_1^2P_2^3}$ 
            & $M_{P_1^3P_2^1}$ & $M_{P_1^3P_2^2}$ & $M_{P_1^3P_2^3}$
            & $M_{P_1^4P_2^1}$ & $M_{P_1^4P_2^2}$ & $M_{P_1^4P_2^3}$
            \\[2.5px]
            \hline
            Value for $X_1^{0,2,3}X_2^{0,1}$
            & 1 & 0 & 1
            & 0 & 0 & 0
            & 1 & 0 & 1
            & 1 & 0 & 1
            \\
            Value for $X_1^0X_2^2$
            & 0 & 0 & 1
            & 0 & 0 & 0
            & 0 & 0 & 0
            & 0 & 0 & 1
            \\
            \hline
            Value for $F$
            & 1 & 0 & 0
            & 0 & 0 & 0
            & 1 & 0 & 1
            & 1 & 0 & 0
        \end{tabular}}
        \caption{Spectrum of the function $F_2$ from Example~\ref{ex:fprm two terms}.}
        \label{tab:ex5-spectral-coeff-xor}
    \end{table}

    From the resulting spectral coefficients, it can be found that 
    \begin{align*}
        F_2
        &= P_1^1P_2^1 \oplus P_1^3P_2^1 \oplus P_1^3P_2^3 \oplus P_1^4P_2^1 \\
        &= 1 \oplus P_1^3 \oplus P_1^3P_2^3 \oplus P_1^4 \\ 
        &= 1 \oplus X_1^{2,3} \oplus X_1^{2,3}X_2^2 \oplus X_1^{1,2,3}.
    \end{align*}

    Calculating the MVI-FPRM form of $F_2$ algebraically yields the same results, as done below:
    \begin{align*}
        F_2
        &= X_1^{0,2,3}X_2^{0,1} \oplus X_1^0X_2^2 \\
        &= (1 \oplus P_1^3 \oplus P_1^4 \oplus P_2^3 \oplus P_1^3P_2^3 \oplus P_1^4P_2^3) \oplus (P_2^3 \oplus P_1^4P_2^3) \\
        &= 1 \oplus P_1^3 \oplus P_1^4 \oplus \textcolor{red}{P_2^3} \oplus P_1^3P_2^3 \oplus \textcolor{blue}{P_1^4P_2^3} \oplus \textcolor{red}{P_2^3} \oplus \textcolor{blue}{P_1^4P_2^3} \\ 
        &= 1 \oplus P_1^3 \oplus P_1^4 \oplus P_1^3P_2^3.
    \end{align*}

    The function $F_2$ can then be realized as a circuit, as shown in Fig.~\ref{fig:ex5-circuit-directly}. 
    
    \begin{figure}[!htb]
        \centering
        \includegraphics[width=0.4\linewidth]{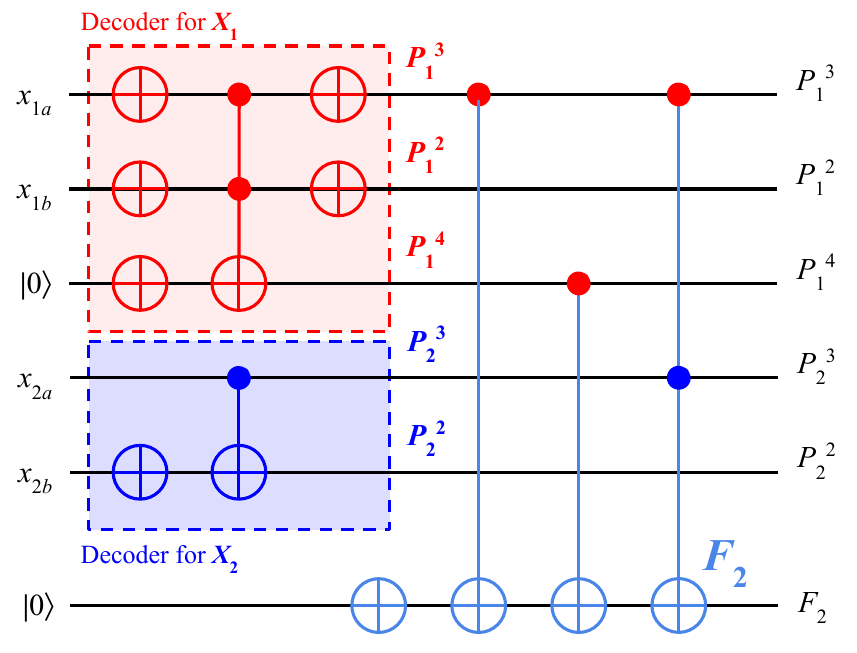}
        \caption{Circuit for $F_2$ realized with decoders from the MVI-FPRM form with polarity $P_1$,$P_2$ for Example~\ref{ex:fprm two terms}.}
        \label{fig:ex5-circuit-directly}
    \end{figure}
    
    The circuit in Fig.~\ref{fig:ex5-circuit-directly} has 7 NOT gates, 3 CNOT gates, and 2 3-bit Toffoli Gates. So, the Maslov cost is 20, and the TQC is 157.
\end{example}

The function $F_2$, from Example~\ref{ex:fprm two terms}, was expressed as $X_1^{0,2,3}X_2^{0,1} \oplus X_1^0X_2^2$. However, functions can be expressed in multiple ways, but the MVI-FPRM is unique for a function. This is shown in Example~\ref{ex:fprm two terms-changed terms}.

\begin{example}
    \label{ex:fprm two terms-changed terms}
    Take the function $F_2$ from Example~\ref{ex:fprm two terms} with polarities $P_1$ and $P_2$.

    The Marquand chart of $F_2$ was drawn in Fig~\ref{fig:ex5-f-chart}. The chart from this figure showed the groupings for the terms $X_1^{0,2,3}X_2^{0,1}$ and $X_1^0X_2^2$, which were the terms used to represent $F_2$ in Example~\ref{ex:fprm two terms}.

    However, there are more possible groupings. One such grouping is shown in Fig.~\ref{fig:ex6-f-chart}.

    \begin{figure}[!htb]
        \centering
        \includegraphics[width=0.25\linewidth]{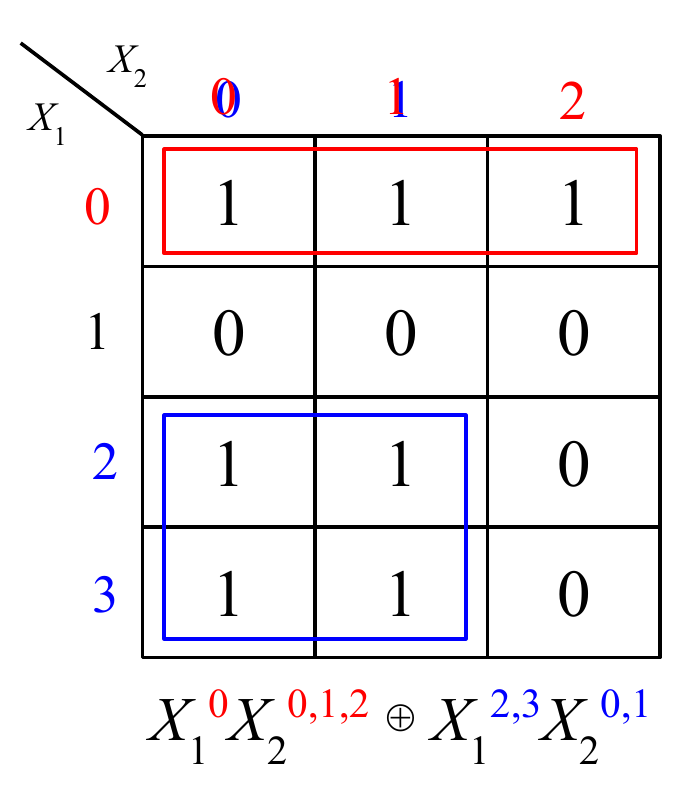}
        \caption{Marquand chart of $F_2$ with a different grouping from Example~\ref{ex:fprm two terms-changed terms}.}
        \label{fig:ex6-f-chart}
    \end{figure}

    The \textcolor{red}{red} group corresponds to the term $X_1^{\textcolor{red}{0}}X_2^{\textcolor{red}{0,1,2}}=X_1^0$, and the \textcolor{blue}{blue} group corresponds to the term $X_1^{\textcolor{blue}{2,3}}X_2^{\textcolor{blue}{0,1}}$. So $$F_2 = X_1^0 \oplus X_1^{2,3}X_2^{0,1}.$$

    The MVI-FPRM form of the first term, $X_1^0$, is 
    \begin{align*}
        X_1^0
        &= X_1^{0,1,2,3} \oplus X_1^{1,2,3} \\
        &= P_1^1 \oplus P_1^4 \\
        &= 1 \oplus P_1^4.
    \end{align*}

    And the MVI-FPRM form of the second term, $X_1^{2,3}X_2^{0,1}$ is 
    \begin{align*}
        X_1^{2,3}X_2^{0,1}
        &= X_1^{2,3}(X_2^{0,1,2} \oplus X_2^2) \\
        &= P_1^3(P_2^1 \oplus P_2^3) \\
        &= P_1^3(1 \oplus P_2^3) \\
        &= P_1^3 \oplus P_1^3P_2^3.
    \end{align*}

    Combining the MVI-FPRM forms of the two terms gives the same form for $F_2$ that was found in Example~\ref{ex:fprm two terms}: $1 \oplus P_1^3 \oplus P_1^4 \oplus P_1^3P_2^3$.

    \begin{figure}[!htb]
        \centering
        \includegraphics[width=0.6\linewidth]{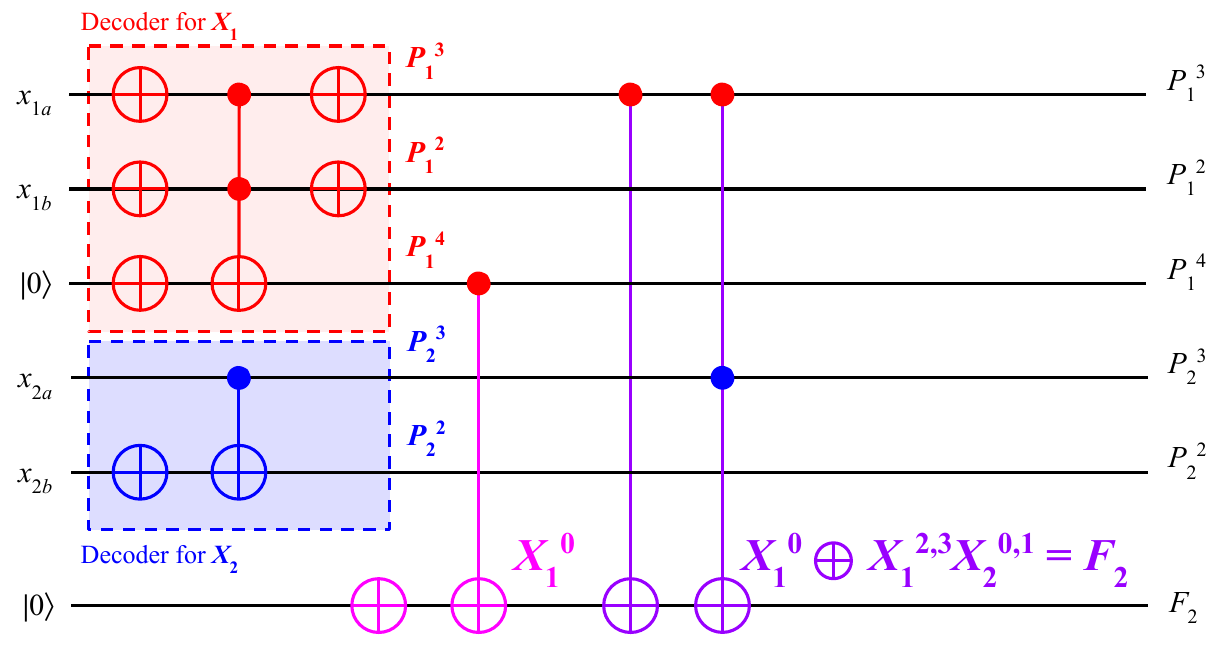}
        \caption{Circuit for $F_2$ realized from the terms $X_1^0$ and $X_1^{2,3}X_2^{0,1}$ from Example~\ref{ex:fprm two terms-changed terms}.}
        \label{fig:ex6-circuit}
    \end{figure}
    
    Constructing the circuit for $F_2$, shown in Fig.~\ref{fig:ex6-circuit}, consists of 7 NOT gates, 3 CNOT gates, and 2 3-bit Toffoli gates and has a cost of 20 (Maslov) and 157 (TQC), which is the same as the circuit from Fig.~\ref{fig:ex5-circuit-directly}. This example shows how the MVI-FPRM is a canonical form.
\end{example}

An important factor for the cost of the circuit is the polarity, which is demonstrated in Example~\ref{ex:fprm two terms-changed polarity}.

\begin{example}
    \label{ex:fprm two terms-changed polarity}
    Take the function $F_2 = X_1^{0,2,3}X_2^{0,1} \oplus X_1^0X_2^2 = X_1^0 \oplus X_1^{2,3}X_2^{0,1}$ from Example~\ref{ex:fprm two terms}. Let's find the MVI-FPRM form of $F_2$ with the polarities: 
    $$
    Q_1=
    \begin{bmatrix}
        1 & 1 & 1 & 1 \\
        1 & 0 & 0 & 0 \\
        0 & 1 & 1 & 0 \\
        0 & 0 & 1 & 1
    \end{bmatrix}
    , \
    Q_2=
    \begin{bmatrix}
        1 & 1 & 1 \\
        1 & 1 & 0 \\
        1 & 0 & 1
    \end{bmatrix}
    .
    $$

    The polarity literals for $X_1$ are $$Q_1^1=X_1^{0,1,2,3}=1, \ Q_1^2=X_1^0, Q_1^3=X_1^{1,2}, \ Q_1^4=X_1^{2,3}.$$ And the polarity literals for $X_2$ are $$Q_2^1=X_2^{0,1,2}=1, \ Q_2^2=X_2^{0,1}, \ Q_2^3=X_2^{0,2}.$$
    
    It can be found that the MVI-FPRM form of $F_2$ is 
    \begin{align*}
        F_2
        &= X_1^0 \oplus X_1^{2,3}X_2^{0,1} \\
        &= Q_1^2 \oplus Q_1^4Q_2^2. \\
    \end{align*}

    \begin{figure}[tbp]
        \centering
        \begin{subfigure}{0.35\linewidth}
            \centering
            \includegraphics[width=\linewidth]{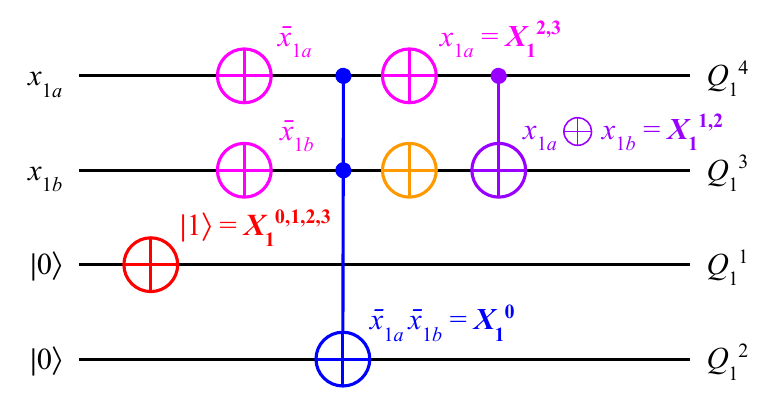}
            \caption{Decoder for $Q_1$.}
            \label{fig:ex7-x1-decoder}
        \end{subfigure}
        \begin{subfigure}{0.19\linewidth}
            \centering
            \includegraphics[width=\linewidth]{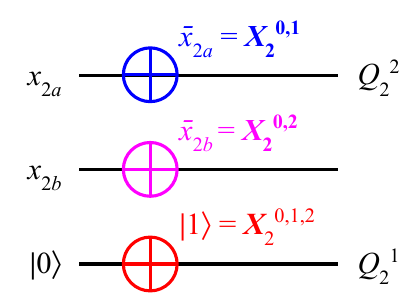}
            \caption{Decoder for $Q_2$.}
            \label{fig:ex7-x2-decoder}
        \end{subfigure}
        \\
        \begin{subfigure}{0.4\linewidth}
            \centering
            \includegraphics[width=\linewidth]{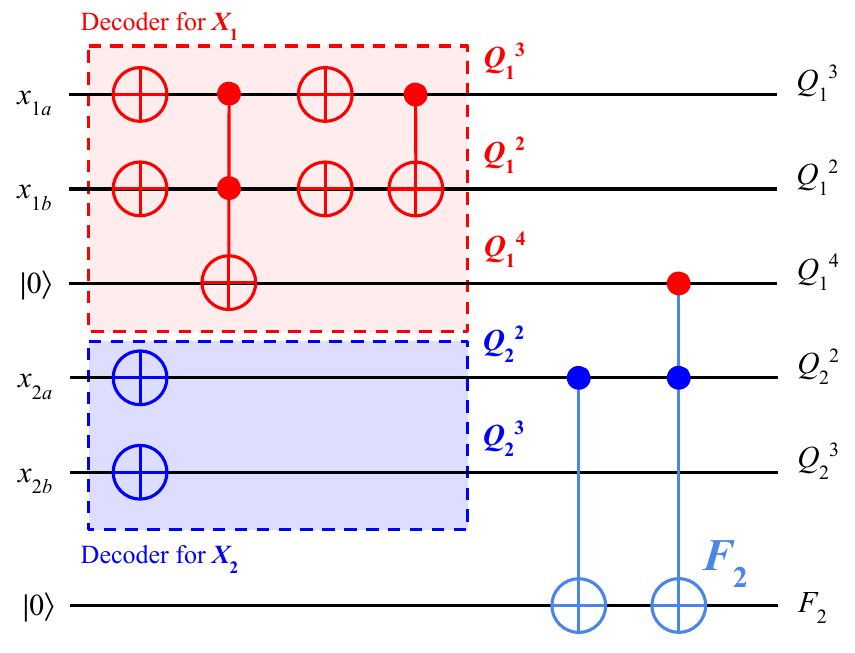}
            \caption{Circuit for $F_2$.}
            \label{fig:ex7-circuit}
        \end{subfigure}
        \caption{Decoders for $Q_1$ and $Q_2$, and circuit for $F_2$, for Example~\ref{ex:fprm two terms-changed polarity}.}
    \end{figure}

    The decoder for $X_1$ with polarity $Q_1$ is realized in Fig.~\ref{fig:ex7-x1-decoder}, the decoder for $X_2$ with polarity $Q_2$ is realized in Fig.~\ref{fig:ex7-x2-decoder}, and the full circuit realization of $F_2$ is shown in Fig.~\ref{fig:ex7-circuit}.

    The resulting circuit has 6 NOT gates, 2 CNOT gates, and 2 3-bit Toffoli gates with a cost of 18 (Maslov) or 142 (TQC).

    This shows that for the function $F_2$, the polarities $Q_1$ and $Q_2$ lead to a less costly circuit than the polarities $P_1$ and $P_2$. Thus, it is important to take into account what polarities to use depending on the function.
\end{example}

All the previous examples are single-output; however, this method is not restricted to single-output functions. This is presented in Example~\ref{ex:fprm multi-output}.

\begin{example}
    \label{ex:fprm multi-output}
    Let $F$ be a multi-output function which outputs $F_1$ (from Examples~\ref{ex:fprm} and \ref{ex:fprm spectrum}) and $F_2$ (from Example~\ref{ex:fprm two terms}) with polarities $Q_1$ and $Q_2$.

    The MVI-FPRM form of the function $F_1$ is
    \begin{align*}
        F_1
        &= X_1^{0,2,3}X_2^{0,1} \\
        &= (X_1^0 \oplus X_1^{2,3})(X_2^{0,1}) \\
        &= (Q_1^2 \oplus Q_1^4)(Q_2^2) \\
        &= Q_1^2Q_2^2 \oplus Q_1^4Q_2^2.
    \end{align*}

    The MVI-FPRM form of $F_2$ was found in Example~\ref{ex:fprm two terms-changed polarity} to be
    $$F_2 = Q_1^2 \oplus Q_1^4Q_2^2.$$

    The spectral coefficients of the two functions are listed in Table~\ref{tab:ex8-spectral-coeff}.

    \begin{table}[!htb]
        \centering
        \renewcommand{\arraystretch}{1.3}
        \resizebox{\textwidth}{!}{%
        \begin{tabular}{l|ccc ccc ccc ccc}
            Coefficient 
            & $M_{Q_1^1Q_2^1}$ & $M_{Q_1^1Q_2^2}$ & $M_{Q_1^1Q_2^3}$ 
            & $M_{Q_1^2Q_2^1}$ & $M_{Q_1^2Q_2^2}$ & $M_{Q_1^2Q_2^3}$ 
            & $M_{Q_1^3Q_2^1}$ & $M_{Q_1^3Q_2^2}$ & $M_{Q_1^3Q_2^3}$
            & $M_{Q_1^4Q_2^1}$ & $M_{Q_1^4Q_2^2}$ & $M_{Q_1^4Q_2^3}$
            \\[2.5px]
            \hline
            Value for \textcolor{magenta}{$F_1$}
            & 0 & 0 & 0
            & 0 & \textcolor{red}{1} & 0
            & 0 & 0 & 0
            & 0 & \textbf{\textcolor{magenta}{1}} & 0
            \\
            Value for \textcolor{brown}{$F_2$}
            & 0 & 0 & 0
            & \textcolor{brown}{1} & 0 & 0
            & 0 & 0 & 0
            & 0 & \textbf{\textcolor{brown}{1}} & 0
            \\
        \end{tabular}}
        \caption{Spectrum of the functions \textcolor{magenta}{$F_1$} and \textcolor{brown}{$F_2$} with polarities $Q_1$ and $Q_2$ for Example~\ref{ex:fprm multi-output}.}
        \label{tab:ex8-spectral-coeff}
    \end{table}

    The multi-output function $F$ can be realized as a circuit by first computing the shared terms, in this case the term $Q_1^4Q_2^2$, and then computing the remaining individual output functions, as illustrated in Fig.~\ref{fig:ex8-circuit}. Note that the term $Q_1^4Q_2^2$ is shared between $F_1$ and $F_2$.

    \begin{figure}[!htb]
        \centering
        \includegraphics[width=0.4\linewidth]{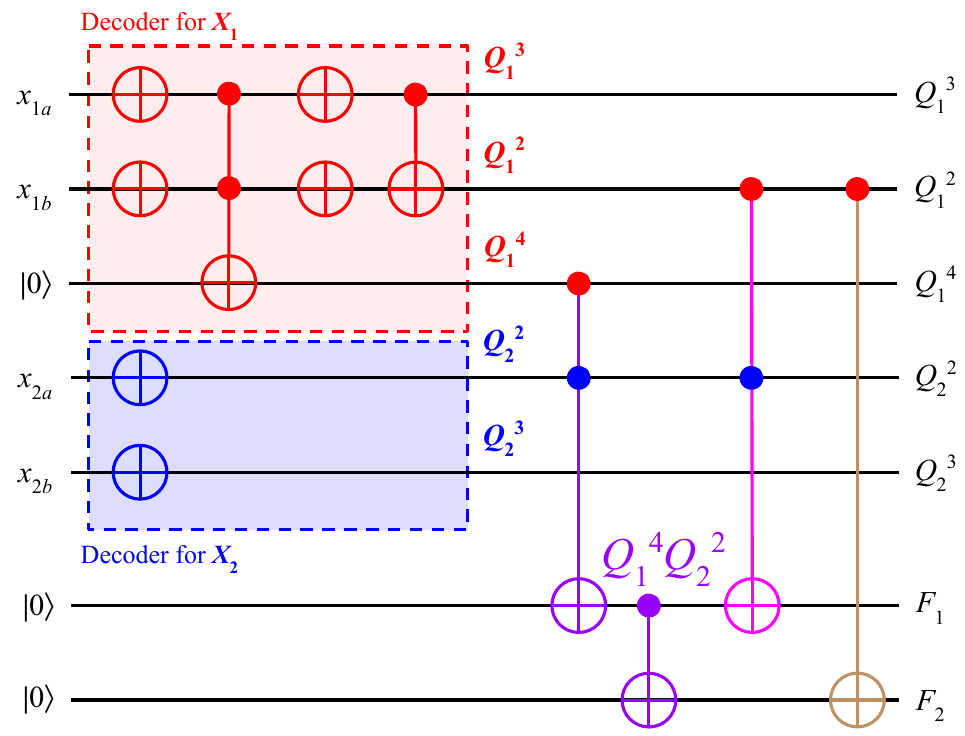}
        \caption{Circuit for $F$ with outputs $F_1$ and $F_2$ from Example~\ref{ex:fprm multi-output}.}
        \label{fig:ex8-circuit}
    \end{figure}
\end{example}

These methods can be applied to more complicated functions, such as the one in Example~\ref{ex:fprm three variables}. 

\begin{example}
    \label{ex:fprm three variables}
    Let $F_3=X_4^1X_5^{0,2} \oplus X_3^{1,2}X_4^{0,1}X_5^0 \oplus X_4^{0,2}X_5^{1,2} \oplus X_3^2X_4^1X_5^1$, where $X_3$, $X_4$, and $X_5$ are ternary variables. 
    
    Let the polarities be
    $$
    P_3=
    \begin{bmatrix}
        1 & 1 & 1 \\
        1 & 0 & 1 \\
        0 & 1 & 1 \\
    \end{bmatrix}
    , \
    P_4=
    \begin{bmatrix}
        1 & 1 & 1 \\
        1 & 1 & 0 \\
        0 & 1 & 0
    \end{bmatrix}
    , \
    P_5=
    \begin{bmatrix}
        1 & 1 & 1 \\
        1 & 1 & 0 \\
        0 & 1 & 1
    \end{bmatrix}
    .
    $$
    
    Thus, the polarity literals for $X_3$ are $$P_3^1=X_3^{0,1,2}=1, \ P_3^2=X_3^{0,2}, \ P_3^3=X_3^{1,2}.$$
    The polarity literals for $X_4$ are $$P_4^1=X_4^{0,1,2}=1, \ P_4^2=X_4^{0,1}, \ P_4^3=X_4^1.$$
    The polarity literals for $X_5$ are $$P_5^1=X_5^{0,1,2}=1, \ P_5^2=X_5^{0,1}, \ P_5^3=X_5^{1,2}.$$

    The MVI-FPRM form of the function $F_3$ can be found to be
    \begin{align*}
        F_3 
        &= P_3^3P_4^2 \oplus P_3^2P_5^3 \oplus P_4^3P_5^2 \\
        &= X_3^{1,2}X_4^{0,1} \oplus X_3^{0,2}X_5^{1,2} \oplus X_4^1X_5^{0,1}
    \end{align*}

    \begin{figure}[!htb]
        \centering
        \begin{subfigure}[!h]{0.21\linewidth}
            \centering
            \includegraphics[width=\linewidth]{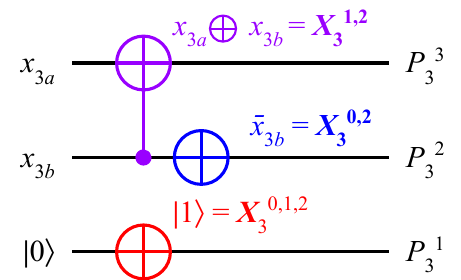}
            \caption{$P_3$.}
            \label{fig:ex10-x3-decoder}
        \end{subfigure}
        \begin{subfigure}[!h]{0.17\linewidth}
            \centering
            \includegraphics[width=\linewidth]{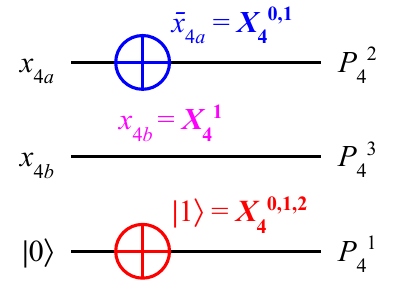}
            \caption{$P_4$.}
            \label{fig:ex10-x4-decoder}
        \end{subfigure}
        \begin{subfigure}[!h]{0.2\linewidth}
            \centering
            \includegraphics[width=\linewidth]{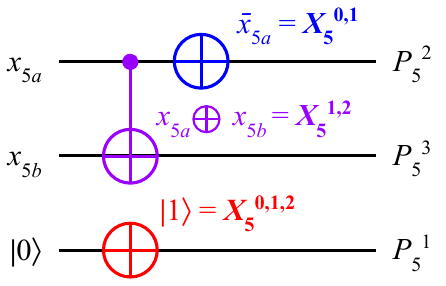}
            \caption{$P_5$.}
            \label{fig:ex10-x5-decoder}
        \end{subfigure}
        \caption{Decoders for $P_3$, $P_4$, and $P_5$ from Example~\ref{ex:fprm three variables}.}
    \end{figure}
    
    \begin{figure}[!htb]
    \centering
        \begin{subfigure}{0.3\linewidth}
            \centering
            \includegraphics[width=\linewidth]{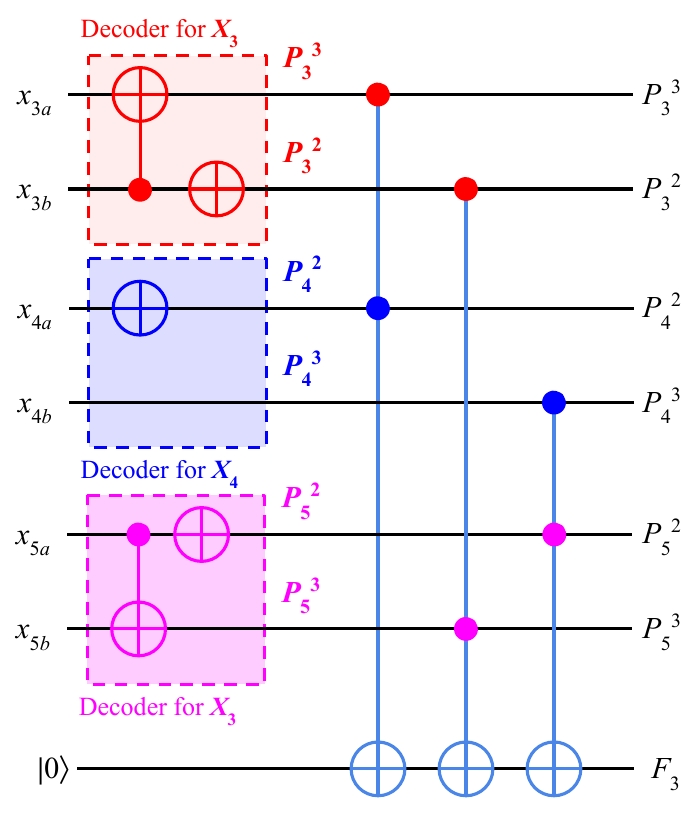}
            \caption{Circuit for $F_3$ based on MVI-FPRM with decoders.}
            \label{fig:ex10-circuit}
        \end{subfigure}
        \begin{subfigure}{0.34\linewidth}
            \centering
            \includegraphics[width=.6\linewidth]{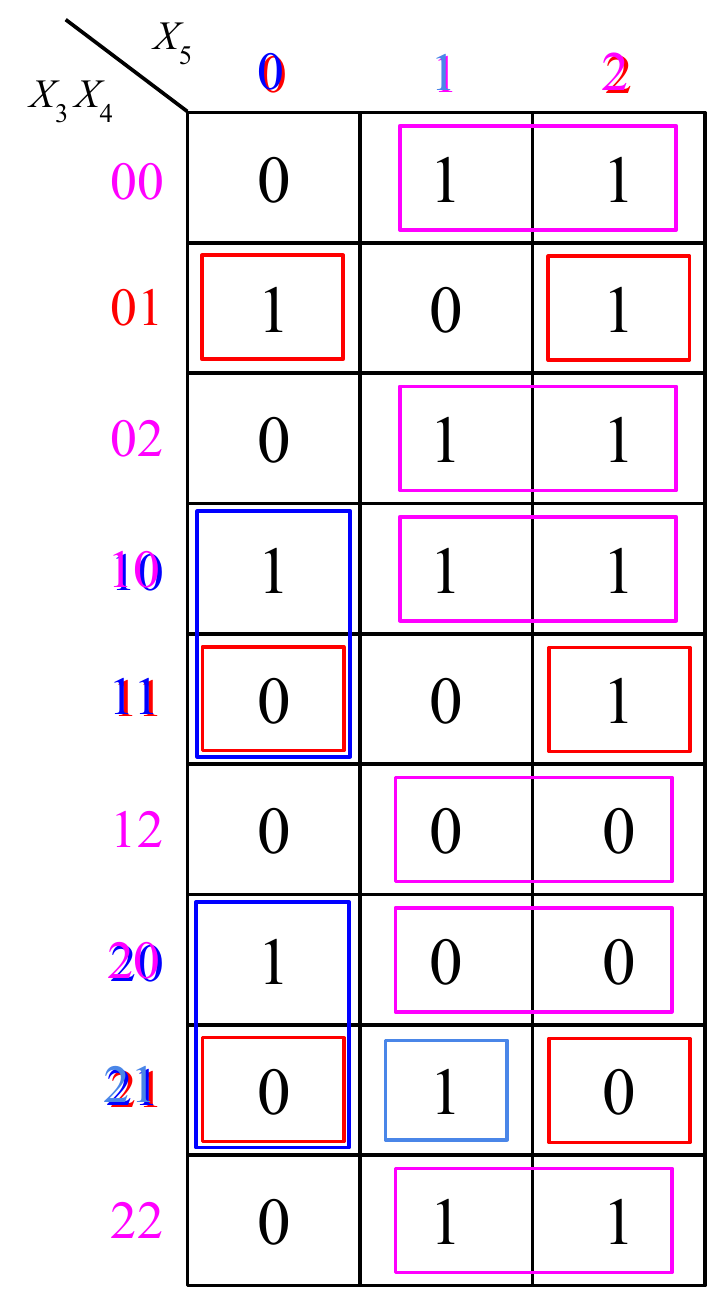}
            \caption{Marquand Chart from the original expression: $X_4^{\textcolor{red}{1}}X_5^{\textcolor{red}{0,2}} \oplus X_3^{\textcolor{blue}{1,2}}X_4^{\textcolor{blue}{0,1}}X_5^{\textcolor{blue}{0}} \oplus  X_4^{\textcolor{magenta}{0,2}}X_5^{\textcolor{magenta}{1,2}} \oplus X_3^{\textcolor{cornflowerblue}{2}}X_4^{\textcolor{cornflowerblue}{1}}X_5^{\textcolor{cornflowerblue}{1}}$.}
            \label{fig:ex10-marquand}
        \end{subfigure}
        \begin{subfigure}{0.34\linewidth}
            \centering
            \includegraphics[width=0.6\linewidth]{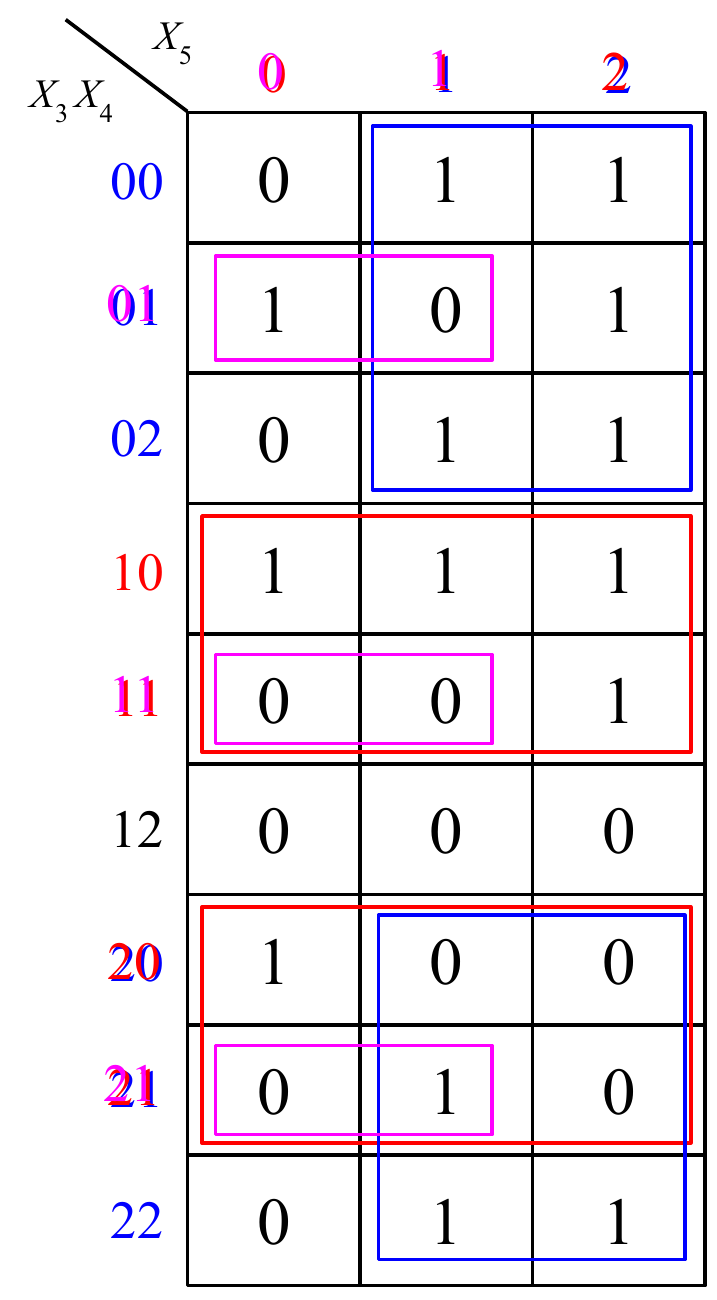}
            \caption{Marquand Chart from the MVI-FPRM: $X_3^{\textcolor{red}{1,2}}X_4^{\textcolor{red}{0,1}} \oplus X_3^{\textcolor{blue}{0,2}}X_5^{\textcolor{blue}{1,2}} \oplus X_4^{\textcolor{magenta}{1}}X_5^{\textcolor{magenta}{0,1}}$.}
            \label{fig:ex10-f-chart}
        \end{subfigure}
        \caption{Circuit and Marquand chart for $F_3$ from Example~\ref{ex:fprm three variables}.}
    \end{figure}

    The decoders for $P_3$, $P_4$, and $P_5$ are realized in Figs.~\ref{fig:ex10-x3-decoder}, \ref{fig:ex10-x4-decoder}, and \ref{fig:ex10-x5-decoder}. The full circuit for $F_3$ is illustrated in Fig.~\ref{fig:ex10-circuit}. The function $F_3$ is shown on the Marquand charts in Fig.~\ref{fig:ex10-marquand} and Fig.~\ref{fig:ex10-f-chart}. The rows and columns of a Marquand chart for over two variables are created by listing the possible values of $X_3$,$X_4$,$X_5$ in natural order.

    The resulting circuit, which contains 2 NOT gates, 2 CNOT gates, and 3 3-bit Toffoli gates, has a Maslov cost of 19 and TQC of 192.

    If, instead, the circuit was realized directly from the binary ESOP expression, then the result would be the circuit in Fig.~\ref{fig:ex10-circuit-esop}, which contains 3 NOT gates, 2 CNOT gates, 1 3-bit Toffoli gate, and 5 4-bit Toffoli gates. This leads to a Maslov cost of 75 and and TQC 630, which is more expensive than the MVI-FPRM-based circuit. 
    \begin{align*}
        F_3
        &= X_4^1X_5^{0,2} \oplus X_3^{1,2}X_4^{0,1}X_5^0 \oplus X_4^{0,2}X_5^{1,2} \oplus X_3^2X_4^1X_5^1 \\
        &= (x_{4b})(\bar{x}_{5b}) \oplus (x_{3a} \oplus x_{3b})(\bar{x}_{4a})(\bar{x}_{5a} \oplus x_{5b}) 
        \oplus (\bar{x}_{4b})(x_{5a} \oplus x_{5b}) \oplus (x_{3a})(x_{4b})(x_{5b}) \\
        &= x_{4b}\bar{x}_{5b} \oplus
        x_{3a}\bar{x}_{4a}\bar{x}_{5a} \oplus x_{3a}\bar{x}_{4a}x_{5b} \oplus
        x_{3b}\bar{x}_{4a}\bar{x}_{5a} \oplus x_{3b}\bar{x}_{4a}x_{5b}
        \oplus \bar{x}_{4b}x_{5a} \oplus \bar{x}_{4b}x_{5b} \oplus x_{3a}x_{4b}x_{5b} \\
        &= x_{4b} \oplus x_{5b} \oplus
        x_{3a}\bar{x}_{4a}\bar{x}_{5a} \oplus x_{3a}\bar{x}_{4a}x_{5b} \oplus
        x_{3b}\bar{x}_{4a}\bar{x}_{5a} \oplus x_{3b}\bar{x}_{4a}x_{5b}
        \oplus \bar{x}_{4b}x_{5a} \oplus x_{3a}x_{4b}x_{5b}
    \end{align*}

    \begin{figure}[!htb]
        \centering
        \includegraphics[width=0.4\linewidth]{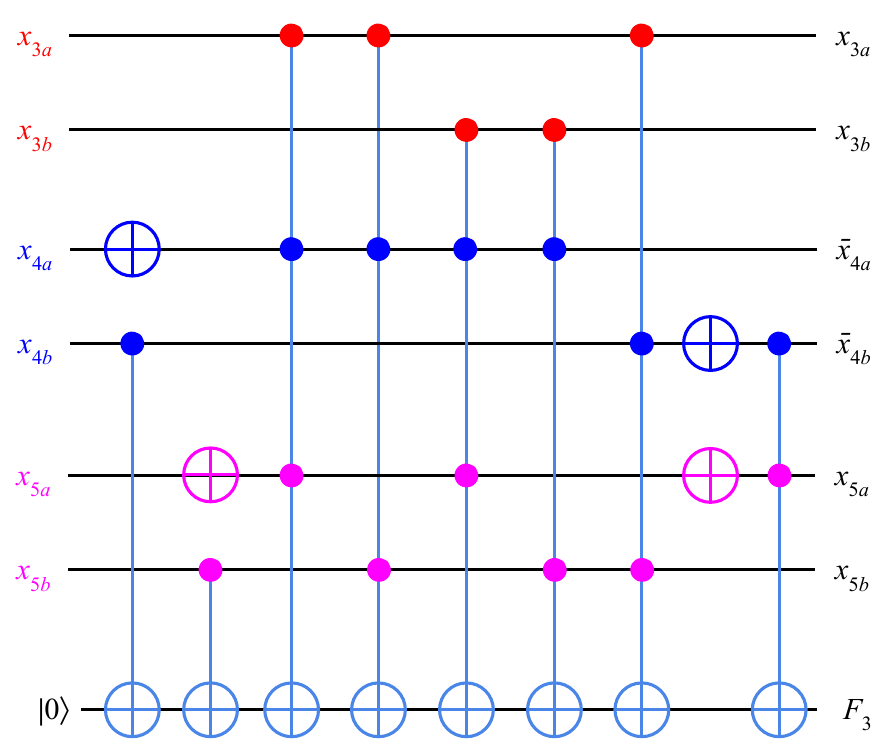}
        \caption{ESOP-based circuit for $F_3$ from Example~\ref{ex:fprm three variables}.}
        \label{fig:ex10-circuit-esop}
    \end{figure}
\end{example}

The full process starting from binary variables is demonstrated in Example~\ref{ex:fprm full}.

\begin{example}
    \label{ex:fprm full}
    Let $F_4=x_ax_b \oplus \bar{x}_ax_bx_d \oplus x_cx_d\bar{x}_e \oplus \bar{x}_c\bar{x}_d\bar{x}_f$, where $x_a$, $x_b$, $x_c$, $x_d$, $x_e$, and $x_f$ are binary variables.

    The first step is to pair the binary variables to create multi-valued variables. Let the quaternary variable $X_6$ be encoded by $x_a$ and $x_b$, let the quaternary variable $X_7$ be encoded by $x_c$ and $x_d$, and let the quaternary variable $X_8$ be encoded by $x_e$ and $x_f$.

    Let the polarity be $P_6$,$P_7$,$P_8$.
    \small{
    $$
    P_6=
    \begin{bmatrix}
        1 & 1 & 1 & 1 \\
        0 & 0 & 1 & 0 \\
        0 & 0 & 0 & 1 \\
        0 & 1 & 0 & 1
    \end{bmatrix}
    , \
    P_7=
    \begin{bmatrix}
        1 & 1 & 1 & 1 \\
        1 & 0 & 0 & 0 \\
        0 & 0 & 0 & 1 \\
        0 & 1 & 0 & 1
    \end{bmatrix}
    , \
    P_8=
    \begin{bmatrix}
        1 & 1 & 1 & 1 \\
        1 & 1 & 0 & 0 \\
        1 & 0 & 1 & 0 \\
        0 & 1 & 1 & 1
    \end{bmatrix}
    .
    $$
    }

    The next step is to rewrite $F_4$ in terms of the polarity literals (with $P_6^1=P_7^1=P_8^1=1$).
    \begin{align*}
        F_4
        &= x_ax_b \oplus \bar{x}_ax_bx_d \oplus x_cx_d\bar{x}_e \oplus \bar{x}_c\bar{x}_d\bar{x}_f \\
        &= X_6^3 \oplus X_6^2X_7^{1,3} \oplus X_7^3X_8^{0,1} \oplus X_7^0X_8^{0,2} \\
        &= P_6^3 \oplus P_6^2P_7^4 \oplus P_7^3P_8^2 \oplus P_7^2P_8^3
    \end{align*}

    \begin{figure}[!htb]
        \centering
        \begin{subfigure}{0.3\linewidth}
            \centering
            \includegraphics[width=\linewidth]{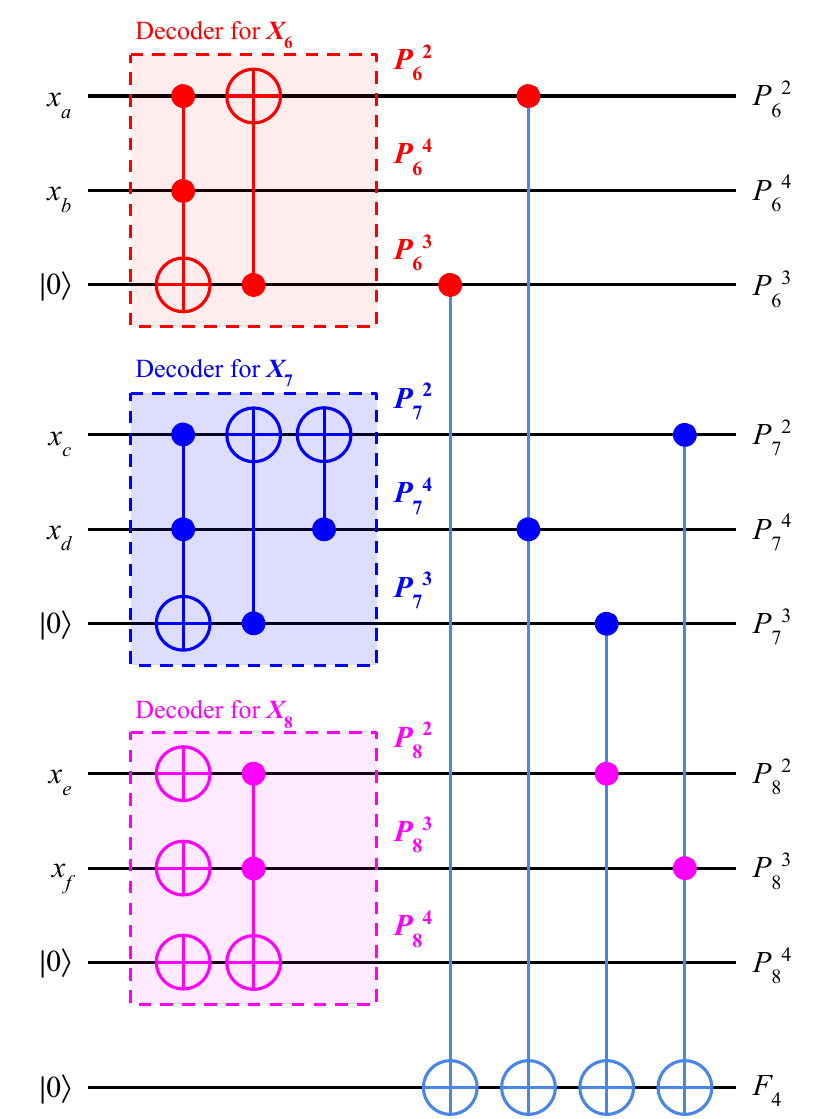}
            \caption{MVI-FPRM-based circuit.}
            \label{fig:ex11-circuit}
        \end{subfigure}
        \hspace{0.05\linewidth}
        \begin{subfigure}{0.3\linewidth}
            \centering
            \includegraphics[width=\linewidth]{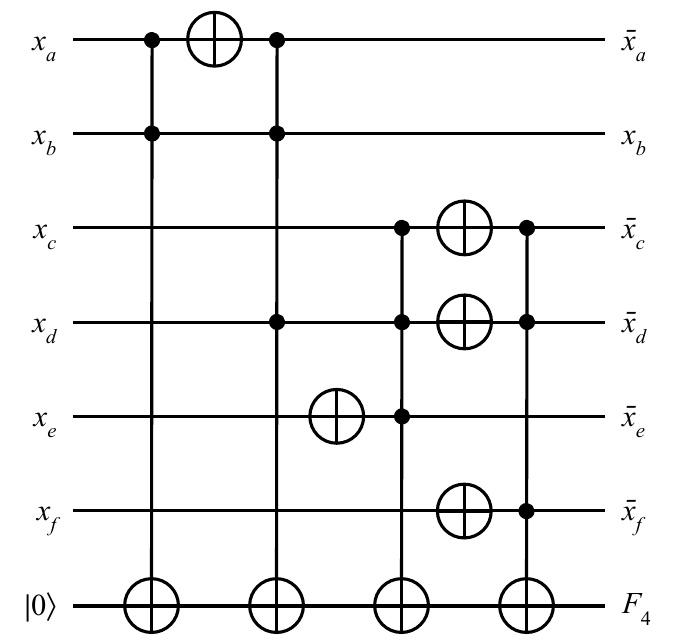}
            \caption{ESOP-based circuit.}
            \label{fig:ex11-circuit-esop}
        \end{subfigure}
        \caption{Circuits for $F_4$ from Example~\ref{ex:fprm full}.}
    \end{figure}

    The function $F_4$ is realized with the circuit in Fig.~\ref{fig:ex11-circuit}. This circuit consists of 3 NOT gates, 4 CNOT gates, and 6 3-bit Toffoli gates, and thus the circuit has a Maslov cost of 37 and a TQC of 383.

    If the function $F_4$ was realized using the ESOP instead, then the resulting circuit, shown in Fig.~\ref{fig:ex11-circuit-esop}, would instead consist of 5 NOT gates, 1 3-bit Toffoli gate, and 3 4-bit Toffoli gates, and thus the circuit has a Maslov cost of 49 and a TQC of 386, making it more expensive than the MVI-FPRM-based circuit.
    
\end{example}

In summary, MVI functions can be transformed into a form analogous to the FPRM form from classical Reed-Muller logic, where each variable has a fixed polarity. In the case of MVI variables, the polarities are not just positive and negative, and instead there are many more options, as explained in Section~\ref{mvi polarity}.

Polarities can be realized beforehand in a circuit by using decoders, as presented in Section~\ref{decoders}, which are placed at the start of a circuit, and the outputs, the polarity literals, can be used to realize all possible functions. The cost of the decoders becomes negligible as the size of the circuit increases.

The main factors, which have been demonstrated by the examples in this section, that determine the cost of realizing an FPRM, are:
\begin{itemize}
    \item Polarity: This affects the cost of the decoders and how many gates are needed to realize the terms. It is important to choose a suitable polarity for a function.
    \item Pairings: The specific binary variables that are paired to encode the multi-valued variables. There also may be more than two binary variables that encode a multi-valued variable (e.g., three binary variables can encode an 8-valued variable).
\end{itemize}

Comparisons between MVI-FPRM-based and ESOP-based circuits for $F_3$ and $F_4$ are summarized in Table~\ref{tab:comparison-fprm}.

\begin{table}[!htb]
    \centering
    \renewcommand{\arraystretch}{1.3}
    \resizebox{\linewidth}{!}{
    \begin{tabular}{c||cccc|cc||cccc|cc}
          & \multicolumn{6}{c||}{MVI-FPRM-Based} & \multicolumn{6}{c}{ESOP-Based}  \\
         \cline{2-13}
         Function & NOT & CNOT & 3-bit Toffoli & 4-bit Toffoli & Maslov Cost & TQC & NOT & CNOT & 3-bit Toffoli & 4-bit Toffoli & Maslov Cost & TQC \\
         \hline
         \hline
         $F_3$ & 2 & 2 & 3 & 0 & 19 & 192 & 3 & 2 & 1 & 5 & 75 & 630 \\
         $F_4$ & 3 & 4 & 6 & 0 & 37 & 383 & 5 & 0 & 1 & 3 & 49 & 386
    \end{tabular}}
    \caption{MVI-FPRM-based vs ESOP-based circuits for $F_3$ and $F_4$.}
    \label{tab:comparison-fprm}
\end{table}

It is important to remember that MVI functions use multi-valued variables as a representation of multiple binary variables, and all the circuits and operations are binary, not multi-valued.

\section{Products-Matching Method}
\label{products-matching}
We introduce a products-matching method to find the exact minimum MVI-FPRM form and the corresponding quantum circuit. 

The method for finding the MVI-FPRM form includes two stages. In the first stage, every multi-valued literal is transformed to a polarity that is specified by the chosen linearly independent polarity matrix. In the second stage, the terms consisting of transformed literals are used to calculate the final MVI-FPRM form. The code for the transformation of a multiple-valued literal is selected in such a way that the second stage of the transformation does not depend on the polarity chosen for the literal. This approach requires the introduction of normalized codes.

\subsection{Transformation of a Multi-Valued Input Literal}
\label{products-matching-mvi-literal}
The basic steps of the method for the transformation of a multiple-valued literal to its representation of polarity literals in the normalized code will be illustrated in Example~\ref{ex:products-matching-literal}. 

\begin{example}
    \label{ex:products-matching-literal}
    The goal is to find the normalized code of the literal $X^0$ (where $X$ is quaternary) with the polarity $P$.
    $$
    P =
    \begin{bmatrix}
        1 & 1 & 1 & 1 \\
        0 & 1 & 0 & 1\\
        0 & 0 & 1 & 1\\
        0 & 1 & 1 & 1
    \end{bmatrix}
    $$

    \begin{table}[!htbp]
        \centering
        \renewcommand{\arraystretch}{1.15}
        \begin{tabular}{c|c||c|c}
             MVI Literal & Binary Code & Polarity Literal & Normalized Code \\
             \hline
             $X^0$ & 1000 & $P^1$ & 1000 \\
             $X^1$ & 0100 & $P^2$ & 0100 \\ 
             $X^2$ & 0010 & $P^3$ & 0010 \\
             $X^3$ & 0001 & $P^4$ & 0001 \\
        \end{tabular}
        \caption{Binary and normalized codes for Example~\ref{ex:products-matching-literal}. The binary code is based on the set of truth values of the MVI literal, while the normalized coder is based on the polarity literals.}
        \label{tab:ex11-binary-and-normalized-codes}
    \end{table}
    
    The binary and normalized codes are listed in Table~\ref{tab:ex11-binary-and-normalized-codes}. The literal $X^0$ can be represented by the binary code 1000. $X^0$ can also be expressed as $P^1 \oplus P^4$, which has the normalized code 1001. The new code representing this combination of polarity literals can be derived by the bit-by-bit XOR operation of the code representation of the polarity literals. The normalized code has a `1’ in its bit representation corresponding to the index of the polarity literal $P^r$.

    \begin{table}[!htbp]
        \centering
        \renewcommand{\arraystretch}{1.15}
        \begin{tabular}{c|c|c}
             MVI Literal & XOR of Polarity Literals & Normalized Code \\
             \hline
             $X^0$ & $P^1 \oplus P^4$ & 1001 \\
             $X^1$ & $P^3 \oplus P^4$ & 0011 \\ 
             $X^2$ & $P^2 \oplus P^4$ & 0101 \\
             $X^3$ &$P^2 \oplus P^3 \oplus P^4$ & 0111 \\
        \end{tabular}
        \caption{MVI literals and their respective XOR expression with polarity literals and normalized codes, for Example~\ref{ex:products-matching-literal}.}
        \label{tab:ex11-normalized-code}
    \end{table}

    The MVI literals $X^0$, $X^1$, $X^2$, and $X^3$ are expressed in terms of the polarity literals and represented by a normalized code in Table~\ref{tab:ex11-normalized-code}. To find the normalized code of any MVI literal $X^S$, apply a bit-by-bit XOR operation on the normalized codes of all the literals $X^r$ such that $r \in S$. For instance, the normalized code for $X^{0,1}$ is 1010, found by the bit-by-bit XOR of the codes 1001 (for $X^0$) and 0011 (for $X^1$).

    \begin{table}[!htbp]
        \centering
        \renewcommand{\arraystretch}{1.15}
        \resizebox{\linewidth}{!}{
        \begin{tabular}{ll||c|c|c|c|c}
             Literals & & $P^1$ & $P^2$ & $P^3$ & $P^4$ & $P^1 \oplus P^2$ \\
             Binary Code & & 1111 & 0101 & 0011 & 0111 & 1010 \\
             \hline
             $X^{0,2}$ & 1010 & & & & & 1 \\
             $X^{1,3}$ & 0101 & & 1 & & \\
             $X^{2,3}$ & 0011 & & & 1 & \\
             $X^{1,2,3}$ & 0111 & & & & 1 & \\
             $X^{0,1,2,3}$ & 1111 & 1 & & & \\
             \hline
             \hline

            Literals & & $P^1 \oplus P^3$ & $P^1 \oplus P^4$ & $P^2 \oplus P^3$ & $P^2 \oplus P^4$ & $P^3 \oplus P^4$ \\
            Binary Code & & 1100 & 1000 & 0110 & 0010 & 0100 \\
            \hline
            $X^0$ & 1000 & & 1 & & \\
            $X^{0,1}$ & 1100 & 1 & & & \\
            $X^{1,2}$ & 0110 & & & 1 & \\
            $X^2$ & 0010 & & & & 1 & \\
            $X^1$ & 0100 & & & & & 1 \\

             \hline
             \hline
             Literals & & $P^1 \oplus P^2 \oplus P^3$ & $P^1 \oplus P^2 \oplus P^4$ & $P^1 \oplus P^3 \oplus P^4$ & $P^2 \oplus P^3 \oplus P^4$ & $P^1 \oplus P^2 \oplus P^3 \oplus P^4$ \\
             Binary Code & & 1001 & 1101 & 1011 & 0001 & 1110 \\
             \hline
             $X^3$ & 0001 & & & & 1 & \\
             $X^{0,3}$ & 1001 & 1 & & & & \\
             $X^{0,1,2}$ & 1110 & & & & & 1 \\
             $X^{0,1,3}$ & 1101 & & 1 & & & \\
             $X^{0,2,3}$ & 1011 & & & 1 & & \\
        \end{tabular}}
        \caption{Table for finding the MVI literals in terms of polarity literals using normalized codes, from Example~\ref{ex:products-matching-literal}.}
        \label{tab:ex-products-matching-literal}
    \end{table}
\end{example}

The algorithm for transforming an MVI literal into its normalized code is described as follows, and is demonstrated in Table~\ref{tab:ex-products-matching-literal}:
\begin{enumerate}
    \item Generate all possible XOR combinations of the polarity literals $P^r$ for later comparison with the original MVI literal (first row).
    \item Compare the binary representation of the MVI literal with the binary representation of the XOR combination (second row and second column). If these two binary representations are equal, assign to the multiple-valued literal its normalized code (shown by marking a `1' in the cell).
\end{enumerate}

\subsection{Transformation of a Multi-Valued Input Function}
\label{products-matching-mvi-function}
After the transformation of every original multiple-valued literal, as described in Section~\ref{products-matching-mvi-literal}, the entire set of multiple-valued terms of the function has to be changed to the MVI-FPRM form. In our method, every product term of the input function is compared to the normalized code of the indices of all spectral coefficients.

The basic steps of the algorithm for the products-matching method are explained in Example~\ref{ex:products-matching-mvi-function}.

\begin{example}
    \label{ex:products-matching-mvi-function}
    Let $G$ be a function of the ternary variable $X_1$, quaternary variable $X_2$, and ternary variable $X_3$. For the MVI-FPRM form of $G$, the variables have the polarities $P_1$, $P_2$, and $P_3$, respectively.

    \begin{table}[!htbp]
        \centering
        \renewcommand{\arraystretch}{1.3}
        \resizebox{\linewidth}{!}{%
        \begin{tabular}{c|ccccccc}
             $M$ & $M_{P_1^1P_2^1P_3^1}$ &
             $M_{P_1^1P_2^1P_3^2}$ &
             $M_{P_1^1P_2^1P_3^3}$ &
             $M_{P_1^1P_2^2P_3^1}$ &
             $M_{P_1^1P_2^2P_3^2}$ &
             $\cdots$ &
             $M_{P_1^3P_2^4P_3^3}$ \\[2.5px]
             \hline
             Code &
             100 1000 100 & 
             100 1000 010 & 
             100 1000 001 & 
             100 0100 100 & 
             100 0100 010 & 
             $\cdots$ &
             001 0001 001
        \end{tabular}
        }
        \caption{The normalized code for the complete MVI-FPRM spectrum of the function $G$ from Example~\ref{ex:products-matching-mvi-function}.}
        \label{tab:ex-spectrum-code}
    \end{table}
    
    The normalized codes corresponding to each spectral coefficient are stated in Table~\ref{tab:ex-spectrum-code}. It can be observed that the code for each spectral coefficient consists of a combination of the codes for three polarity literals, in which every polarity literal is taken from a different original multi-valued variable.
\end{example}
    
The final value of the spectral coefficient $M_{P_1^{r_1}P_2^{r_2}P_3^{r_3}}$, determined by a column from Table~\ref{tab:ex-spectrum-code}, is obtained by comparing the normalized codes of all product terms of the Boolean function with the variables $X_1, X_2, \dots, X_n$ represented with normalized codes. The value $M_{P_1^{r_1}P_2^{r_2}P_3^{r_3}}$ is the representation of the output function, if the intersection (bit-by-bit AND operation) of the codes is not empty (no part of the code for a variable is all 0's). To obtain the final value $M_{P_1^{r_1}P_2^{r_2}P_3^{r_3}}$ for all terms of the function, the XOR operation has to be performed for all values $M_{P_1^{r_1}P_2^{r_2}P_3^{r_3}}$.

The final MVI-FPRM form includes terms that are obtained by replacing the polarity literals in the indices of the spectral coefficients, which have a value that is not 0, with their binary representation.

To summarize, the procedure for obtaining MVI-FPRM of a multi-output Boolean function can be described by the following algorithm:
\begin{enumerate}
    \item Transform all the MVI literals of the function to their normalized codes according to the chosen polarities for the variables (algorithm in Section~\ref{products-matching-mvi-literal}).

    \item Calculate the MVI-FPRM spectrum for the normalized codes.
    
    \item Replace the polarity literals of non-zero spectral coefficients by their original MVI literal. The output functions for the terms are given by the values of the spectral coefficients.
\end{enumerate}

As one can observe, each spectral coefficient can be calculated in turn, which allows one to solve larger problems. In addition, every expression created by the program can be factorized and transformed to a form of MVI-GRM, which reduces the cost of the resulting quantum circuit.

\section{MVI-FPRM-Based 2-Bit Adder}
\label{adder}

MVI-FPRM forms can be used practically to create quantum adders. Assume the binary variables $x_a$, $x_b$, $x_c$, $x_d$. The variables $x_a$ and $x_b$ are paired to create the quaternary variable $X_1$, and $x_c$ and $x_d$ are paired to create $X_2$, as stated in Table~\ref{tab:adder-truth-table}. The functions $f_c$, $f_0$, and $f_1$ are represented by the Marquand charts in Fig.~\ref{fig:adder-fc}, Fig.~\ref{fig:adder-f0}, and Fig.~\ref{fig:adder-f1}, respectively.

\begin{table}[!htbp]
    \centering
    \renewcommand{\arraystretch}{1.1}
    \begin{tabular}{cc|cc|c||cc|cc|c}
         $x_a$, $x_b$ & $X_1$ & $x_c$, $x_d$ & $X_2$ & $f_cf_0f_1$ & $x_a$, $x_b$ & $X_1$ & $x_c$, $x_d$ & $X_2$ & $f_cf_0f_1$ \\
         \hline
         00 & 0 & 00 & 0 & 000 & 
         10 & 2 & 00 & 0 & 010 \\
         00 & 0 & 01 & 1 & 001 & 
         10 & 2 & 01 & 1 & 011 \\
         00 & 0 & 10 & 2 & 010 & 
         10 & 2 & 10 & 2 & 100 \\
         00 & 0 & 11 & 3 & 011 & 
         10 & 2 & 11 & 3 & 101 \\
         01 & 1 & 00 & 0 & 001 &
         11 & 3 & 00 & 0 & 011 \\
         01 & 1 & 01 & 1 & 010 &
         11 & 3 & 01 & 1 & 100 \\
         01 & 1 & 10 & 2 & 011 &
         11 & 3 & 10 & 2 & 101 \\
         01 & 1 & 11 & 3 & 100 &
         11 & 3 & 11 & 3 & 110 \\
    \end{tabular}
    \caption{Truth table for a 2-bit MVI adder.}
    \label{tab:adder-truth-table}
\end{table}
\begin{figure}[!htbp]
    \centering
    \begin{subfigure}{0.25\linewidth}
        \centering
        \includegraphics[width=\linewidth]{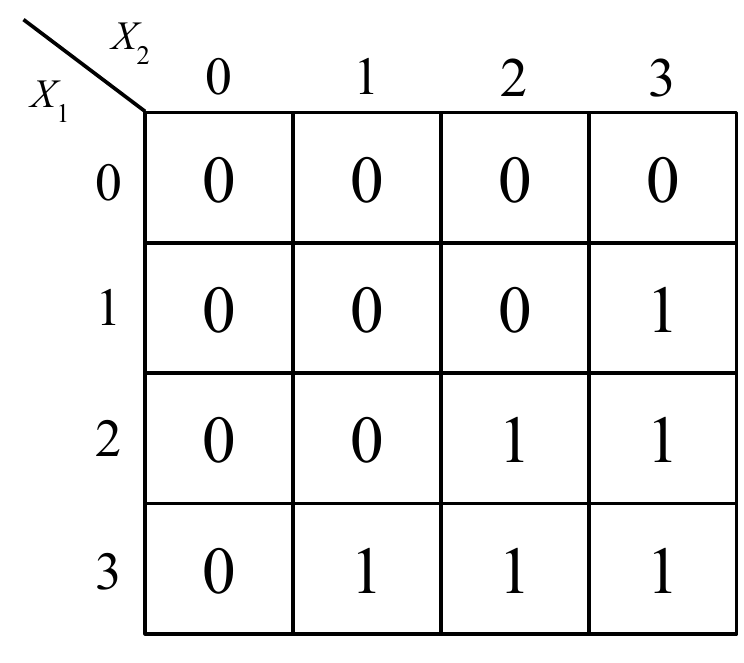}
        \caption{$f_c$.}
        \label{fig:adder-fc}
    \end{subfigure}
    \begin{subfigure}{0.25\linewidth}
        \centering
        \includegraphics[width=\linewidth]{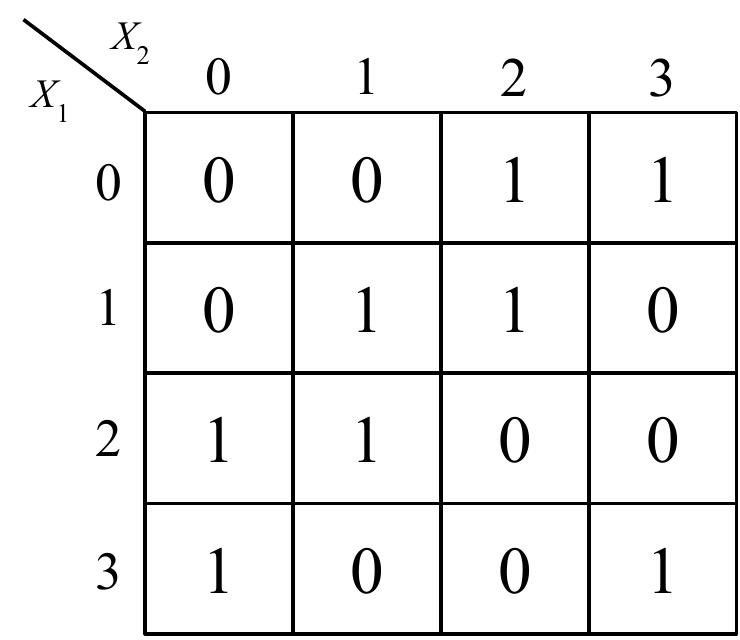}
        \caption{$f_0$.}
        \label{fig:adder-f0}
    \end{subfigure}
    \begin{subfigure}{0.25\linewidth}
        \centering
        \includegraphics[width=\linewidth]{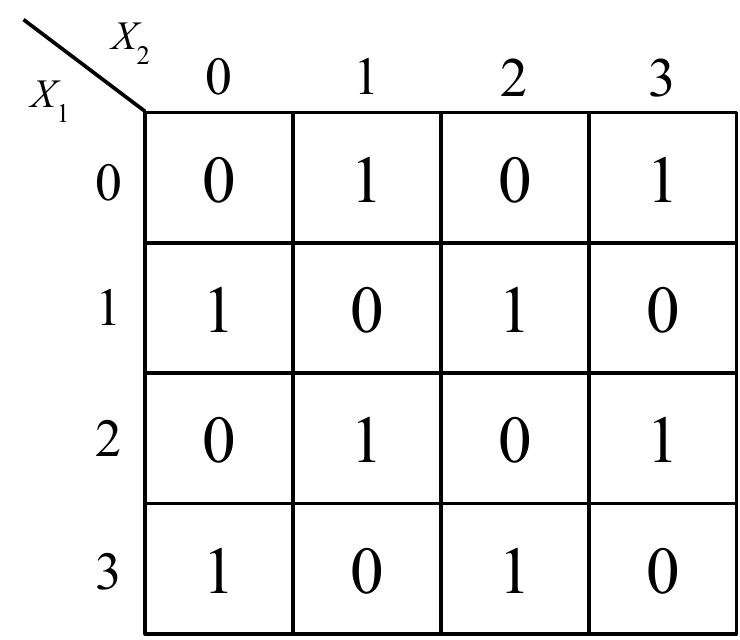}
        \caption{$f_1$.}
        \label{fig:adder-f1}
    \end{subfigure}
    \caption{Marquand charts for $f_c$, $f_0$, and $f_1$.}
    \end{figure}

There are many possible polarities for realizing the adder, and this was only one possible pairing. This is a difficult search problem, since there are $6 \cdot 480^2$ possible pairings and polarities. This section will cover two possible solutions. The intention here is to present the cost differences between various solutions.

\subsection{First Possible 2-Bit Adder}
\label{adder-1}

Let $X_1$ and $X_2$ have the polarities $P_1$ and $P_2$ respectively.
$$
P_1 = 
\begin{bmatrix}
    1 & 1 & 1 & 1 \\
    0 & 1 & 0 & 1 \\
    0 & 0 & 1 & 0 \\
    1 & 1 & 0 & 0
\end{bmatrix}
, \
P_2 = 
\begin{bmatrix}
    1 & 1 & 1 & 1 \\
    0 & 1 & 0 & 1 \\
    0 & 0 & 1 & 0 \\
    1 & 1 & 0 & 0
\end{bmatrix}
$$

\begin{table}[!htbp]
    \centering
    \renewcommand{\arraystretch}{1.1}
    \resizebox{\linewidth}{!}{%
    \begin{tabular}{cc|c|c|c||cc|c|c|c}
         $X_1$ & $X_2$ & Binary Code & Normalized Code & $f_cf_0f_1$ & $X_1$ & $X_2$ & Binary Code & Normalized Code & $f_cf_0f_1$ \\
         \hline
         0 & 0 & 1000 1000 & 1110 1110 & 000 &
         2 & 0 & 0010 1000 & 0010 1110 & 010 \\
         0 & 1 & 1000 0100 & 1110 1111 & 001 &
         2 & 1 & 0010 0100 & 0010 1111 & 011 \\
         0 & 2 & 1000 0010 & 1110 0010 & 010 &
         2 & 2 & 0010 0010 & 0010 0010 & 100 \\
         0 & 3 & 1000 0001 & 1110 1011 & 011 &
         2 & 3 & 0010 0001 & 0010 1011 & 101 \\
         1 & 0 & 0100 1000 & 1111 1110 & 001 &
         3 & 0 & 0001 1000 & 1011 1110 & 011 \\
         1 & 1 & 0100 0100 & 1111 1111 & 010 &
         3 & 1 & 0001 0100 & 1011 1111 & 100 \\
         1 & 2 & 0100 0010 & 1111 0010 & 011 &
         3 & 2 & 0001 0010 & 1011 0010 & 101 \\
         1 & 3 & 0100 0001 & 1111 1011 & 100 &
         3 & 3 & 0001 0001 & 1011 1011 & 110 \\
    \end{tabular}}
    \caption{Normalized codes for a 2-bit MVI adder with the polarities $P_1$ and $P_2$.}
    \label{tab:adder-codes}
\end{table}

The products-matching method can be applied to obtain an MVI-FPRM form, which then can be realized as a binary quantum circuit.

\textbf{Step 1:} \textit{Transform all the MVI literals of the function to their normalized codes according to the chosen polarities for the variables.}

The literals are $X_i^0=P_i^1 \oplus P_i^2 \oplus P_i^3$ (normalized code 1110), $X_i^1=P_1^1 \oplus P_i^2 \oplus P_i^3 \oplus P_i^4$ (normalized code 1111), $X_i^2=P_1^3$ (normalized code 0010), and $X_i^3=P_1^1 \oplus P_i^3 \oplus P_i^4$ (normalized code 1011). The codes for all the potential inputs are listed in Table~\ref{tab:adder-codes}.

\textbf{Step 2:} \textit{Calculate the MVI-FPRM spectrum for the normalized codes.}

The normalized code obtained in the previous step is now compared with all the indices of the spectral coefficients of the general spectrum for the adder functions.

\begin{table}[!htb]
    \centering
    \renewcommand{\arraystretch}{1.1}
    \resizebox{\linewidth}{!}{%
    \begin{tabular}{l||c|c|c|c|c|c|c|c}
        Term & $M_{P_1^1P_2^1}$ & $M_{P_1^1P_2^2}$ & $M_{P_1^1P_2^3}$ & $M_{P_1^1P_2^4}$ & $M_{P_1^2P_2^1}$ & $M_{P_1^2P_2^2}$ & $M_{P_1^2P_2^3}$ & $M_{P_1^2P_2^4}$ \\
        & 1000 1000  & 1000 0100 & 1000 0010 & 1000 0001 & 0100 1000  & 0100 0100 & 0100 0010 & 0100 0001 \\
        \hline
        1110 1110 & 000 & 000 & 000 & --- & 000 & 000 & 000 & --- \\
        1110 1111 & 001 & 001 & 001 & 001 & 001 & 001 & 001 & 001 \\
        1110 0010 & --- & --- & 010 & --- & --- & --- & 010 & --- \\
        1110 1011 & 011 & --- & 011 & 011 & 011 & --- & 011 & 011 \\
        1111 1110 & 001 & 001 & 001 & --- & 001 & 001 & 001 & --- \\
        1111 1111 & 010 & 010 & 010 & 010  & 010 & 010 & 010 & 010 \\
        1111 0010 & --- & --- & 011 & --- & --- & --- & 011 & --- \\
        1111 1011 & 100 & --- & 100 & 100 & 100 & --- & 100 & 100 \\
        0010 1110 & --- & --- & --- & --- & --- & --- & --- & --- \\
        0010 1111 & --- & --- & --- & --- & --- & --- & --- & --- \\
        0010 0010 & --- & --- & --- & --- & --- & --- & --- & --- \\
        0010 1011 & --- & --- & --- & --- & --- & --- & --- & --- \\
        1011 1110 & 011 & 011 & 011 & --- & --- & --- & --- & --- \\
        1011 1111 & 100 & 100 & 100 & 100 & --- & --- & --- & --- \\
        1011 0010 & --- & --- & 101 & --- & --- & --- & --- & --- \\
        1011 1011 & 110 & --- & 110 & 110 & --- & --- & --- & --- \\
        \hline
        Result & 100 & 101 & 000 & 110 & 101 & 010 & 100 & 100 \\
        \hline
        \hline
        Term & $M_{P_1^3P_2^1}$ & $M_{P_1^3P_2^2}$ & $M_{P_1^3P_2^3}$ & $M_{P_1^3P_2^4}$ & $M_{P_1^4P_2^1}$ & $M_{P_1^4P_2^2}$ & $M_{P_1^4P_2^3}$ & $M_{P_1^4P_2^4}$ \\
        & 0010 1000  & 0010 0100 & 0010 0010 & 0010 0001 & 0001 1000  & 0001 0100 & 0001 0010 & 0001 0001 \\
        \hline
        1110 1110 & 000 & 000 & 000 & --- & --- & --- & --- & --- \\
        1110 1111 & 001 & 001 & 001 & 001 & --- & --- & --- & --- \\
        1110 0010 & --- & --- & 010 & --- & --- & --- & --- & --- \\
        1110 1011 & 011 & --- & 011 & 011 & --- & --- & --- & --- \\
        1111 1110 & 001 & 001 & 001 & --- & 001 & 001 & 001 & --- \\
        1111 1111 & 010 & 010 & 010 & 010 & 010 & 010 & 010 & 010 \\
        1111 0010 & --- & --- & 011 & --- & --- & --- & 011 & --- \\
        1111 1011 & 100 & --- & 100 & 100 & 100 & --- & 100 & 100 \\
        0010 1110 & 010 & 010 & 010 & --- & --- & --- & --- & --- \\
        0010 1111 & 011 & 011 & 011 & 011 & --- & --- & --- & --- \\
        0010 0010 & --- & --- & 100 & --- & --- & --- & --- & --- \\
        0010 1011 & 101 & --- & 101 & 101 & --- & --- & --- & --- \\
        1011 1110 & 011 & 011 & 011 & --- & 011 & 011 & 011 & --- \\
        1011 1111 & 100 & 100 & 100 & 100 & 100 & 100 & 100 & 100 \\
        1011 0010 & --- & --- & 101 & --- & --- & --- & 101 & --- \\
        1011 1011 & 110 & --- & 110 & 110 & 110 & --- & 110 & 110 \\
        \hline
        Result & 000 & 100 & 000 & 000 & 110 & 100 & 000 & 100 \\
    \end{tabular}}
    \caption{Spectrum of the adder from Table~\ref{tab:adder-codes}.}
    \label{tab:adder-spectrum}
\end{table}

The calculation of the MVI spectrum for Table~\ref{tab:adder-codes} is illustrated in Table~\ref{tab:adder-spectrum}. The table has the output functions of the terms as entries (in the order $f_cf_0f_1$) with the intersection of the term in each cell. 

For instance, the cell for the intersection of row 1110 1011 and column 0100 1000 is `011’. This is because the bit-by-bit intersection (AND operation) of the indices is 0100 1000, so both literals are not empty, thus the value of the cell is equal to the outputs of $f_c$, $f_0$, and $f_1$ for the normalized code 1110 1011, which is stated in the fourth row of  Table~\ref{tab:adder-codes} to be 011. 

The intersection of row 1110 0010 and column 0100 1000 is 0100 0000, so `---’ is placed in the corresponding cell because 0000 appears (the literal is empty). The final coefficients for the functions in Table~\ref{tab:adder-spectrum} are obtained by the bit-by-bit XOR operation on the entries of the columns.

From Table~\ref{tab:adder-spectrum}, it can be found that the functions $f_c$, $f_0$, and $f_1$ are
\small{
\begin{align*}
    f_c &= P_1^1P_2^1 \oplus P_1^1P_2^2 \oplus P_1^1P_2^4 \oplus P_1^2P_2^1 \oplus P_1^2P_2^3 \oplus P_1^2P_2^4 \oplus P_1^3P_2^2 \oplus P_1^4P_2^1 \oplus P_1^4P_2^2 \oplus P_1^4P_2^4 \\
    f_0 &= P_1^1P_2^4 \oplus P_1^2P_2^2 P_1^4P_2^1 \\
    f_1 &= P_1^1P_2^2 \oplus P_1^2P_2^1.
\end{align*}
}

\textbf{Step 3:} \textit{Replace the polarity literals of non-zero spectral coefficients by their original MVI literal. The output functions for the terms are given by the values of the spectral coefficients.}
Replacing the polarity literals with the original MVI literals yields the MVI-FPRM forms of the functions.
\small{
\begin{align*}
    f_c &= 1 \oplus X_2^{1,3} \oplus X_2^{0,1} \oplus X_1^{1,3} \oplus X_1^{1,3}X_2^{2} \oplus X_1^{1,3}X_2^{0,1} \oplus X_1^{2}X_2^{1,3} \oplus X_1^{0,1} \oplus X_1^{0,1}X_2^{1,3} \oplus X_1^{0,1}X_2^{0,1} \\
    f_0 &= X_2^{0,1} \oplus X_1^{1,3}X_2^{1,3} \oplus X_1^{0,1} \\
    f_1 &= X_2^{1,3} \oplus X_1^{1,3}.
\end{align*}
}

From the products-matching method, the MVI-FPRM forms (replacing $P_1^1$ and $P_2^1$ with 1) for $f_c$, $f_0$, and $f_1$ are 
\small{
\begin{align*}
    f_c &= 1 \oplus \textcolor{magenta}{P_2^2} \oplus \textcolor{blue}{P_2^4} \oplus \textcolor{purple}{P_1^2} \oplus P_1^2P_2^3 \oplus P_1^2P_2^4 \oplus P_1^3P_2^2 \oplus \textcolor{teal}{P_1^4} \oplus P_1^4P_2^2 \oplus P_1^4P_2^4 \\
    f_0 &= \textcolor{blue}{P_2^4} \oplus P_1^2P_2^2 \oplus \textcolor{teal}{P_1^4} \\
    f_1 &= \textcolor{magenta}{P_2^2} \oplus \textcolor{purple}{P_1^2}.
\end{align*}
}
The adder is realized as a quantum circuit in Fig.~\ref{fig:adder-circuit}. It is made up of 7 NOT gates, 6 CNOT gates, and 8 3-bit Toffoli gates. The circuit has a Maslov cost of 53 and a TQC of 523. Note that all $n$-bit Toffoli gates are at most $n=3$.

\begin{figure}[!htbp]
    \centering
    \includegraphics[width=0.6\linewidth]{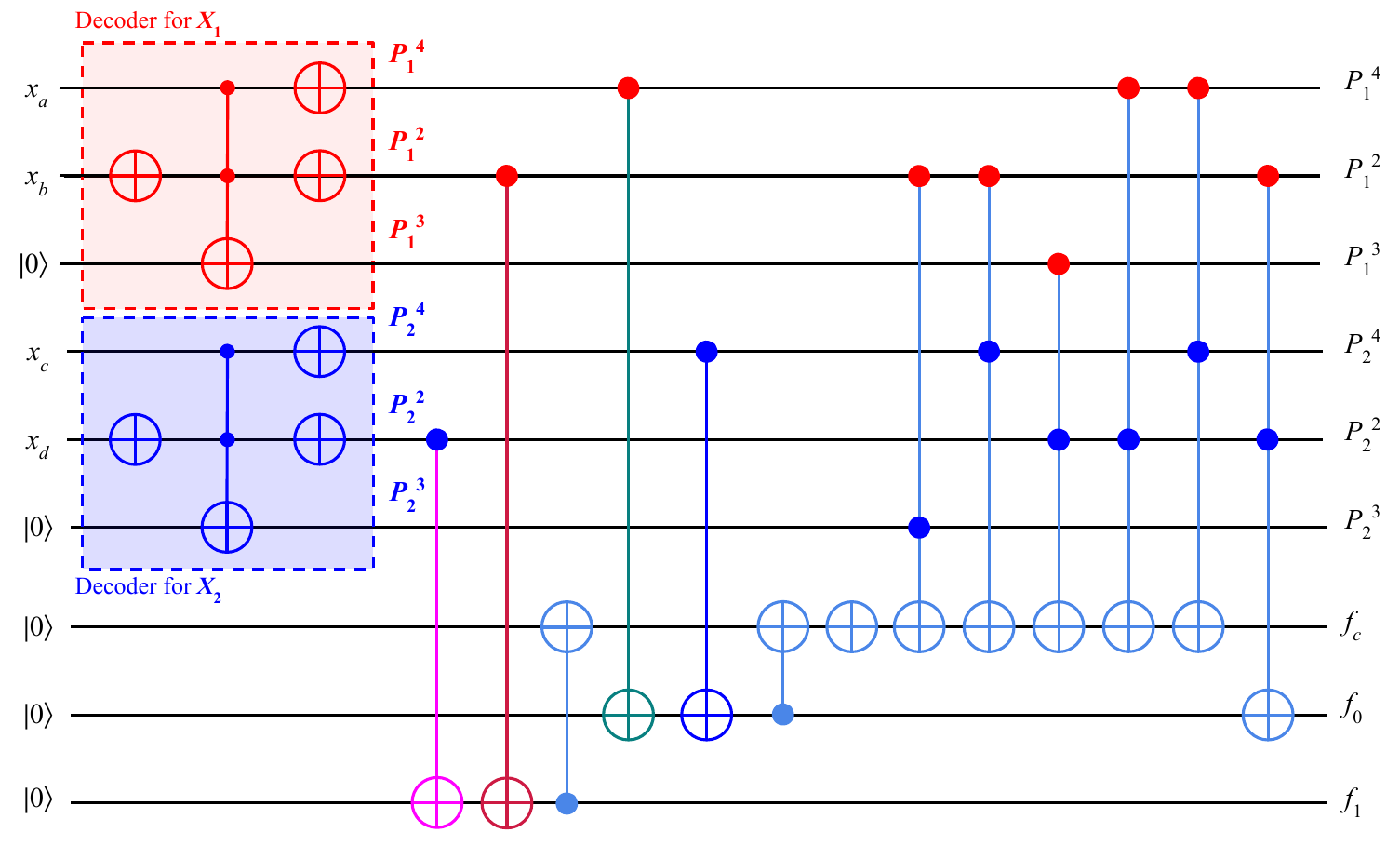}
    \caption{The 2-bit quantum adder based on the MVI-FPRM with the polarity $P_1$,$P_2$.}
    \label{fig:adder-circuit}
\end{figure}

\subsection{Second Possible 2-Bit Adder}
\label{adder-2}

For the second possible solution, let $X_1$ and $X_2$ instead have the polarities $Q_1$ and $Q_2$ respectively.
$$
Q_1 = 
\begin{bmatrix}
    1 & 1 & 1 & 1 \\
    0 & 1 & 1 & 0 \\
    0 & 0 & 1 & 0 \\
    1 & 1 & 0 & 0
\end{bmatrix}
, \
Q_2 = 
\begin{bmatrix}
    1 & 1 & 1 & 1 \\
    0 & 1 & 1 & 0 \\
    0 & 0 & 1 & 0 \\
    1 & 1 & 0 & 0
\end{bmatrix}
$$

\begin{table}[!htbp]
    \centering
    \renewcommand{\arraystretch}{1.1}
    \resizebox{\linewidth}{!}{%
    \begin{tabular}{cc|c|c|c||cc|c|c|c}
         $X_1$ & $X_2$ & Binary Code & Normalized Code & $f_cf_0f_1$ & $X_1$ & $X_2$ & Binary Code & Normalized Code & $f_cf_0f_1$ \\
         \hline
         0 & 0 & 1000 1000 & 0111 0111 & 000 &
         2 & 0 & 0010 1000 & 0010 0111 & 010 \\
         0 & 1 & 1000 0100 & 0111 0110 & 001 &
         2 & 1 & 0010 0100 & 0010 0110 & 011 \\
         0 & 2 & 1000 0010 & 0111 0010 & 010 &
         2 & 2 & 0010 0010 & 0010 0010 & 100 \\
         0 & 3 & 1000 0001 & 0111 1011 & 011 &
         2 & 3 & 0010 0001 & 0010 1011 & 101 \\
         1 & 0 & 0100 1000 & 0110 0111 & 001 &
         3 & 0 & 0001 1000 & 1011 0111 & 011 \\
         1 & 1 & 0100 0100 & 0110 0110 & 010 &
         3 & 1 & 0001 0100 & 1011 0110 & 100 \\
         1 & 2 & 0100 0010 & 0110 0010 & 011 &
         3 & 2 & 0001 0010 & 1011 0010 & 101 \\
         1 & 3 & 0100 0001 & 0110 1011 & 100 &
         3 & 3 & 0001 0001 & 1011 1011 & 110 \\
    \end{tabular}}
    \caption{Normalized codes for a 2-bit MVI adder with the polarities $Q_1$ and $Q_2$.}
    \label{tab:adder-codes2}
\end{table}

\textbf{Step 1:} \textit{Transform all the MVI literals of the function to their normalized codes according to the chosen polarities for the variables.}

The literals are $X_i^0 = Q_i^2 \oplus Q_i^3 \oplus Q_i^4$ (normalized code 0111), $X_i^1=Q_1^2 \oplus Q_i^3$ (normalized code 0110), $X_i^2=Q_i^3$ (normalized code 0010), and $X_i^3=Q_i^1 \oplus Q_i^3 \oplus Q_i^4$ (normalized code 1011). The codes for all the potential inputs are listed in Table~\ref{tab:adder-codes2}.

\textbf{Step 2:} \textit{Calculate the MVI-FPRM spectrum for the normalized codes.}

\begin{table}[!htb]
    \centering
    \renewcommand{\arraystretch}{1.1}
    \resizebox{\linewidth}{!}{%
    \begin{tabular}{l||c|c|c|c|c|c|c|c}
        Term & $M_{Q_1^1Q_2^1}$ & $M_{Q_1^1Q_2^2}$ & $M_{Q_1^1Q_2^3}$ & $M_{Q_1^1Q_2^4}$ & $M_{Q_1^2Q_2^1}$ & $M_{Q_1^2Q_2^2}$ & $M_{Q_1^2Q_2^3}$ & $M_{Q_1^2Q_2^4}$ \\
        & 1000 1000  & 1000 0100 & 1000 0010 & 1000 0001 & 0100 1000  & 0100 0100 & 0100 0010 & 0100 0001 \\
        \hline
        0111 0111 & --- & --- & --- & --- & --- & 000 & 000 & 000 \\
        0111 0110 & --- & --- & --- & --- & --- & 001 & 001 & --- \\
        0111 0010 & --- & --- & --- & --- & --- & --- & 010 & --- \\
        0111 1011 & --- & --- & --- & --- & 011 & --- & 011 & 011 \\
        
        0110 0111 & --- & --- & --- & --- & --- & 001 & 001 & 001 \\
        0110 0110 & --- & --- & --- & --- & --- & 010 & 010 & --- \\
        0110 0010 & --- & --- & --- & --- & --- & --- & 011 & --- \\
        0110 1011 & --- & --- & --- & --- & 100 & --- & 100 & 100 \\
        
        0010 0111 & --- & --- & --- & --- & --- & --- & --- & --- \\
        0010 0110 & --- & --- & --- & --- & --- & --- & --- & --- \\
        0010 0010 & --- & --- & --- & --- & --- & --- & --- & --- \\
        0010 1011 & --- & --- & --- & --- & --- & --- & --- & --- \\
        
        1011 0111 & --- & 011 & 011 & 011 & --- & --- & --- & --- \\
        1011 0110 & --- & 100 & 100 & --- & --- & --- & --- & --- \\
        1011 0010 & --- & --- & 101 & --- & --- & --- & --- & --- \\
        1011 1011 & 110 & --- & 110 & 110 & --- & --- & --- & --- \\
        \hline
        Result & 110 & 111 & 100 & 101 & 111 & 010 & 100 & 110 \\
        \hline
        \hline
        Term & $M_{Q_1^3Q_2^1}$ & $M_{Q_1^3Q_2^2}$ & $M_{Q_1^3Q_2^3}$ & $M_{Q_1^3Q_2^4}$ & $M_{Q_1^4Q_2^1}$ & $M_{Q_1^4Q_2^2}$ & $M_{Q_1^4Q_2^3}$ & $M_{Q_1^4Q_2^4}$ \\
        & 0010 1000  & 0010 0100 & 0010 0010 & 0010 0001 & 0001 1000  & 0001 0100 & 0001 0010 & 0001 0001 \\
        \hline
        0111 0111 & --- & 000 & 000 & 000 & --- & 000 & 000 &  \\
        0111 0110 & --- & 001 & 001 & --- & --- & 001 & 001 & --- \\
        0111 0010 & --- & --- & 010 & --- & --- & --- & 010 & --- \\
        0111 1011 & 011 & --- & 011 & 011 & 011 & --- & 011 & 011 \\
        
        0110 0111 & --- & 001 & 001 & 001 & --- & --- & --- & --- \\
        0110 0110 & --- & 010 & 010 & --- & --- & --- & --- & --- \\
        0110 0010 & --- & --- & 011 & --- & --- & --- & --- & --- \\
        0110 1011 & 100 & --- & 100 & 100 & --- & --- & --- & --- \\
        
        0010 0111 & --- & 010 & 010 & 010 & --- & --- & --- & --- \\
        0010 0110 & --- & 011 & 011 & --- & --- & --- & --- & --- \\
        0010 0010 & --- & --- & 100 & --- & --- & --- & --- & --- \\
        0010 1011 & 101 & --- & 101 & 101 & --- & --- & --- & --- \\
        
        1011 0111 & --- & 011 & 011 & 011 & --- & 011 & 011 & 011 \\
        1011 0110 & --- & 100 & 100 & --- & --- & 100 & 100 & --- \\
        1011 0010 & --- & --- & 101 & --- & --- & --- & 101 & --- \\
        1011 1011 & 110 & --- & 110 & 110 & 110 & --- & 110 & 110 \\
        \hline
        Result & 100 & 100 & 100 & 100 & 101 & 110 & 100 & 110 \\
    \end{tabular}}
    \caption{Spectrum of the adder from Table~\ref{tab:adder-codes2}.}
    \label{tab:adder-spectrum2}
\end{table}

The calculation of the MVI spectrum for Table~\ref{tab:adder-codes2} is illustrated in Table~\ref{tab:adder-spectrum2}. From Table~\ref{tab:adder-spectrum2}, it can be found that the functions $f_c$, $f_0$, and $f_1$ are
\small{
\begin{align*}
    f_c &= Q_1^1Q_2^1 \oplus Q_1^1Q_2^2 \oplus Q_1^1Q_2^3 \oplus Q_1^1Q_2^4 \oplus Q_1^2Q_2^1 \oplus Q_1^2Q_2^3 \oplus Q_1^2Q_2^4 \oplus Q_1^3Q_2^1 \oplus Q_1^3Q_2^2 \oplus Q_1^3Q_2^4 \oplus Q_1^4Q_2^1 \oplus Q_1^4Q_2^2 \\&\oplus Q_1^4Q_2^3 \oplus Q_1^4Q_2^4\\
    f_0 &= Q_1^1Q_2^1 \oplus Q_1^1Q_2^2 \oplus Q_1^2Q_2^1 \oplus Q_1^2Q_2^2 \oplus Q_1^2Q_2^4 \oplus Q_1^4Q_2^2 \oplus Q_1^4Q_2^4 \\
    f_1 &= Q_1^1Q_2^2 \oplus Q_1^1Q_2^4 \oplus Q_1^2Q_2^1 \oplus Q_1^4Q_2^1.
\end{align*}
}

\textbf{Step 3:} \textit{Replace the polarity literals of non-zero spectral coefficients by their original MVI literal. The output functions for the terms are given by the values of the spectral coefficients.}
Replacing the polarity literals with the original MVI literals yields the MVI-FPRM forms of the functions.
\small{
\begin{align*}
    f_c &= 1 \oplus X_2^{1,2} \oplus X_2^{2} \oplus X_2^{0,1} \oplus X_1^{1,2} \oplus X_1^{1,2}X_2^{2} \oplus X_1^{1,2}X_2^{0,1} \oplus X_1^{2} \oplus X_1^{2}X_2^{1,2} \oplus X_1^{2}X_2^{0,1} \oplus X_1^{0,1} \oplus X_1^{0,1}X_2^{1,2} \\&\oplus X_1^{0,1}X_2^{2} \oplus X_1^{0,1}X_2^{0,1}\\
    f_0 &= 1 \oplus X_2^{1,2} \oplus X_1^{1,2} \oplus X_1^{1,2}X_2^{1,2} \oplus X_1^{1,2}X_2^{0,1} \oplus X_1^{0,1}X_2^{1,2} \oplus X_1^{0,1}X_2^{0,1} \\
    f_1 &= X_2^{1,2} \oplus X_2^{0,1} \oplus X_1^{1,2} \oplus X_1^{0,1}.
\end{align*}
}

From the products-matching method, the MVI-FPRM forms (replacing $P_1^1$ and $P_2^1$ with 1) for $f_c$, $f_0$, and $f_1$ are 
\small{
\begin{align*}
    f_c &= \textcolor{purple}{1} \oplus \textcolor{blue}{Q_2^2} \oplus Q_2^3 \oplus \textcolor{teal}{Q_2^4} \oplus \textcolor{blue}{Q_1^2} \oplus Q_1^2Q_2^3 \oplus \textcolor{purple}{Q_1^2Q_2^4} \oplus Q_1^3 \oplus Q_1^3Q_2^2 \oplus Q_1^3Q_2^4 \oplus \textcolor{teal}{Q_1^4} \oplus \textcolor{purple}{Q_1^4Q_2^2} \oplus Q_1^4Q_2^3 \oplus \textcolor{purple}{Q_1^4Q_2^4} \\
    f_0 &= \textcolor{purple}{1} \oplus \textcolor{blue}{Q_2^2} \oplus \textcolor{blue}{Q_1^2} \oplus Q_1^2Q_2^2 \oplus \textcolor{purple}{Q_1^2Q_2^4} \oplus \textcolor{purple}{Q_1^4Q_2^2} \oplus \textcolor{purple}{Q_1^4Q_2^4} \\
    f_1 &= \textcolor{blue}{Q_2^2} \oplus \textcolor{teal}{Q_2^4} \oplus \textcolor{blue}{Q_1^2} \oplus \textcolor{teal}{Q_1^4}.
\end{align*}
}
The adder is realized as a quantum circuit in Fig.~\ref{fig:adder-circuit2}, which consists of 5 NOT gates, 12 CNOT gates, and 10 3-bit Toffoli gates. The circuit has a Maslov cost of 67 and a TQC of 713.

\begin{figure}[!htbp]
    \centering
    \includegraphics[width=0.6\linewidth]{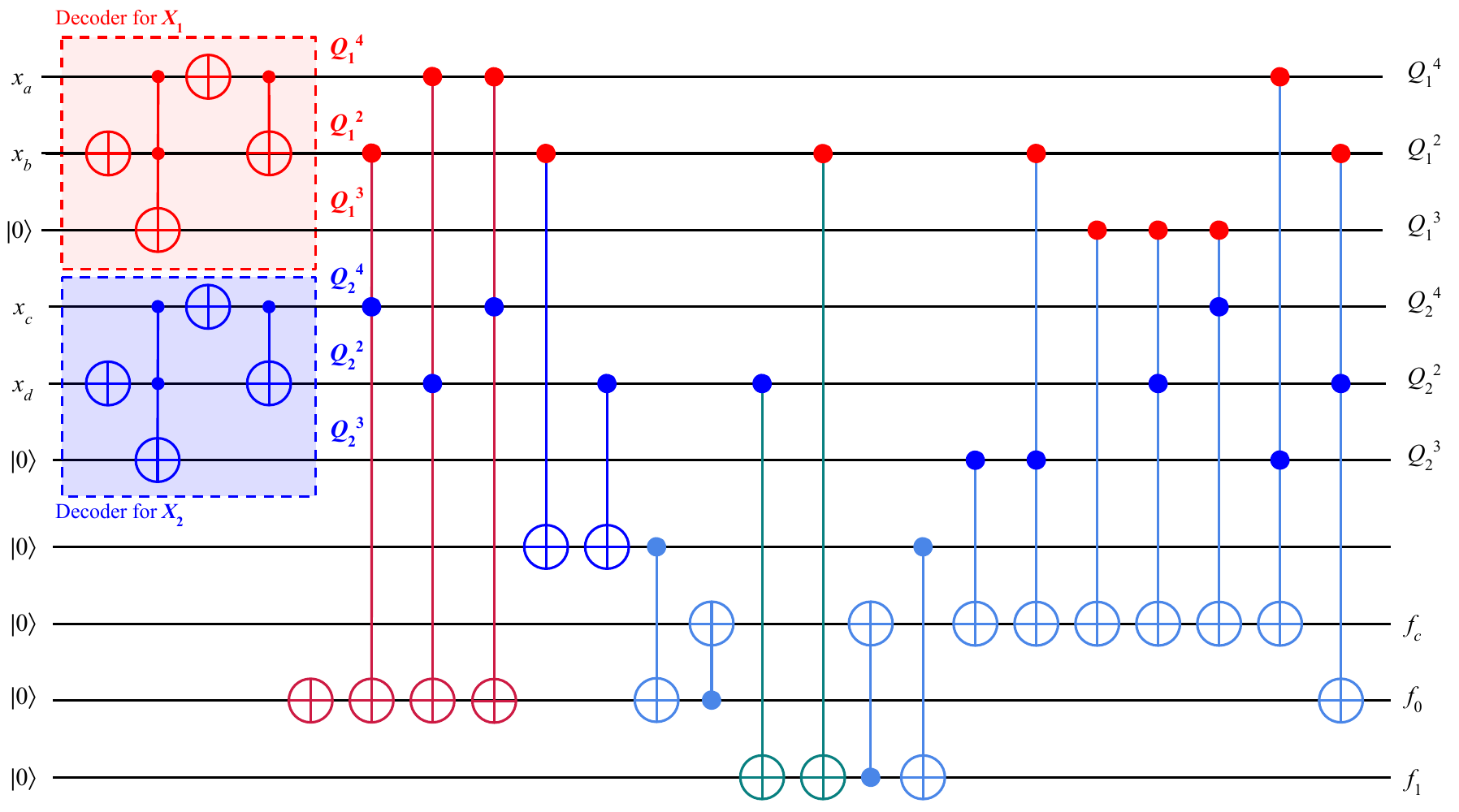}
    \caption{The 2-bit quantum adder based on the MVI-FPRM with the polarities $Q_1$ and $Q_2$.}
    \label{fig:adder-circuit2}
\end{figure}

\section{Butterfly Diagram Method}
\label{fprm-butterfly}

This paper also introduces a design method based on butterfly diagrams \cite{weinstein69, shanks69, jin20, lee16} for MVI-FPRM forms that transform \textit{minterms}, terms that contain single-valued literals of every variable, to a polarity. This is a method analogous to the binary FPRM butterfly in Section~\ref{butterfly}.

\begin{figure}[!htbp]
    \centering
    \begin{subfigure}{0.27\linewidth}
        \includegraphics[width=\linewidth]{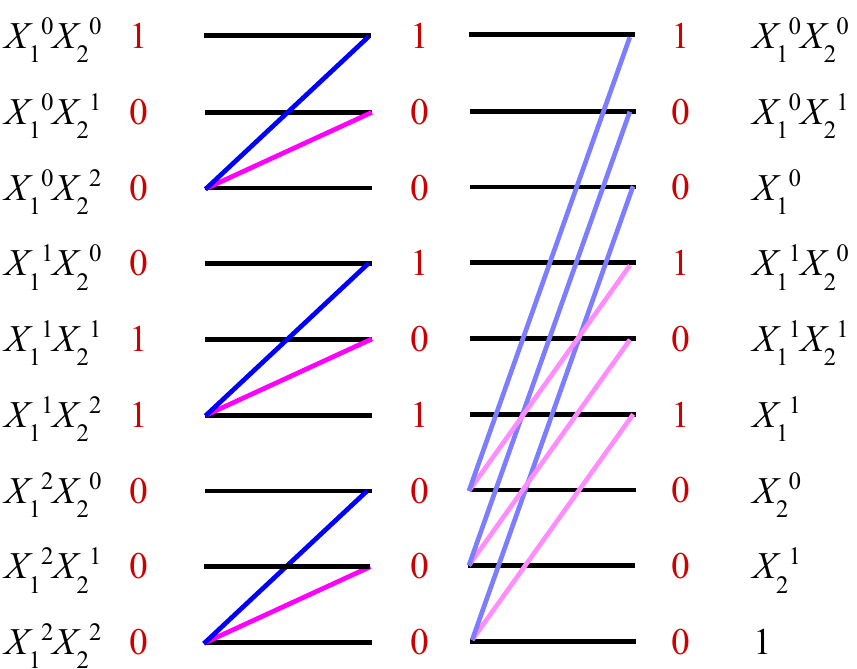}
        \caption{Our polarity-based ternary butterfly diagram.}
        \label{fig:butterfly}
    \end{subfigure}
    \hspace{.05\linewidth}
    \begin{subfigure}{0.2\linewidth}
        \includegraphics[width=\linewidth]{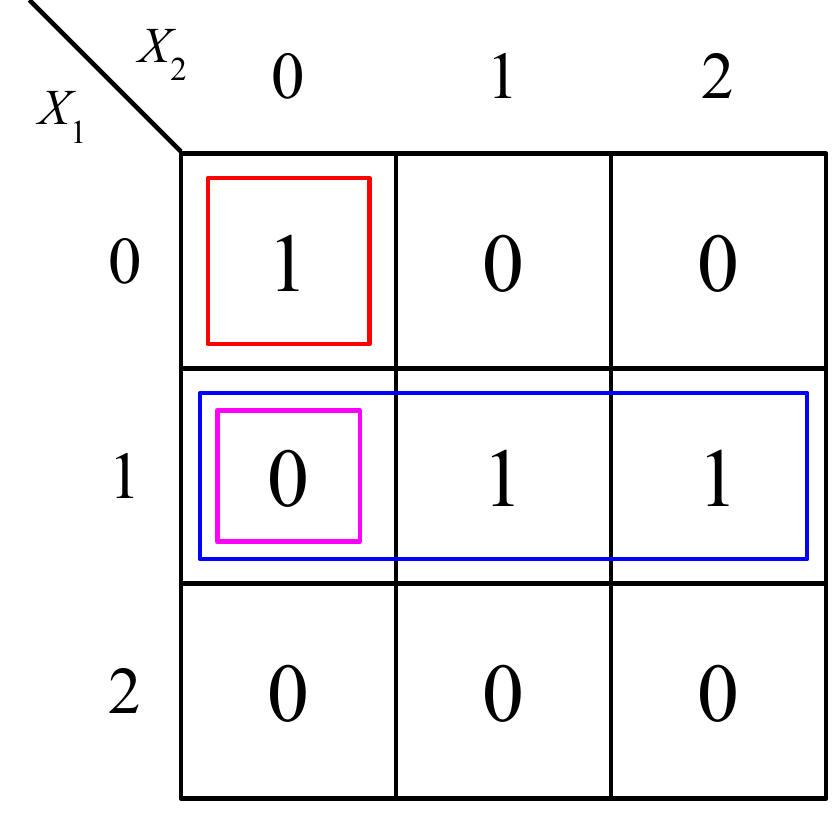}
        \caption{Marquand chart of MVI-FPRM.} 
        \label{fig:butterfly-marquand-chart}
    \end{subfigure}
    \caption{Example of butterfly diagram for $X_1^0X_2^0 \oplus X_1^1X_2^1 \oplus X_1^1X_2^2=X_1^0X_2^0 \oplus X_1^1X_2^0 \oplus X_1^1$ with the corresponding Marquand chart.}
\end{figure}

One example of a butterfly diagram is shown in Fig.~\ref{fig:butterfly}, which transforms the minterms of two ternary variables $X_1$ and $X_2$ to the polarity $P_1$,$P_2$ with literals $X_i^0$, $X_i^1$, $X_i^{0,1,2}=1$. The values shown in the butterfly diagram transform $X_1^0X_2^0 \oplus X_1^1X_2^1 \oplus X_1^1X_2^2$ to $X_1^0X_2^0 \oplus X_1^1X_2^0 \oplus X_1^1$. The values of this function are illustrated in the Marquand chart in Fig.~\ref{fig:butterfly-marquand-chart}. In this case, both variables are of the same polarity, but diagrams can be created for variables with various polarities. For a function of $n$ ternary variables, there $28^n$ possible butterfly diagrams, out of which the best solutions are selected. Note that another solution for this function is $X_1^0 \oplus X_1^{0,1}X_2^{1,2}$, which is an MVI-GRM.

\begin{example}
    \label{ex:butterfly-part}
    Let the variable $X$ be ternary. The single-valued literals for $X$ are $X^0$, $X^1$, and $X^2$. The polarity matrix for this example is $P$. The goal is to create a butterfly diagram that transforms an XOR of the minterms (for one variable, the minterms are essentially just the single-valued literals) to the MVI-FPRM form for a function of a single ternary variable. This diagram can then be generalized to MVI functions with multiple variables.
    $$
    P=
    \begin{bmatrix}
        1 & 0 & 0 \\
        1 & 1 & 0 \\
        1 & 0 & 1
    \end{bmatrix}
    $$

    Any MVI function of one variable can be expressed in the form
    $$f=a_0X^0 \oplus a_1X^1 \oplus a_2X^2.$$
    Note that this form is equivalent to the ternary input Shannon expansion in Eq.~\eqref{eq:ternary-input-shannon} from Section~\ref{ternary-expansions}, where the cofactors $f_{X=k}=a_k$. The coefficients $a_0$, $a_1$, and $a_2$ are the inputs to the butterfly diagram.

    The single-valued literals are expressed in terms of the polarity literals as follows: 
    \begin{align*}
        X^0 &= P^1 \\
        X^1 &= P^1 \oplus P^2 \\
        X^2 &= P^1 \oplus P^3.
    \end{align*}
    The single-valued literals can be substituted as polarity literals to yield a new expression for $f$.
    \begin{align*}
        f 
        &= a_0X^0 \oplus a_1X^1 \oplus a_2X^2 \\
        &= a_0(P^1) \oplus a_1(P^1 \oplus P^2) \oplus a_2(P^1 \oplus P^3) \\
        &= a_0P^1 \oplus a_1P^1 \oplus a_1P^2 \oplus a_2P^1 \oplus a_2P^3 \\
        &= a_0P^1 \oplus a_1P^1 \oplus a_2P^1 \oplus a_1P^2 \oplus a_2P^3 \\
        &= (a_0 \oplus a_1 \oplus a_2)P^1 \oplus a_1P^2 \oplus a_2P^3 \\
        &= M_{P^1}P^1 \oplus M_{P^2}P^2 \oplus M_{P^1}P^3
    \end{align*}
    This expansion, which is equivalent to the ternary input Davio-like expansion for the polarity $P$, which was found in Eq.~\eqref{eq:ternary-davio-example2}, finds that spectral coefficients, which are the outputs of the butterfly diagram, are
    \begin{align*}
        M_{P^1} &= a_0 \oplus a_1 \oplus a_2 \\
        M_{P^2} &= a_1\\
        M_{P^3} &= a_2.
    \end{align*}

    The butterfly diagram for $P$ is shown in Fig.~\ref{fig:ex-butterfly-part}.

    \begin{figure}[!htb]
        \centering
        \includegraphics[width=0.4\linewidth]{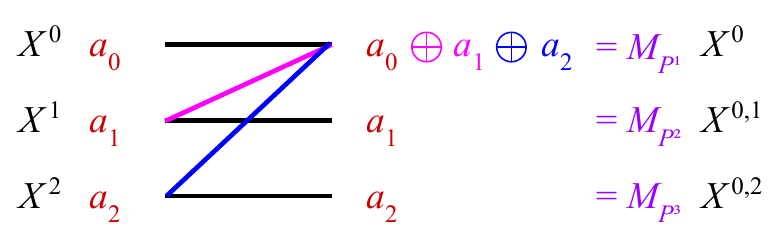}
        \caption{Butterfly diagram for $P$ with a single ternary variable, $X$, for Example~\ref{ex:butterfly-part}.}
        \label{fig:ex-butterfly-part}
    \end{figure}
\end{example}

The same process from Example~\ref{ex:butterfly-part} can be applied for all polarities and radices. The butterfly diagram that transforms an XOR of minterms to the MVI-FPRM form for a single ternary variable for each polarity is listed in Table~\ref{tab:ternary-input-butterflies}.

\begin{table}[htb]
    \centering
    \resizebox{\linewidth}{!}{
    \begin{tabular}{c|c||c|c||c|c||c|c}
         Polarity & Butterfly Diagram & Polarity & Butterfly Diagram & Polarity & Butterfly Diagram & Polarity & Butterfly Diagram \\ 
         \hline
         &&&&&&&\\[-8px]
         
         $\begin{bmatrix}
             1 & 0 & 0 \\
             0 & 1 & 0 \\
             0 & 0 & 1
         \end{bmatrix}$ 
         & 
         \begin{minipage}{0.2\linewidth}
             \centering \includegraphics[width=0.4\linewidth]{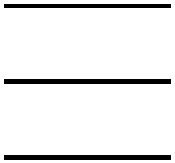}
        \end{minipage} 
        &
        $\begin{bmatrix}
             1 & 0 & 0 \\
             1 & 1 & 0 \\
             0 & 0 & 1
         \end{bmatrix}$ 
         & 
         \begin{minipage}{0.2\linewidth}
             \centering \includegraphics[width=0.4\linewidth]{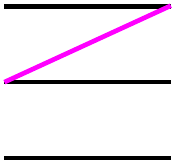}
        \end{minipage}
        &
        $\begin{bmatrix}
             1 & 0 & 0 \\
             0 & 1 & 0 \\
             0 & 1 & 1
         \end{bmatrix}$ 
         & 
         \begin{minipage}{0.2\linewidth}
             \centering \includegraphics[width=0.4\linewidth]{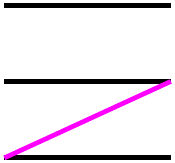}
        \end{minipage}
        &
        $\begin{bmatrix}
             1 & 0 & 0 \\
             1 & 1 & 0 \\
             1 & 1 & 1
         \end{bmatrix}$ 
         & 
         \begin{minipage}{0.2\linewidth}
             \centering \includegraphics[width=0.4\linewidth]{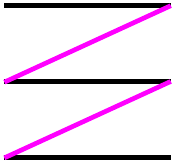}
        \end{minipage}
        \\[20px]

        $\begin{bmatrix}
             1 & 0 & 0 \\
             0 & 1 & 0 \\
             1 & 0 & 1
         \end{bmatrix}$ 
         & 
         \begin{minipage}{0.2\linewidth}
             \centering \includegraphics[width=0.4\linewidth]{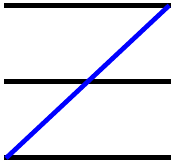}
        \end{minipage}
        &
        $\begin{bmatrix}
             1 & 0 & 0 \\
             1 & 1 & 0 \\
             1 & 0 & 1
         \end{bmatrix}$ 
         & 
         \begin{minipage}{0.2\linewidth}
             \centering \includegraphics[width=0.4\linewidth]{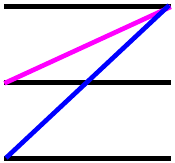}
        \end{minipage}
        &
        $\begin{bmatrix}
             1 & 0 & 0 \\
             0 & 1 & 0 \\
             1 & 1 & 1
         \end{bmatrix}$ 
         & 
         \begin{minipage}{0.2\linewidth}
             \centering \includegraphics[width=0.4\linewidth]{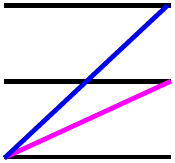}
        \end{minipage}
        &
        $\begin{bmatrix}
             1 & 0 & 0 \\
             1 & 1 & 0 \\
             0 & 1 & 1
         \end{bmatrix}$ 
         & 
         \begin{minipage}{0.2\linewidth}
             \centering \includegraphics[width=0.4\linewidth]{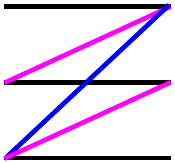}
        \end{minipage}
        \\[20px]

        $\begin{bmatrix}
             1 & 1 & 0 \\
             0 & 1 & 0 \\
             0 & 0 & 1
         \end{bmatrix}$ 
         & 
         \begin{minipage}{0.2\linewidth}
             \centering \includegraphics[width=0.4\linewidth]{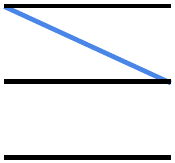}
        \end{minipage}
        &
        $\begin{bmatrix}
             1 & 0 & 0 \\
             0 & 1 & 1 \\
             0 & 0 & 1
         \end{bmatrix}$ 
         & 
         \begin{minipage}{0.2\linewidth}
             \centering \includegraphics[width=0.4\linewidth]{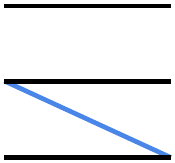}
        \end{minipage}
        &
        $\begin{bmatrix}
             1 & 1 & 1 \\
             0 & 1 & 1 \\
             0 & 0 & 1
         \end{bmatrix}$ 
         & 
         \begin{minipage}{0.2\linewidth}
             \centering \includegraphics[width=0.4\linewidth]{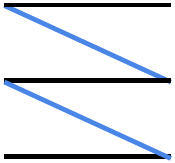}
        \end{minipage}
        &
        $\begin{bmatrix}
             1 & 0 & 0 \\
             1 & 1 & 1 \\
             0 & 0 & 1
         \end{bmatrix}$ 
         & 
         \begin{minipage}{0.2\linewidth}
             \centering \includegraphics[width=0.4\linewidth]{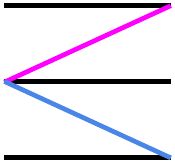}
        \end{minipage}
        \\[20px]

        $\begin{bmatrix}
             1 & 1 & 0 \\
             0 & 1 & 0 \\
             0 & 1 & 1
         \end{bmatrix}$ 
         & 
         \begin{minipage}{0.2\linewidth}
             \centering \includegraphics[width=0.4\linewidth]{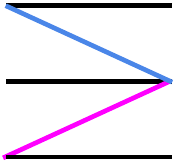}
        \end{minipage}
        &
        $\begin{bmatrix}
             0 & 1 & 1 \\
             1 & 1 & 1 \\
             1 & 1 & 0
         \end{bmatrix}$ 
         & 
         \begin{minipage}{0.2\linewidth}
             \centering \includegraphics[width=0.4\linewidth]{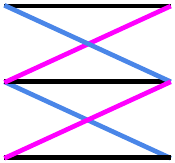}
        \end{minipage}
        &
        $\begin{bmatrix}
             1 & 1 & 0 \\
             0 & 1 & 0 \\
             1 & 1 & 1
         \end{bmatrix}$ 
         & 
         \begin{minipage}{0.2\linewidth}
             \centering \includegraphics[width=0.4\linewidth]{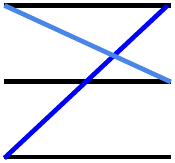}
        \end{minipage}
        &
        $\begin{bmatrix}
             1 & 0 & 0 \\
             1 & 1 & 1 \\
             1 & 0 & 1
         \end{bmatrix}$ 
         & 
         \begin{minipage}{0.2\linewidth}
             \centering \includegraphics[width=0.4\linewidth]{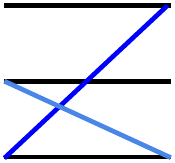}
        \end{minipage}
        \\[20px]

        $\begin{bmatrix}
             1 & 1 & 0 \\
             0 & 1 & 0 \\
             1 & 0 & 1
         \end{bmatrix}$ 
         & 
         \begin{minipage}{0.2\linewidth}
             \centering \includegraphics[width=0.4\linewidth]{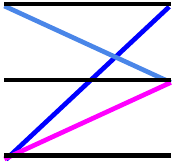}
        \end{minipage}
        &
        $\begin{bmatrix}
             1 & 0 & 0 \\
             0 & 1 & 1 \\
             1 & 0 & 1
         \end{bmatrix}$ 
         & 
         \begin{minipage}{0.2\linewidth}
             \centering \includegraphics[width=0.4\linewidth]{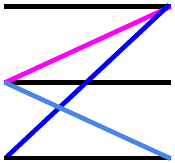}
        \end{minipage}
        &
        $\begin{bmatrix}
             1 & 0 & 1 \\
             0 & 1 & 0 \\
             0 & 0 & 1
         \end{bmatrix}$ 
         & 
         \begin{minipage}{0.2\linewidth}
             \centering \includegraphics[width=0.4\linewidth]{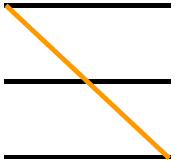}
        \end{minipage}
        &
        $\begin{bmatrix}
             1 & 0 & 1 \\
             1 & 1 & 1 \\
             0 & 0 & 1
         \end{bmatrix}$ 
         & 
         \begin{minipage}{0.2\linewidth}
             \centering \includegraphics[width=0.4\linewidth]{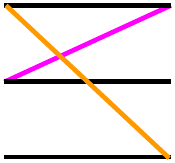}
        \end{minipage}
        \\[20px]

        $\begin{bmatrix}
             1 & 1 & 1 \\
             0 & 1 & 0 \\
             0 & 1 & 1
         \end{bmatrix}$ 
         & 
         \begin{minipage}{0.2\linewidth}
             \centering \includegraphics[width=0.4\linewidth]{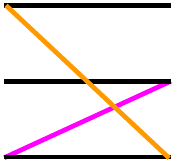}
        \end{minipage}
        &
        $\begin{bmatrix}
             1 & 1 & 1 \\
             0 & 1 & 0 \\
             0 & 0 & 1
         \end{bmatrix}$ 
         & 
         \begin{minipage}{0.2\linewidth}
             \centering \includegraphics[width=0.4\linewidth]{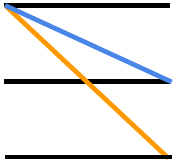}
        \end{minipage}
        &
        $\begin{bmatrix}
             1 & 0 & 1 \\
             0 & 1 & 1 \\
             0 & 0 & 1
         \end{bmatrix}$ 
         & 
         \begin{minipage}{0.2\linewidth}
             \centering \includegraphics[width=0.4\linewidth]{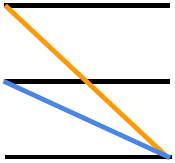}
        \end{minipage}
        &
        $\begin{bmatrix}
             1 & 1 & 0 \\
             0 & 1 & 1 \\
             0 & 0 & 1
         \end{bmatrix}$ 
         & 
         \begin{minipage}{0.2\linewidth}
             \centering \includegraphics[width=0.4\linewidth]{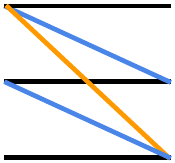}
        \end{minipage}
        \\[20px]

        $\begin{bmatrix}
             1 & 0 & 1 \\
             1 & 1 & 0 \\
             0 & 0 & 1
         \end{bmatrix}$ 
         & 
         \begin{minipage}{0.2\linewidth}
             \centering \includegraphics[width=0.4\linewidth]{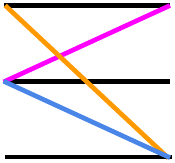}
        \end{minipage}
        &
        $\begin{bmatrix}
             1 & 0 & 1 \\
             0 & 1 & 0 \\
             0 & 1 & 1
         \end{bmatrix}$ 
         & 
         \begin{minipage}{0.2\linewidth}
             \centering \includegraphics[width=0.4\linewidth]{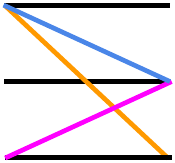}
        \end{minipage}
        &
        $\begin{bmatrix}
             0 & 1 & 1 \\
             1 & 0 & 1 \\
             1 & 1 & 1
         \end{bmatrix}$ 
         & 
         \begin{minipage}{0.2\linewidth}
             \centering \includegraphics[width=0.4\linewidth]{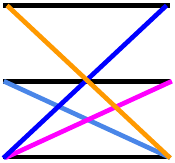}
        \end{minipage}
        &
        $\begin{bmatrix}
             1 & 1 & 0 \\
             1 & 0 & 1 \\
             1 & 1 & 1
         \end{bmatrix}$ 
         & 
         \begin{minipage}{0.2\linewidth}
             \centering \includegraphics[width=0.4\linewidth]{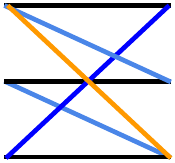}
        \end{minipage}
        \\[16px]
    \end{tabular}}
    \caption{All single-variable butterfly diagrams for the transformations to ternary input polarities.}
    \label{tab:ternary-input-butterflies}
\end{table}

To create the MVI-FPRM butterfly diagram for the ternary input function $F(X_1, X_2,\dots, X_n)$ with polarity $P_1, P_2,\dots, P_n$:
\begin{itemize}
    \item The inputs represent the coefficients of the minterms, listed in natural order.
    \item The butterfly diagram is formed from $n$ columns for each variable in descending order ($X_n,X_{n-1},\dots,X_1$). 
    \begin{itemize}
        \item The $k$th column transforms the single-valued literals of the variable $X_{n-k+1}$ to the polarity $P_{n-k+1}$. 
        \item The $k$th column will be constructed with $3^k$ kernels, based on the polarity.
    \end{itemize}
    \item The outputs represent the spectral coefficients of the polarity terms, listed in natural order.
\end{itemize}
This method can be extended to any radix. This is demonstrated in Example~\ref{ex:butterfly}.

\begin{example}
\label{ex:butterfly}
    Take the function  $F_3=X_4^1X_5^{0,2} \oplus X_3^{1,2}X_4^{0,1}X_5^0 \oplus X_4^{0,2}X_5^{1,2} \oplus X_3^2X_4^1X_5^1$ from Example~\ref{ex:fprm three variables}. Recall that the polarities used in Example~\ref{ex:fprm three variables} were $P_3$, $P_4$, and $P_5$.
    $$
    P_3=
    \begin{bmatrix}
        1 & 1 & 1 \\
        1 & 0 & 1 \\
        0 & 1 & 1 \\
    \end{bmatrix}
    , \
    P_4=
    \begin{bmatrix}
        1 & 1 & 1 \\
        1 & 1 & 0 \\
        0 & 1 & 0
    \end{bmatrix}
    , \
    P_5=
    \begin{bmatrix}
        1 & 1 & 1 \\
        1 & 1 & 0 \\
        0 & 1 & 1
    \end{bmatrix}
    $$

    To fit the matrices to those of Table~\ref{tab:ternary-input-butterflies}, the polarities will be redefined as $Q_3$, $Q_4$, $Q_5$, where the rows are reordered.
    $$
    Q_3=
    \begin{bmatrix}
        0 & 1 & 1 \\
        1 & 0 & 1 \\
        1 & 1 & 1 \\
    \end{bmatrix}
    , \
    Q_4=
    \begin{bmatrix}
        1 & 1 & 0 \\
        0 & 1 & 0 \\
        1 & 1 & 1
    \end{bmatrix}
    , \
    Q_5=
    \begin{bmatrix}
        0 & 1 & 1 \\
        1 & 1 & 1 \\
        1 & 1 & 0
    \end{bmatrix}
    $$

    The function $F_3$ can be expanded to yield an equivalent expression with only minterms. The same expression can also be found through the Marquand Chart (Fig.~\ref{fig:ex10-f-chart} from Example~\ref{ex:fprm three variables}).
    \small{
    \begin{align*}
        F_3 
        &= X_4^1X_5^{0,2} \oplus
        X_3^{1,2}X_4^{0,1}X_5^0 \oplus 
        X_4^{0,2}X_5^{1,2} \oplus 
        X_3^2X_4^1X_5^1 \\
        &= (X_3^0 \oplus X_3^1 \oplus X_3^2)(X_4^1)(X_5^0 \oplus X_5^2) \oplus 
        (X_3^1 \oplus X_3^2)(X_4^0 \oplus X_4^1)(X_5^0) \\ &\oplus
        (X_3^0 \oplus X_3^1 \oplus X_3^2)(X_4^0 \oplus X_4^2)(X_5^1 \oplus X_5^2) \oplus 
        (X_3^2)(X_4^1)(X_5^1) \\
        &= X_3^0X_4^0X_5^1 \oplus
        X_3^0X_4^0X_5^2 \oplus
        X_3^0X_4^1X_5^0 \oplus
        X_3^0X_4^1X_5^2 \oplus
        X_3^0X_4^2X_5^1 \oplus
        X_3^0X_4^2X_5^2 \oplus
        X_3^1X_4^0X_5^0 \oplus
        X_3^1X_4^0X_5^1 \oplus
        X_3^1X_4^0X_5^2 \\ &\oplus
        X_3^1X_4^1X_5^2 \oplus
        X_3^2X_4^0X_5^0 \oplus
        X_3^2X_4^1X_5^1 \oplus
        X_3^2X_4^2X_5^1 \oplus
        X_3^2X_4^2X_5^2
    \end{align*}}

    \begin{table}[htb]
        \centering
        \renewcommand{\arraystretch}{1.2}
        \resizebox{\linewidth}{!}{
        \begin{tabular}{l|ccc ccc ccc}
             Minterm & $X_3^0X_4^0X_5^0$ & $X_3^0X_4^0X_5^1$ & $X_3^0X_4^0X_5^2$ & $X_3^0X_4^1X_5^0$ & $X_3^0X_4^1X_5^1$ & $X_3^0X_4^1X_5^2$ & $X_3^0X_4^2X_5^0$ & $X_3^0X_4^2X_5^1$ & $X_3^0X_4^2X_5^2$ \\
             \hline
             Input & 0 & 1 & 1 & 1 & 0 & 1 & 0 & 1 & 1  \\
             \hline
             \hline
             Minterm & $X_3^1X_4^0X_5^0$ & $X_3^1X_4^0X_5^1$ & $X_3^1X_4^0X_5^2$ & $X_3^1X_4^1X_5^0$ & $X_3^1X_4^1X_5^1$ & $X_3^1X_4^1X_5^2$ & $X_3^1X_4^2X_5^0$ & $X_3^1X_4^2X_5^1$ & $X_3^1X_4^2X_5^2$ \\
             \hline
             Input & 1 & 1 & 1 & 0 & 0 & 1 & 0 & 0 & 0  \\
             \hline
             \hline
             Minterm & $X_3^2X_4^0X_5^0$ & $X_3^2X_4^0X_5^1$ & $X_3^2X_4^0X_5^2$ & $X_3^2X_4^1X_5^0$ & $X_3^2X_4^1X_5^1$ & $X_3^2X_4^1X_5^2$ & $X_3^2X_4^2X_5^0$ & $X_3^2X_4^2X_5^1$ & $X_3^2X_4^2X_5^2$ \\
             \hline
             Input & 1 & 0 & 0 & 0 & 1 & 0 & 0 & 1 & 1
        \end{tabular}}
        \caption{The inputs to the butterfly diagram, in Example~\ref{ex:butterfly}, for the function $F_3$, which are based on the minterms.}
        \label{tab:ex-butterfly-minterms}
    \end{table}

    \textit{The inputs represent the coefficients of the minterms, listed in natural order.} The minterms form the input to the butterfly, which is listed in Table~\ref{tab:ex-butterfly-minterms}, where a value of 1 is inputted if the corresponding minterm appears in the XOR of the minterms expression above, and a value of 0 if not. For instance, the number corresponding to the minterm $X_3^1X_4^0X_5^1$ is 1 because the minterm is in the expression, but the number corresponding to $X_3^1X_4^2X_5^1$ is 0 because the minterm is not in the expression.

    \textit{The butterfly diagram is formed from $n$ columns for each variable in descending order.} The first column corresponds to the polarity $Q_5$ for the variable $X_5$. The first column repeats the single-variable butterfly diagram for $Q_5$ for every three ($3^1$) rows (repeated nine times).

    \begin{figure}[!htb]
        \centering
        \includegraphics[width=0.25\linewidth]{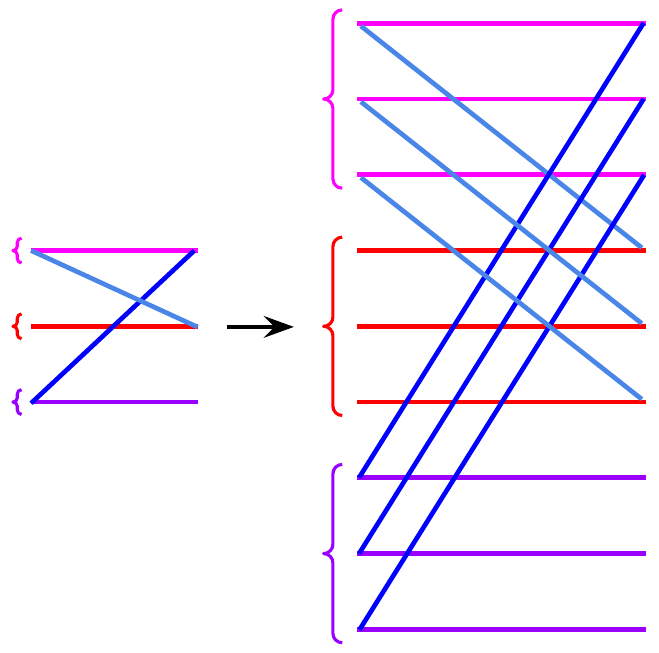}
        \caption{Process to create part of the column for $Q_4$ in Example~\ref{ex:butterfly}.}
        \label{fig:ex-butterfly-q4-process}
    \end{figure}

    The second column corresponds to the polarity $Q_4$ for the variable $X_4$, and repeats a butterfly diagram kernel for every nine ($3^2$) rows. This butterfly diagram kernel is derived from the single-variable butterfly diagram for the polarity $Q_4$. This part is altered to fit nine rows in the method demonstrated in Fig.~\ref{fig:ex-butterfly-q4-process}. This kernel is repeated three times.

    \begin{figure}[!htb]
        \centering
        \includegraphics[width=0.8\linewidth]{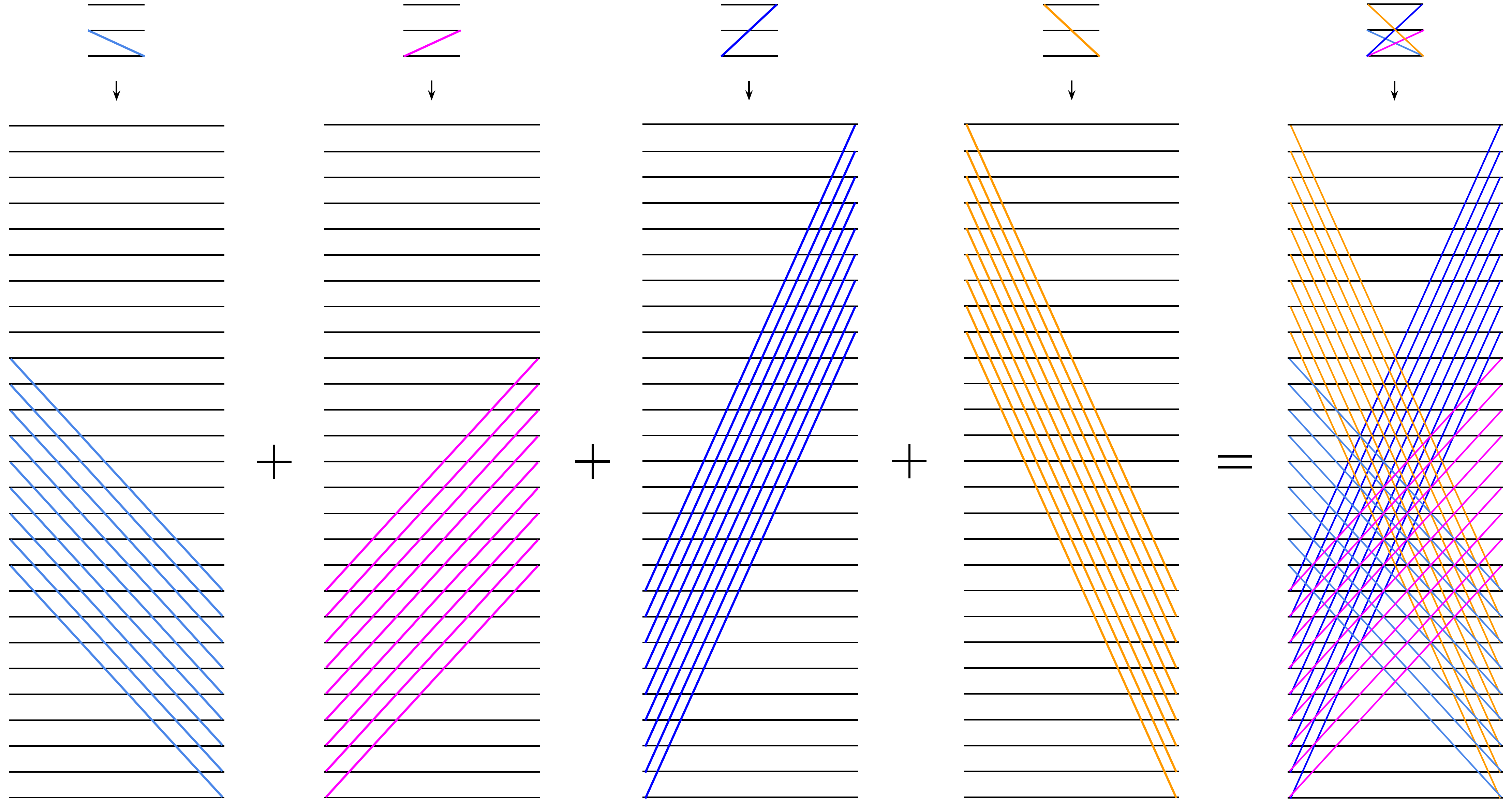}
        \caption{Process to create column for $Q_3$ in Example~\ref{ex:butterfly}.}
        \label{fig:ex-butterfly-q3}
    \end{figure}

    The third column corresponds to the polarity $Q_3$ for the variable $X_3$. The butterfly diagram for this column is altered from the single-variable diagram for $Q_3$, as illustrated in Fig.~\ref{fig:ex-butterfly-q3}.

    \begin{figure}[!htb]
    \centering
        \begin{subfigure}{0.34\linewidth}
            \centering
            \includegraphics[width=\linewidth]{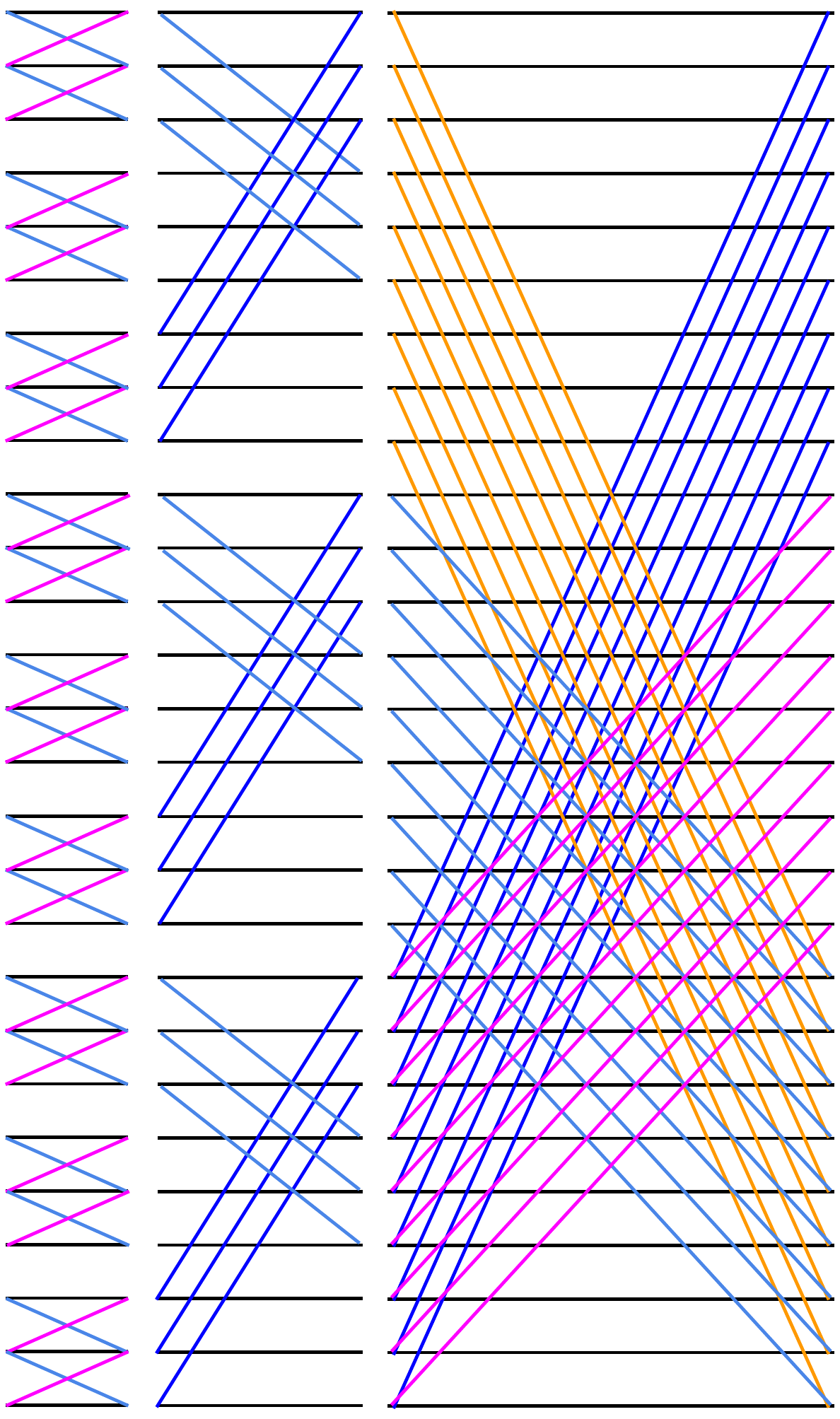}
            \caption{Diagram.}
            \label{fig:ex-full-butterfly-diagram}
        \end{subfigure}
        \hspace{0.05\linewidth}
        \begin{subfigure}{0.56\linewidth}
            \centering
            \includegraphics[width=\linewidth]{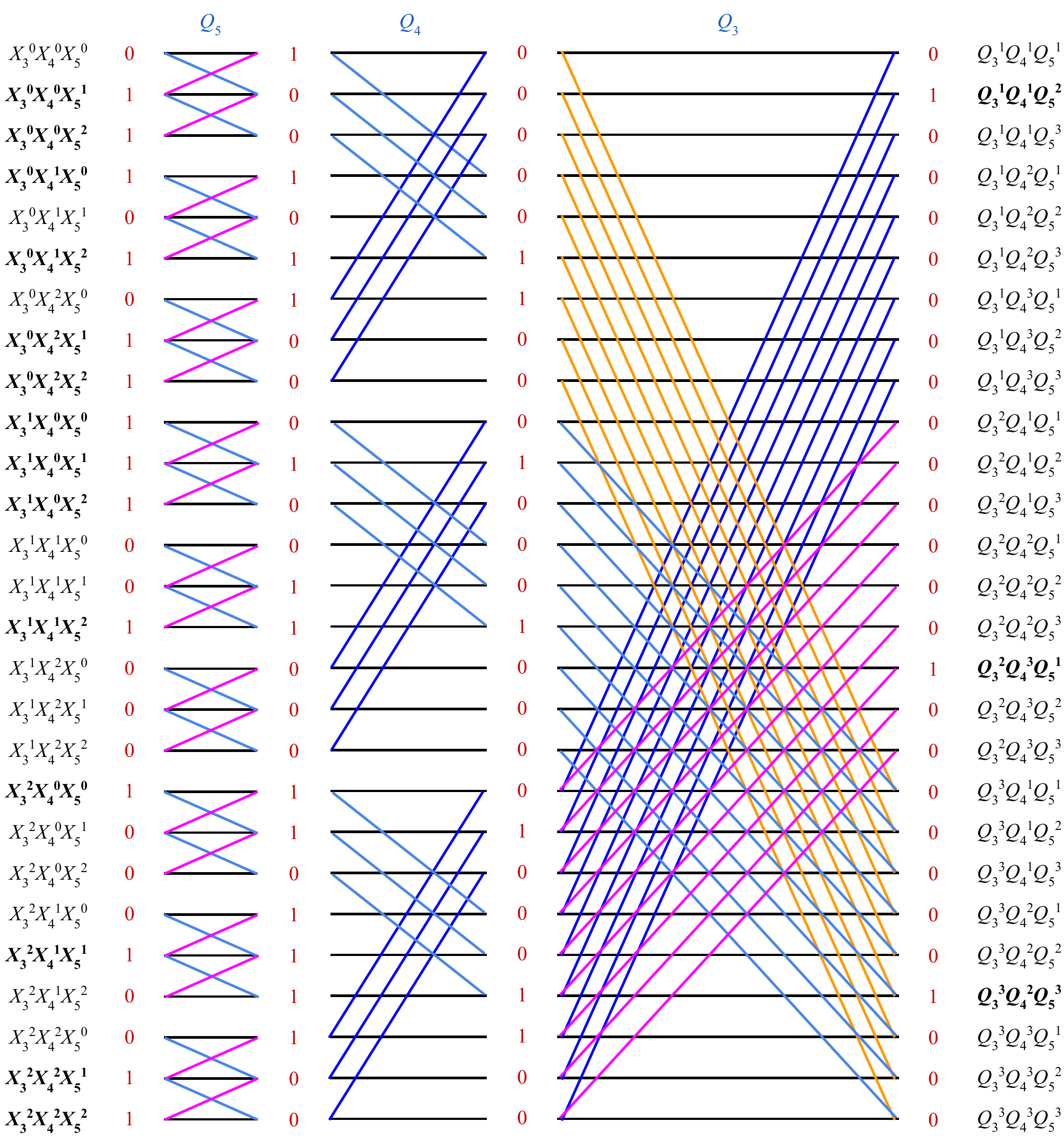}
            \caption{Inputs and outputs for $F_3$.}
            \label{fig:ex-butterfly}
        \end{subfigure}
        
        \caption{Full butterfly diagram for the polarity $Q_3$,$Q_4$,$Q_5$, from Example~\ref{ex:butterfly}.}
    \end{figure}

    Combining all three columns yields the butterfly diagram shown in Fig.~\ref{fig:ex-full-butterfly-diagram} for transforming any function of three ternary variables to the polarity $Q_3$,$Q_4$,$Q_5$.

    \textit{The outputs represent the spectral coefficients of the polarity terms, listed in natural order.} The input for $F_3$ in the butterfly diagram leads to the MVI-FPRM 
    \begin{align*}
        F_3
        &= Q_3^1Q_4^1Q_5^2 \oplus Q_3^2Q_4^3Q_5^1 \oplus Q_3^3Q_4^2Q_5^3 \\
        &= X_3^{1,2}X_4^{0,1} \oplus X_3^{0,2}X_5^{1,2} \oplus X_4^{1}X_5^{0,1}
    \end{align*}
    as illustrated in Fig.~\ref{fig:ex-butterfly}. This is the same MVI-FPRM as the one found in Example~\ref{ex:fprm three variables}.
    
\end{example}

\section{Circuit Synthesis Based on MVI-GRM}
\label{grm-circuit}

The method for MVI-GRM selects one of the top results from the MVI-FPRM method and utilizes factorization and transformation rules to determine the quasi-minimal MVI-GRM. In the multi-output MVI-FPRM, the same literals or even products are repeated for various outputs; therefore, they can be shared in the final circuit. However, in MVI-GRM, the method may create non-repeated literals and products that cannot be shared. This method is demonstrated in Example~\ref{ex:mvi-grm1}.

\begin{example}
\label{ex:mvi-grm1}
    The MVI-FPRM forms $F_1 = 1 \oplus X_2^2 \oplus X_1^{2,3} \oplus X_1^{2,3}X_2^2 \oplus X_1^{1,2,3} \oplus X_1^{1,2,3}X_2^2$ (from Example~\ref{ex:fprm}), $F_2 = 1 \oplus X_1^{2,3} \oplus X_1^{1,2,3} \oplus X_1^{2,3}X_2^2$ (from Example~\ref{ex:fprm two terms}) were selected in the first phase, where the variable $X_1$ is quaternary and $X_2$ is ternary. Both of these forms can be factorized to create MVI-GRM forms, which are realized as circuits in Fig.~\ref{fig:grm-circuit}. The circuit for $F_1$ has a Maslov cost of 13 and a TQC of 57, and the circuit for $F_2$ has the same Maslov and TQC cost.
    \small{
    \begin{align*}
        F_1 
        &= 1 \oplus X_2^2 \oplus X_1^{2,3} \oplus X_1^{2,3}X_2^2 \oplus X_1^{1,2,3} \oplus X_1^{1,2,3}X_2^2 \\
        &= (1 \oplus X_1^{2,3} \oplus X_1^{1,2,3})(1  \oplus X_2^2) \\
        &= X_1^{0,2,3}X_2^{0,1} \\
        F_2
        &= 1 \oplus X_1^{2,3} \oplus X_1^{1,2,3} \oplus X_1^{2,3}X_2^2 \\
        &= (1 \oplus X_1^{1,2,3}) \oplus X_1^{2,3}(1 \oplus X_2^2) \\
        &= X_1^0 \oplus X_1^{2,3}X_2^{0,1} 
    \end{align*}
    }

    \begin{figure}[!htbp]
        \centering
        \begin{subfigure}{0.25\linewidth}
            \includegraphics[width=\linewidth]{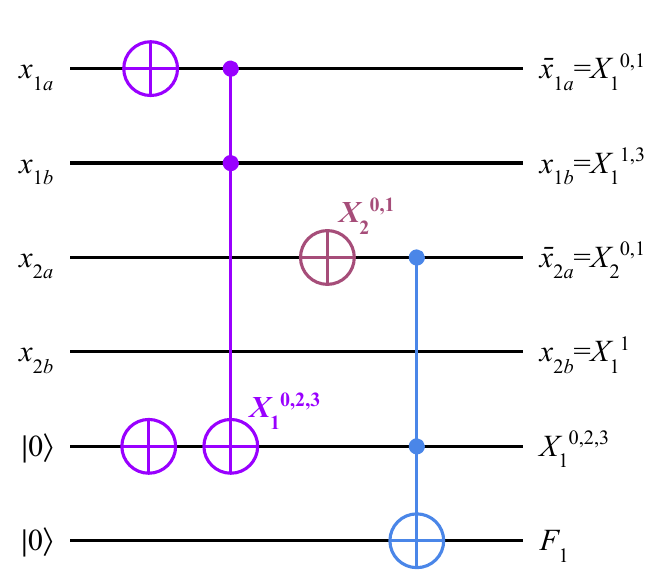}
            \caption{$F_1$.}
        \end{subfigure}
        \hspace{0.075\linewidth}
        \begin{subfigure}{0.31\linewidth}
            \includegraphics[width=\linewidth]{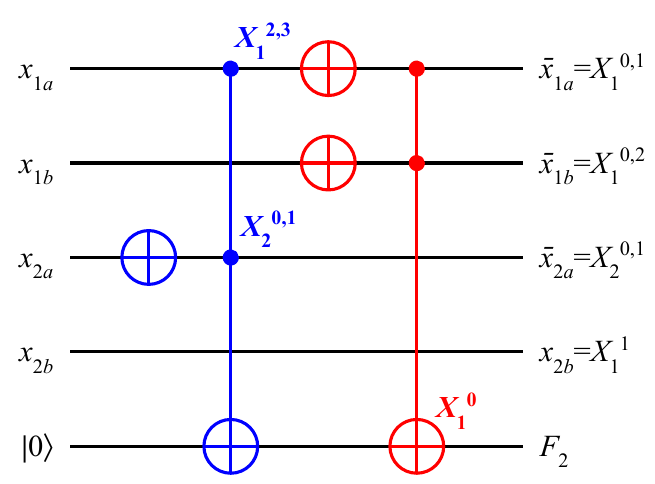}
            \caption{$F_2$.}
        \end{subfigure}
        \caption{The circuits for $F_1$ and $F_2$ based on MVI-GRM, from Example~\ref{ex:mvi-grm1}.}
        \label{fig:grm-circuit}
    \end{figure}

    If, instead, the functions $F_1=\bar{x}_{2a} \oplus \bar{x}_{1a}x_{1b}\bar{x}_{2a}$ and $F_2=\bar{x}_{1a}\bar{x}_{1b} \oplus x_{1a}\bar{x}_{2a}$ were realized as binary ESOP-based circuits, then the result would be the circuits shown in Fig.~\ref{fig:grm-circuit-esop}. The circuit for $F_1$ has a Maslov cost of 16 and a TQC of 125, and the circuit for $F_2$ is the same as the MVI-GRM-based circuit. Comparisons between the MVI-GRM-based circuits and ESOP-based circuits are summarized in Table~\ref{tab:comparison-grm}

    \begin{figure}[!htbp]
        \centering
        \begin{subfigure}{0.235\linewidth}
            \includegraphics[width=\linewidth]{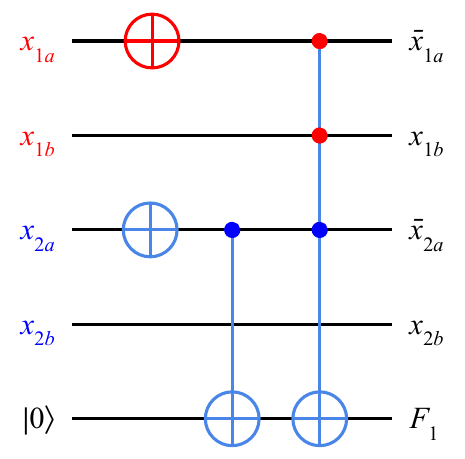}
            \caption{$F_1$.}
        \end{subfigure}
        \hspace{0.075\linewidth}
        \begin{subfigure}{0.25\linewidth}
            \includegraphics[width=\linewidth]{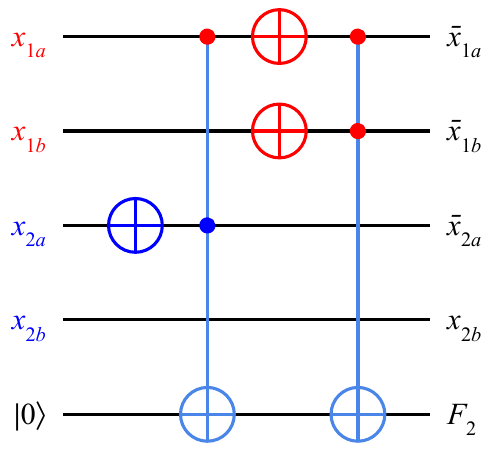}
            \caption{$F_2$.}
        \end{subfigure}
        \caption{The circuits for $F_1$ and $F_2$ based on ESOP, from Example~\ref{ex:mvi-grm1}.}
        \label{fig:grm-circuit-esop}
    \end{figure}

    \begin{table}[!htb]
        \centering
        \renewcommand{\arraystretch}{1.3}
        \resizebox{\linewidth}{!}{
        \begin{tabular}{c||cccc|cc||cccc|cc}
              & \multicolumn{6}{c||}{MVI-GRM-Based} & \multicolumn{6}{c}{ESOP-Based}  \\
             \cline{2-13}
             Function & NOT & CNOT & 3-bit Toffoli & 4-bit Toffoli & Maslov Cost & TQC & NOT & CNOT & 3-bit Toffoli & 4-bit Toffoli & Maslov Cost & TQC \\
             \hline
             \hline
             $F_1$ & 3 & 0 & 2 & 0 & 13 & 57 & 2 & 1 & 0 & 1 & 16 & 125 \\
             $F_2$ & 3 & 0 & 2 & 0 & 13 & 57 & 3 & 0 & 2 & 0 & 13 & 57
        \end{tabular}}
        \caption{GRM-based vs ESOP-based circuits for $F_1$ and $F_2$.}
        \label{tab:comparison-grm}
    \end{table}
\end{example}

\section{Conclusions}
\label{conclusions}
Two new types of binary quantum circuits are introduced in this paper, both based on multi-valued input, binary output (MVI) logic. We introduced the MVI-FPRM and MVI-GRM canonical forms. Two methods are presented for the synthesis of exact or approximate MVI-FPRM. The first method is based on products-matching, and the second method is based on our new MVI butterfly diagrams. MVI-GRM synthesis is based on the factorization of the selected-best MVI-FPRM with subsequent use of simplifying rules. In circuit synthesis based on MVI-FPRM, the same products are repeated in several outputs, so they can be reused, as achieved using polarity-based decoders for the MVI function, which aids the synthesis of multi-output functions. However, in MVI-GRM, non-repeated literals and products are created, which does not allow for sharing sub-functions. This is a typical trade-off with MVI-FPRM and MVI-GRM. Although this paper only deals with completely specified functions, it will be extended to incompletely specified functions in future papers. Our future work focuses on selecting variables for pairing and choosing the best decoders for a given Boolean function.

\bibliographystyle{ieeetr}
\bibliography{references}

@inproceedings{perkowski89,
    title = {Minimization of multiple-valued input multi-output mixed-radix exclusive sums of products for incompletely specified {Boolean} functions},
    author = {Perkowski, Marek and Helliwell, Martin and Wu, Pan},
    booktitle = {Proceedings 19th Int. Symp. on Multiple-Valued Logic},
    pages = {256-263},
    year = {1989},
    doi = {10.1109/ISMVL.1989.37793}
}

@inproceedings{schafer91,
    title = {Multiple-valued generalized {Reed}-{Muller} forms},
    year = {1991},
    author = {Sch{\"a}fer, Ingo and Perkowski, Marek A},
    booktitle = {Proceedings Int. Symp. on Multiple-Valued Logic},
    pages = {40-48},
    doi = {10.1109/ISMVL.1991.130703},
}

@article{song96,
  author={Song, N. and Perkowski, M.A.},
  journal={IEEE Transactions on Computer-Aided Design of Integrated Circuits and Systems}, 
  title={Minimization of exclusive sum-of-products expressions for multiple-valued input, incompletely specified functions}, 
  year={1996},
  volume={15},
  number={4},
  pages={385-395},
  doi={10.1109/43.494702}}

@article{maslov04,
  title={Reversible cascades with minimal garbage},
  author={Maslov, Dmitri and Dueck, Gerhard W},
  journal={IEEE Transactions on Computer-Aided Design of Integrated Circuits and Systems},
  volume={23},
  number={11},
  pages={1497--1509},
  year={2004},
  publisher={IEEE}
}

@article{jiang22,
  author={Jiang, Jie-Hong R. and De Micheli, Giovanni and Smith, Kaitlin and Soeken, Mathias},
  journal={IEEE Journal on Emerging and Selected Topics in Circuits and Systems}, 
  title={Design and Automation for Quantum Computation and Quantum Technologies}, 
  year={2022},
  volume={12},
  number={3},
  pages={581-583},
  doi={10.1109/JETCAS.2022.3207268}
}

@article{saeedi13,
    author = {Saeedi, Mehdi and Markov, Igor L.},
    title = {Synthesis and optimization of reversible circuits—a survey},
    year = {2013},
    issue_date = {February 2013},
    publisher = {Association for Computing Machinery},
    address = {New York, NY, USA},
    volume = {45},
    number = {2},
    issn = {0360-0300},
    url = {https://doi.org/10.1145/2431211.2431220},
    doi = {10.1145/2431211.2431220},
    journal = {ACM Comput. Surv.},
    month = mar,
    articleno = {21},
    numpages = {34}
}

@inproceedings{mishchenko01,
    title = {Fast heuristic minimization of exclusive-sums-of-products},
    author = {Mishchenko, Alan and Perkowski, Marek},
    booktitle = {5th International Workshop on Applications of the Reed Muller Expansion in Circuit Design},
    year = {2001}
}

@inproceedings{fazel07,
    title = {{ESOP}-based {Toffoli} gate cascade generation},
    author = {Fazel, Kenneth and Thornton, Mitchell A and Rice, Jacqueline E},
    booktitle = {IEEE Pacific Rim Conference on Communications, Computers and Signal Processing},
    pages = {206--209},
    year = {2007},
    organization = {IEEE}
}

@inproceedings{schmitt19,
  title = {Scaling-up {ESOP} synthesis for quantum compilation},
  author = {Schmitt, Bruno and Soeken, Mathias and De Micheli, Giovanni and Mishchenko, Alan},
  booktitle = {IEEE 49th International Symposium on Multiple-Valued Logic},
  pages = {13--18},
  year = {2019},
  organization = {IEEE}
}

@inproceedings{meuli18,
    title = {A best-fit mapping algorithm to facilitate {ESOP}-decomposition in {Clifford}+{T} quantum network synthesis},
    author = {Meuli, Giulia and Soeken, Mathias and Roetteler, Martin and Wiebe, Nathan and De Micheli, Giovanni},
    booktitle = {23rd Asia and South Pacific Design Automation Conference},
    pages = {664--669},
    year = {2018},
    organization = {IEEE}
}

@inproceedings{meuli19,
  title={Evaluating {ESOP} optimization methods in quantum compilation flows},
  author={Meuli, Giulia and Schmitt, Bruno and Ehlers, R{\"u}diger and Riener, Heinz and De Micheli, Giovanni},
  booktitle={International Conference on Reversible Computation},
  pages={191--206},
  year={2019},
  organization={Springer}
}

@inproceedings{rice09,
  title={{Toffoli} gate cascade generation using {ESOP} minimization and {QMDD}-based swapping},
  author={Rice, JE and Fazel, K and Thornton, M and Kent, KB},
  booktitle={Proceedings of the 2009 Reed-Muller Workshop},
  pages={63--72},
  year={2009}
}

@inproceedings{rice11,
    title = {Ordering techniques for {ESOP}-based {Toffoli} cascade generation},
    author = {Rice, JE and Nayeem, NM},
    booktitle = {Proceedings of 2011 IEEE Pacific Rim Conference on Communications, Computers and Signal Processing},
    pages = {274--279},
    year = {2011},
    organization = {IEEE}
}

@article{nayeem11,
  title={A shared-cube approach to {ESOP}-based synthesis of reversible logic},
  author={Nayeem, Noor M and Rice, Jacqueline E},
  journal={Facta universitatis-series: Electronics and Energetics},
  volume={24},
  number={3},
  pages={385--402},
  year={2011}
}

@inproceedings{sasao93,
  title={An exact minimization of {AND}-{EXOR} expressions using reduced covering functions},
  author={Sasao, Tsutomu},
  booktitle={Proceedings of the Synthesis and Simulation Meeting and International Interchange},
  pages={374--383},
  year={1993}
}

@article{csanky93,
  title={Canonical restricted mixed-polarity exclusive-{OR} sums of products and the efficient algorithm for their minimisation},
  author={Csanky, L and Perkowski, MA and Schaefer, I},
  journal={IEE Proceedings E (Computers and Digital Techniques)},
  volume={140},
  number={1},
  pages={69--77},
  year={1993},
  publisher={IET}
}

@inproceedings{kazimirov21,
  title={Genetic algorithm for minimization of {ESOP} representations for multiple-output logic functions},
  author={Kazimirov, A and Maleyev, V},
  booktitle={Journal of Physics: Conference Series},
  volume={1847},
  pages={012028},
  year={2021},
  organization={IOP Publishing}
}

@article{papakonstantinou17,
  title={Exclusive or sum of complex terms expressions minimization},
  author={Papakonstantinou, George},
  journal={Integration},
  volume={56},
  pages={44--52},
  year={2017},
  publisher={Elsevier}
}

@inproceedings{riener19,
  title={Exact synthesis of {ESOP} forms},
  author={Riener, Heinz and Ehlers, R{\"u}diger and Schmitt, Bruno de O and Micheli, Giovanni De},
  booktitle={Advanced Boolean Techniques: Selected Papers from the 13th International Workshop on Boolean Problems},
  pages={177--194},
  year={2019},
  organization={Springer}
}

@inproceedings{song97,
    author = {H. B. Song},
    title = {A Study on Minimization Algorithm for {ESOP} of Multiple-Valued Function},
    booktitle = {The Transactions of the Korea Information Processing},
    year = {1997}
}

@article{green91,
  title={Families of {Reed}-{Muller} canonical forms},
  author={Green, DH},
  journal={International Journal of Electronics Theoretical and Experimental},
  volume={70},
  number={2},
  pages={259--280},
  year={1991},
  publisher={Taylor \& Francis}
}

@inproceedings{drechsler95,
    title = {A genetic algorithm for minimization of fixed polarity {Reed}-{Muller} expressions},
    author = {Drechsler, Rolf and Becker, Bernd and G{\"o}ckel, Nicole},
    booktitle = {Artificial Neural Nets and Genetic Algorithms: Proceedings of the International Conference in Al{\`e}s, France, 1995},
    pages = {393--395},
    year = {1995},
    organization = {Springer}
}

@inproceedings{sarabi92,
    title = {Fast exact and quasi-minimal minimization of highly testable fixed-polarity {AND}/{XOR} canonical networks},
    author = {Sarabi, A and Perkowski, MA},
    booktitle = {Proceedings of 29th ACM/IEEE Design Automation Conference},
    pages = {30--35},
    year = {1992},
    organization = {IEEE}
}

@article{debnath96,
    title = {{GRMIN2}: A heuristic simplification algorithm for generalised {Reed}-{Muller} expressions},
    author = {D. Debnath and T. Sasao}, 
    journal = {IEE Proceedings-Computers and Digital Techniques},
    volume={143},
    number={6},
    pages={376--384},
    year={1996},
    publisher={IET}
}

@article{dill97,
    title = {Minimization of Generalized {Reed}-{Muller} Forms with a Genetic Algorithm},
    author = {Dill, KM and Perkowski, MA},
    journal = {Proceedings of Genetic Programming},
    volume = {97},
    year = {1997}
}

@article{dill01,
    title = {Baldwinian learning utilizing genetic and heuristic algorithms for logic synthesis and minimization of incompletely specified data with Generalized {Reed}-{Muller} ({AND}-{EXOR}) forms},
    author = {Dill, Karen M and Perkowski, Marek A},
    journal = {Journal of Systems Architecture},
    volume = {47},
    number = {6},
    pages = {477--489},
    year = {2001},
    publisher = {Elsevier}
}

@inproceedings{helliwell88,
  title={A fast algorithm to minimize multi-output mixed-polarity generalized {Reed}-{Muller} forms},
  author={Helliwell, Martin and Perkowski, Marek},
  booktitle={Proceedings of the 25th ACM/IEEE Design automation conference},
  pages={427--432},
  year={1988}
}

@article{sasao84,
    author = {Sasao, Tsutomu},
    journal = {IEEE Transactions on Computers}, 
    title = {Input Variable Assignment and Output Phase Optimization of {PLA}'s}, 
    year = {1984},
    volume = {C-33},
    number = {10},
    pages = {879-894},
    doi={10.1109/TC.1984.1676349}
}

@article{al-bayaty23gala,
    title = {{GALA}-n: Generic architecture of layout-aware n-bit quantum operators for cost-effective realization on {IBM} quantum computers},
    author = {Al-Bayaty, Ali and Perkowski, Marek},
    journal = {arXiv e-prints},
    eprint = {2311.06760},
    archivePrefix = {arXiv},
    primaryClass = {quant-ph},
    year = {2023},
    doi = {10.48550/arXiv.2311.06760}
}

@article{al-bayaty24cala,
    author = {{Al-Bayaty}, Ali and {Song}, Xiaoyu and {Perkowski}, Marek},
    title = "{{CALA}-n: A Quantum Library for Realizing Cost-Effective 2-, 3-, 4-, and 5-bit Gates on {IBM} Quantum Computers using {Bloch} Sphere Approach, {Clifford}+{T} Gates, and Layouts}",
    journal = {arXiv e-prints},
    year = {2024},
    doi = {10.48550/arXiv.2408.01025},
    archivePrefix = {arXiv},
    eprint = {2408.01025},
    primaryClass = {quant-ph},
}

@article{maslov03,
    title = {Improved quantum cost for n-bit {Toffoli} gates},
    author = {Maslov, Dmitri and Dueck, Gerhard W},
    journal = {Electronics Letters},
    volume = {39},
    number = {25},
    pages = {1790-1791},
    year = {2003},
    doi = {10.1049/el:20031202}
}

@book{roman05,
  title={Advanced linear algebra},
  author={Roman, Steven},
  year={2005},
  publisher={Springer}
}

@article{jin20,
  title={A Polarity-based approach for optimization of multivalued quantum multiplexers with arbitrary single-qubit target gates},
  author={Jin, Kevin and Soffat, Tahsin and Morgan, Justin and Perkowski, Marek},
  journal={Journal of Applied Logics--IfCoLog Journal of Logics and their Applications},
  volume={7},
  number={1},
  year={2020}
}

@techreport{weinstein69,
    title = {Quantization effects in digital filters},
    author = {Weinstein, Clifford J},
    year = {1969},
    institution = {MIT Lincoln Laboratory}
}

@article{shanks69,
  title={Computation of the fast {Walsh}-{Fourier} transform},
  author={Shanks, John L},
  journal={IEEE Transactions on Computers},
  volume={100},
  number={5},
  pages={457--459},
  year={1969},
  publisher={IEEE}
}

@inproceedings{lee16,
    author = {Lee, Bryan and Perkowski, Marek},
    booktitle = {Euromicro Conference on Digital System Design (DSD)}, 
    title = {Quantum machine learning based on minimizing {Kronecker}-{Reed}-{Muller} forms and {Grover} search algorithm with hybrid oracles}, 
    year = {2016},
    pages = {413-422}
}

@article{falkowski06,
  title={Matrix decomposition and butterfly diagrams for mutual relations between Hadamard-Haar and arithmetic spectra},
  author={Falkowski, Bogdan J and Yan, Shixing},
  journal={IEEE Transactions on Circuits and Systems I: Regular Papers},
  volume={53},
  number={5},
  pages={1119--1129},
  year={2006},
  publisher={IEEE}
}

@inproceedings{vahid20,
  title={Butterfly transform: An efficient fft based neural architecture design},
  author={Vahid, Keivan Alizadeh and Prabhu, Anish and Farhadi, Ali and Rastegari, Mohammad},
  booktitle={2020 IEEE/CVF conference on computer vision and pattern recognition (CVPR)},
  pages={12021--12030},
  year={2020},
  organization={IEEE}
}

@inproceedings{li06,
  title={A Quantum {CAD} {Accelerator} based on {Grover}'s algorithm for finding the minimum {Fixed} {Polarity} {Reed}-{Muller} form},
  author={Li, Lun and Thornton, Mitch and Perkowski, Marek},
  booktitle={36th International Symposium on Multiple-Valued Logic (ISMVL'06)},
  pages={33--33},
  year={2006},
  organization={IEEE}
}

@inproceedings{grover96,
  title={A fast quantum mechanical algorithm for database search},
  author={Grover, Lov K},
  booktitle={Proceedings of the twenty-eighth annual ACM symposium on Theory of Computing},
  pages={212--219},
  year={1996}
}

@article{al-bayaty24bht-qaoa,
    author = {Al-Bayaty, Ali and Perkowski, Marek},
    title = {{BHT}-{QAOA}: The Generalization of Quantum Approximate Optimization Algorithm to Solve Arbitrary Boolean Problems as Hamiltonians},
    journal = {Entropy},
    volume = {26},
    year = {2024},
    number = {10},
    article-number = {843},
    url = {https://www.mdpi.com/1099-4300/26/10/843},
    doi = {10.3390/e26100843}
}

@book{alasow24,
    title = {Quantum search algorithms for constraint satisfaction and optimization problems using {Grover}'s search and quantum walk algorithms With advanced oracle design},
    author = {Alasow, Abdirahman Sheikh Hassan},
    year = {2024},
    publisher = {Portland State University}
}

@article{ilyas22,
    doi = {10.1088/1751-8121/ac7b55},
    url = {https://doi.org/10.1088/1751-8121/ac7b55},
    year = {2022},
    month = {jul},
    publisher = {IOP Publishing},
    volume = {55},
    number = {30},
    pages = {305302},
    author = {Ilyas, Muhammad and Cui, Shawn and Perkowski, Marek},
    title = {Ternary logic design in topological quantum computing},
    journal = {Journal of Physics A: Mathematical and Theoretical},
}

@book{wakerly99,
    title = {Digital Design: Principles and Practices},
    author = {Wakerly, J.F.},
    isbn = {9780130825995},
    lccn = {99036681},
    series = {Prentice Hall Xilinx design series},
    year = {1999},
    publisher = {Prentice Hall}
}

@book{davio,
    author = {Davio, Marc and Deschamps, Jean-Pierre and Thayse, André},
    title = {Discrete and Switching Functions},
    publisher = {McGraw-Hill},
    year = {1978}
}

@article{shannon,
    title = {The synthesis of two-terminal switching circuits},
    author = {Shannon, Claude E},
    journal = {The Bell System Technical Journal},
    volume = {28},
    number = {1},
    pages = {59--98},
    year = {1949},
}

@article{ibm,
    title = {IBM quantum computers: evolution, performance, and future directions},
    author = {AbuGhanem, Muhammad},
    journal = {The Journal of Supercomputing},
    volume = {81},
    pages = {687},
    year = {2025}
}

@article{rudiger95,
  title={A solar dynamo in the overshoot layer: cycle period and butterfly diagram.},
  author={R{\"u}diger, G{\"u}nther and Brandenburg, Axel},
  journal={Astronomy and Astrophysics, v. 296, p. 557},
  volume={296},
  pages={557},
  year={1995}
}

@article{qiskit,
  title={Quantum computing with Qiskit},
  author={Javadi-Abhari, Ali and Treinish, Matthew and Krsulich, Kevin and Wood, Christopher J and Lishman, Jake and Gacon, Julien and Martiel, Simon and Nation, Paul D and Bishop, Lev S and Cross, Andrew W and others},
  journal={arXiv preprint arXiv:2405.08810},
  year={2024}
}

\end{document}